\documentclass[a4paper,12pt]{article}
\usepackage{geometry}
\usepackage{setspace}

\usepackage{enumitem}
\usepackage{amsmath,amssymb,amsthm}
\usepackage{mathtools}

\newtheorem{theorem}{Theorem}
\newtheorem{lemma}{Lemma}
\newtheorem{proposition}{Proposition}
\newtheorem{corollary}{Corollary}
\theoremstyle{definition}
\newtheorem{assumption}{Assumption}

\numberwithin{subassumption}{assumption}

\newtheorem{remark}{Remark}

\allowdisplaybreaks

\newcommand{\E}{E}
\newcommand{\I}{I}
\renewcommand{\d}{\mathrm{d}}
\newcommand\T{{ \mathrm{\scriptscriptstyle T} }}
\renewcommand{\Pr}{\mathrm{pr}}

\newcommand{\Alpha}{\mathrm{A}}
\newcommand{\indep}{\mathrel{\reflectbox{\rotatebox[origin=c]{90}{$\models$}}}}
\makeatletter
\newcommand*\dott{\mathpalette\dott@{.75}}
\newcommand*\dott@[2]{\mathbin{\vcenter{\hbox{\scalebox{#2}{$\m@th#1\bullet$}}}}}
\makeatother

\usepackage{booktabs}
\usepackage{multirow}
\usepackage{array}

\usepackage{tikz}
\usetikzlibrary{arrows.meta,automata,calc,positioning,shapes.arrows,shapes.geometric,shapes.multipart,decorations.pathreplacing,decorations.pathmorphing,positioning,swigs}
\tikzset{vertex/.style={inner sep=0pt,minimum size=1em},
         sq/.style={draw,rectangle,inner sep=0pt,minimum size=1em},
         hidden/.style={draw,shape=circle,preaction={fill=gray!50,fill opacity=0.5},inner sep=1pt,minimum size=1em},
         swig vsplit={gap=3pt,line color right=red},
         ell/.style={draw,inner sep=0pt,shape=ellipse}}
\tikzset{snake it/.style={decorate, decoration={snake, amplitude=.4mm,segment length=1.5mm,post length=1mm,pre length=1mm}}}

\usepackage{natbib}

\usepackage{multibib}
\newcites{suppmat}{References}
\bibliographystylesuppmat{apalike}

\usepackage{times}
\usepackage{newtxmath}
\usepackage[cal=cm,bb=ams]{mathalpha}

\usepackage{authblk}

\usepackage{titlesec}
\titleformat*{\section}{\bfseries}
\titleformat*{\subsection}{\bfseries}
\titleformat*{\paragraph}{\bfseries}

\usepackage{minitoc}
\noptcrule

\usepackage{etoolbox}

\makeatletter
\patchcmd{\@maketitle}{\large}{\normalsize}{}{}
\patchcmd{\@maketitle}{\LARGE}{\large}{}{}
\makeatother

\title{Improving precision of cumulative incidence estimates in randomized controlled trials with external controls}
\date{}

\author[1]{Zehao Su\footnote{Corresponding author.}}
\author[1]{Helene C. W. Rytgaard}
\author[2]{Henrik Ravn}
\author[1]{Frank Eriksson}
\affil[1]{Section of Biostatistics, University of Copenhagen, Copenhagen, Denmark}
\affil[2]{Novo Nordisk A/S, S\o{}borg, Denmark}

\begin{document}

\maketitle

\begin{abstract}
Augmenting the control arm in clinical trials with external data can improve statistical power for demonstrating treatment effects.
In many time-to-event outcome trials, participants are subject to truncation by death.
Direct application of methods for competing risks analysis on the joint data may introduce bias, for example, due to covariate shifts between the populations.
In this work, we consider transportability of the conditional cause-specific hazard of the event of interest under the control treatment.
Under this assumption, we derive semiparametric efficiency bounds of causal cumulative incidences.
This allows for quantification of the theoretical efficiency gain from incorporating the external controls.
We propose triply robust estimators that can achieve the efficiency bounds, where the trial controls and external controls are made comparable through time-specific weights in a martingale integral.
We conducted a simulation study to show the precision gain of the proposed fusion estimators compared to their counterparts without utilizing external controls.
As a real data application, we used two cardiovascular outcome trials conducted to assess the safety of glucagon-like peptide-1 agonists.
Incorporating the external controls from one trial into the other, we observed a decrease in the standard error of the treatment effects on adverse non-fatal cardiovascular events with all-cause death as the competing risk.

\end{abstract}

\section{Introduction}

Randomized control trials (RCTs) are the gold standard for evaluation of new treatments.
Nonetheless, demonstrating the expected efficacy may require a substantial sample size, thereby requiring long duration of trials and driving up overall costs.
Motivated by these issues, recent years have seen a growing interest in the use of historical data in clinical trials.
In rare-disease trials, where it may be impractical and unethical to randomize a patient to the standard of care, regulatory bodies have discussed the feasibility of replacing trial controls with an external control arm \citep{fda2023considerations,ema2023reflection}.
Another example is hybrid control designs, in which the control arm in a clinical trial is augmented with external controls.
The external controls should match the characteristics of the trial controls to avoid introducing bias, and their transportability should be carefully assessed in the planning phase of trials.

Leveraging external controls in clinical trials is an instance of data fusion.
Despite the ubiquity of time-to-event outcomes in clinical trials, current literature on data fusion in causal inference mostly deals with continuous or binary outcomes.
In this work, we consider external control augmentation for the estimation of treatment effects on the time-to-event, where an individual is subject to multiple modes of failure.
Specifically, we wish to make inference on cumulative incidence functions defined on the counterfactual event time.

In the estimand framework, many transportability studies in survival analysis estimate the risk difference from at-risk indicators at predetermined timepoints \citep{ramagopalan2022transportability,zuo2022transportability,dang2023case}.
Unless the censoring rate is ignorable, risk estimators constructed from dichotomized event times suffer from censoring bias.
\citet{lee2022doubly} and \citet{cao2024transporting} provide a more formal treatment of the problem in generalizing treatment effects from a clinical trial to its superpopulation.
They propose estimators for the target population counterfactual survival curve assuming transportability of the survival time distribution after conditioning on relevant baseline covariates.
However, if the data contains competing events, their identification formula directly corresponds to the all-cause survival function, rather than the estimands desired here.
Moreover, in our application, we observe the outcome for both trial participants and external controls, hence requiring separate estimation strategies \citep{colnet2024causal}.

To accommodate competing risks, we work under the assumption of transportability of the cause-specific hazards, which are natural objects of interest in multi-state models.
Although other assumptions can be postulated, they are either unnecessarily strong, such as transportability of the joint distribution of the event time and type, or lacking of interpretability in the data generating process, such as transportability of the subdistribution function \citep{fine1999proportional}.
We construct semiparametrically efficient estimators by studying the nonparametric efficient influence functions of the parameters.
The resulting estimators show robustness against model misspecification different from existing nonparametric estimators for cumulative incidence functions without data fusion \citep{rytgaard2023estimation}.

In the absence of competing risks, a related line of work extends dynamic borrowing methods to survival analysis \citep{kwiatkowski2024case,tan2022augmenting,li2022conditional,sengupta2023emulating}.
These methods control the extent to which external controls are incorporated into the target population by modifying the data likelihood.
They estimate the hazard ratio between the active arm and the control arm, which has been criticized for lacking causal interpretation.
In contrast, we directly assume hazard transportability and consider robust estimators for marginal causal parameters with efficiency gain.

\section{Identifiability of causal cumulative incidence difference}

Without loss of generality, we consider two types of events: the event of interest (\(J=1\)) and the competing event (\(J=2\)).
For the underlying event time \(T\) and event type \(J\) censored by the censoring time \(C\), we observe the right-censored versions \(\tilde{T}=T\wedge C\) and \(\tilde{J}=\I(T\leq C)J\) with a maximum observation period of \((0,\tau]\).
Data is collected independently from two populations: the target population \((D=1)\) and the source population \((D=0)\).
Besides the outcome tuple \((\tilde{T},\tilde{J})\), a set of baseline covariates \(X\) is also observed in both populations.
In our application, the target population is the study population of an RCT with both an active treatment (\(A=1\)) and a control treatment (\(A=0\)), while the source population contributes only controls.
The supports of the baseline covariates in the RCT population and in the external control population are denoted by \(\mathcal{X}_{1}\) and \(\mathcal{X}_{0}\), respectively.

The observed data is sampled in a non-nested fashion, where random sampling is performed separately within the target population and the external control population \citep{dahabreh2021study}.
More concretely, we have a probability sample \((\tilde{T}_{i},\tilde{J}_{i},A_{i},X_{i})\) from the target population for \(i=1,2,\dots,n_{1}\) and another probability sample \((\tilde{T}_{i},\tilde{J}_{i},X_{i})\) from the external control population for \(i=n_{1}+1,n_{1}+2,\dots,n_{1}+n_{0}\).
The total sample size is denoted by \(n\).
For the asymptotic arguments that appear later, we need the following condition on the sampling scheme.
\begin{assumption}[Stable sampling probability]
  As \(n\to\infty\), \(n_{1}/n\to\alpha\in(0,1)\).
\end{assumption}
When the sample size \(n\) is large, we may view the joint sample as a random sample from some superpopulation distribution of \(O=(\tilde{T},\tilde{J},DA,X,D)\) such that \(\Pr(D=1)=\alpha\).

We are interested in the causal \(\tau\)-time cumulative incidence difference in the target population for both event types.
Under a specific treatment, the causal \(\tau\)-time cumulative incidence is defined as the average probability of having an event by time \(\tau\), had all subjects in the target population received that treatment.
Let the potential outcomes \(\{T(a),J(a)\}\) denote time to event and event type under the static intervention \(a=0,1\).
The population-level target parameters defined before can be represented by
\[
  \theta_{j}(a)=\Pr\{T(a)\leq \tau,J(a)=j \mid D=1\},\quad \theta_j=\theta_j(1)-\theta_j(0),
\]
for event type \(j\in\{1,2\}\).



\begin{assumption}[Causal assumptions]
  \label{asn:causal}
  \hfill
  \begin{enumerate}[label=(\roman*),nosep]
  \item (Consistency) \(T_i(a)=T_i\) and \(J_{i}(a)=J_{i}\) if \(A_{i}=a\) for \(a\in\{0,1\}\);
  \item (Randomization) \(\{T(a),J(a)\}\indep A\mid (X,D=1)\) and \(\Pr(A=a\mid X,D=1)>0\) for \(a\in\{0,1\}\).
  \end{enumerate}
\end{assumption}

With Assumption~\ref{asn:causal}, the target parameters defined on the counterfactual data distribution are identifiable from the uncensored data distribution.
Let \(F_{1j}(t\mid a,x)=\Pr(T\leq t,J=j\mid A=a,X=x,D=1)\) and \(F_{0j}(t\mid x)=\Pr(T\leq t,J=j\mid X=x,D=0)\) be the conditional cumulative incidence functions.
The causal \(\tau\)-time cause \(j\) cumulative incidence under the intervention \(a\) is identified by the g-formula
\[
  \theta_{j}(a)=\E\{F_{1j}(\tau\mid a,X)\mid D=1\}.
\]
To identify the parameter \(\theta_{j}(a)\) with the observed data, some conditions on the censoring time are needed.
Denote the survival functions of the all-cause event time by
\(S_{1}(t\mid a,x)= \Pr(T>t\mid A=a,X=x,D=1)\) and \(S_{0}(t\mid x)= \Pr(T>t\mid X=x,D=0)\), 
and denote the survival functions of the censoring time by \(S^{c}_1(t\mid a,x)= \Pr(C>t\mid A=a,X=x,D=1)\) and \(S^{c}_0(t\mid x)= \Pr(C>t\mid X=x,D=0)\).
\begin{assumption}[Censoring]
  \label{asn:censoring}
  \hfill
  \begin{enumerate}[label=(\roman*),nosep]
  \item (Positivity of censoring time) For all \(t\in(0,\tau]\),
  \begin{equation*}
  \begin{aligned}
    &S_1(t\mid a,x)>0 \Rightarrow S^{c}_1(t\mid a,x)>0,&\quad &\text{for }a\in\{0,1\}, x\in\mathcal{X}_{1};\\
    &S_0(t\mid x)>0 \Rightarrow S^c_0(t\mid x)>0,&\quad &\text{for }x\in\mathcal{X}_{1}\cap\mathcal{X}_{0}.
    \end{aligned}
  \end{equation*}
  \item\label{eqn:independent-censoring} (Independent censoring) \((T,J)\indep C\mid (A,X,D=1)\); \((T,J)\indep C\mid (X,D=0)\).
  \end{enumerate}
\end{assumption}
Under Assumption~\ref{asn:censoring}, the observed data likelihood at the realization \(o=(t,j,a,x,d)\) of \(O\sim P\) factorizes as
\begin{multline*}
  \d P(o)= \d P(x)\{\pi(x)e_1(a\mid x)\}^d\{1-\pi(x)\}^{(1-d)}\\
  \big[\{\d\Alpha_{1j}(t\mid a,x)\}^{I(j\neq 0)}S_{1}(t\!-\!\mid a,x)\big]^d\big[\{\d\Alpha_{0j}(t\mid x)\}^{I(j\neq 0)}S_{0}(t\!-\!\mid x)\big]^{(1-d)}\\
  \big(\big[\d\Alpha_{1}^c(t\mid a,x)\{1-\triangle\Alpha_{11}(t\mid a,x)-\triangle\Alpha_{12}(t\mid a,x)\}\big]^{I(j=0)}S_{1}^c(t\!-\!\mid a,x)\big)^d\\
  \big(\big[\d\Alpha_{0}^c(t\mid x)\{1-\triangle\Alpha_{01}(t\mid x)-\triangle\Alpha_{02}(t\mid x)\}\big]^{I(j=0)}S_{0}^c(t\!-\!\mid x)\big)^{(1-d)}
\end{multline*}
where \(\pi(x)= P(D=1\mid X=x)\) is the target population selection score, \(e_1(a\mid x)= P(A=a\mid X=x,D=1)\) is the treatment propensity score in the RCT,
and the infinitesimal increment of the conditional cumulative hazards of the events and censoring are
\begin{align*}
  \d\Alpha_{1j}(t\mid a,x) &= \frac{\d F_{1j}(t\mid a,x)}{S_1(t\!-\!\mid a,x)}, & \d\Alpha_{1}^c(t\mid a,x)&=-\frac{\d S_1^c(t\mid a,x)}{S_1^c(t\!-\!\mid a,x)}, \\
  \d\Alpha_{0j}(t\mid x) &= \frac{\d F_{0j}(t\mid x)}{S_0(t\!-\!\mid x)}, & \d\Alpha_{0}^c(t\mid x)&=-\frac{\d S_0^c(t\mid x)}{S_0^c(t\!-\!\mid x)}.
\end{align*}
Given Assumptions~\ref{asn:causal}--\ref{asn:censoring}, the parameter \(\theta_{j}(a)\) can be identified as a functional of the observed data distribution.
In Supplementary Material \S\ref{sec:notation-app}, we relate the quantities defined in the observed data distribution to those defined in the uncensored data distribution.

\section{Semiparametric theory for cumulative incidence}

\subsection{Transportability of the cause-specific hazard of the event of interest}
\label{sec:cif-hazard}

We propose a key transportability assumption under which the RCT controls and the external controls are compatible.
In data fusion, we adjust for prognostic variables with shifted distribution between the target population and the source population, so that conditional on these variables, the intervened populations are comparable in a certain respect.
The baseline covariates \(X\) are obviously sufficient for this purpose if
\begin{equation}
  \{N_1(0)(\tau),N_2(0)(\tau)\} \indep D \mid X.\label{eq:distribution-transportability-discrete}
\end{equation}
This strong condition states that the entire event processes under the control treatment become interchangeable between the populations, once the baseline covariates are controlled for.
We will discuss an example where \ref{eq:distribution-transportability-discrete} is violated, but a weaker transportability assumption sufficient for our purpose is fulfilled.

We motivate the assumption in the simplified case where the time to event is observed on a discrete grid.
Let \(\triangle N_j(0)(t)=\I\{T(0)=t,J(0)=j\}\) be the counterfactual indicator for an event of type \(j\) occurring at time \(t\in\{1,\dots,\tau\}\) under the control treatment and let \(N_j(0)(t)=\sum_{s=1}^{t}\triangle N_j(0)(s)\).
In addition to the variables in the previous section, we introduce shifted, unobserved prognostic variables \(U\), whose existence may violate \eqref{eq:distribution-transportability-discrete}, because \(D\) and \(\triangle N_1(0)(t)\) cannot be d-separated without blocking \(U\).
Consider a time-discretized data generating process encoded by the single-world intervention graph \citep{richardson2013single} displayed in Figure~\ref{fig:dag}.
At any timepoint, the variables \(U\) directly affect the competing event \(\triangle N_2(0)(t)\) but act only indirectly on the event of interest \(\triangle N_1(0)(t)\) through the history of the events \(\{N_1(0)(t-1), N_2(0)(t-1)\}\).
In this case, it holds that
\begin{equation}
  \triangle N_1(0)(t) \indep D \mid \{N_{1}(0)(t-1), N_2(0)(t-1),X\}\label{eq:hazard-transportability-discrete}
\end{equation}
for the event of interest without conditioning on the unobserved \(U\).
By definition,
\begin{multline*}
  \Pr\{\triangle N_1(0)(t)=1\mid N_{1}(0)(t-1)=0,N_2(0)(t-1)=0,X,D\}\\
  =\Pr\{T(0)=t, J(0)=1\mid T(0)\geq t,X,D\},
\end{multline*}
so the conditional independence \eqref{eq:hazard-transportability-discrete} is equivalent to transportability of the cause-specific hazard of the event of interest.

In continuous time, an analogous formulation to \eqref{eq:hazard-transportability-discrete} is the following.

\begin{assumption}[Transportability of conditional cause 1 hazard]
  \label{asn:transportability-hazard}
  \(\Alpha_{11}(0)(t\mid x)=\Alpha_{01}(0)(t\mid x)\) for \(x\in\mathcal{X}_{1}\cap\mathcal{X}_{0}\).
\end{assumption}

The interpretation is that for two subjects with the same baseline covariates, one in the RCT and one in the external population, given they have not experienced any event, the probability with which they immediately experience the event of interest are the same.
Unobserved variables like \(U\) can also be time-varying, as long as they have no direct effect on the event of interest.

\tikzset{v/.style={inner sep=0pt,outer sep=2pt},
  hidden/.style={draw,shape=circle,preaction={fill=gray!50,fill opacity=0.5},inner sep=1pt,minimum size=1em}}
\begin{figure}
  \centering
  \scalebox{0.8}{
  \begin{tikzpicture}
    \node[v] (D) at (0,7) { \(D\)};
    \node[v] (U) at (2.5,7) { \(U\)};
    \node[v] (X) at (0,5) { \(X\)};
    \node[v] (AA) at (0,3) { \(A\)};
    \node[v] (A) at (.7,3) { \(0\)};
    \draw (.35,2.7)--(.35,3.3);
    \node[v] (dN11) at (2.5,3) { \(\triangle N_1(0)(1)\)};
    \node[v] (dN21) at (2.5,5) { \(\triangle N_{2}(0)(1)\)};
    \node at (4.5,5) { \(\cdots\)};
    \node[v] (dN1t) at (10,3) { \(\triangle N_{1}(0)(t)\)};
    \node[v] (dN2t) at (10,5) { \(\triangle N_{2}(0)(t)\)};
    \path[-{Latex[length=1ex,width=.5ex]}]
    (D) edge (X)
    (D) edge (U)
    (D) edge[bend right] (AA)
    (X) edge (AA)
    (X) edge (dN11)
    (X) edge (dN21)
    (X) edge[out=45,in=180] (3.5,6)
    (U) edge (dN21)
    (U) edge (3.5,7)
    (A) edge (dN11)
    (A) edge (dN21)
    (A) edge[out=-45,in=180] (3.5,2)
    (dN11) edge (dN21)
    (dN11) edge (4,3)
    (dN21) edge (4,5)
    (dN1t) edge (dN2t)
    (5,3) edge node[midway,above]{\(N_{1}(0)(t-1),N_{2}(0)(t-1)\)} node[midway,below]{\(X\)} (dN1t)
    (5,5) edge node[midway,above]{\(N_{1}(0)(t),N_{2}(0)(t-1)\)} node[midway,below]{\(X,U\)} (dN2t);
  \end{tikzpicture}}
  \caption{Discrete-time single-world intervention graph of a data generating process satisfying Assumptions~\ref{asn:causal} and \ref{asn:transportability-hazard}.}
  \label{fig:dag}
\end{figure}
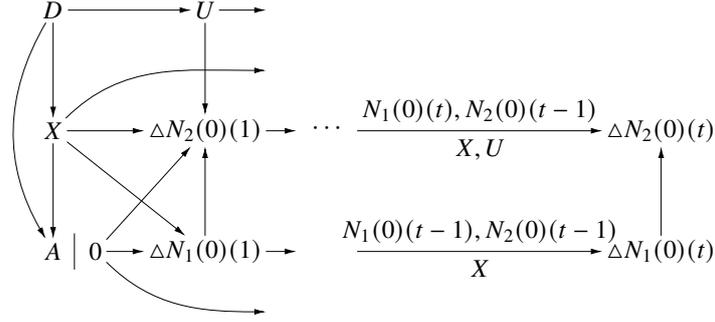

\subsection{Semiparametric efficiency bound}

If Assumption~\ref{asn:causal} is satisfied, Assumption~\ref{asn:transportability-hazard} further implies that
\begin{equation}
  \label{eqn:restriction}
\d{\Alpha}_{11}(t\mid 0,x)=\d{\Alpha}_{01}(t\mid x),\qquad x\in\mathcal{X}_{1}\cap\mathcal{X}_{0}.
\end{equation}
Consider the model \(\mathcal{P}\) of observed data distributions over \(O\) such that for any \(P\in\mathcal{P}\), the distribution of \(A\equiv 0\) is degenerate when \(D=1\), and the conditional cause 1 hazard under the control treatment \(A=0\) is transportable in the sense of \eqref{eqn:restriction}.
Since the hazard increments in \eqref{eqn:restriction} are not population-specific, we define \(\d\Alpha_{\dott 1}(t\mid 0,x)\) for all \(x\in\mathcal{X}_{1}\cup\mathcal{X}_{0}\) such that \(\d{\Alpha}_{\dott 1}(t\mid 0,x)= \d{\Alpha}_{11}(t\mid 0,x)\) if \(x\in\mathcal{X}_{1}\) and \(\d{\Alpha}_{\dott 1}(t\mid 0,x)= \d{\Alpha}_{01}(t\mid x)\) if \(x\in\mathcal{X}_{0}\).

In Proposition~\ref{ppn:eif-discontinuity} of the Supplementary Material, we characterize the semiparametric efficiency bounds of the parameters \(\theta_{1}(0)\) and \(\theta_{2}(0)\) in model \(\mathcal{P}\) under general formulations of the event processes with competing risks.
To gain insights on how the external controls can be most efficiently integrated under the transportability assumption of the cause 1 hazard, in Lemma~\ref{lem:eif} below, we state the efficient influence functions under a mild regularity condition on the cumulative hazards.
Let \(\mathcal{A}\) be the class of cumulative hazards \(\Alpha:(0,\tau]\to[0,\infty)\) which are c\`{a}dl\`{a}g, non-decreasing functions with finite variation and jump sizes no larger than \(1\).

\begin{assumption}[Disjoint discontinuity points]
  \label{asn:discontinuity}
  \(\int_{0}^\tau\triangle{\Alpha}_{11}(t\mid 0,x)\d {\Alpha}_{12}(t\mid 0,x)=0\) for \(x\in\mathcal{X}_1\) and \(\int_0^\tau\triangle{\Alpha}_{01}(t\mid x)\d{\Alpha}_{02}(t\mid x)=0\) for \(x\in\mathcal{X}_1\cap\mathcal{X}_0\).
\end{assumption}

In words, the first part of Assumption~\ref{asn:discontinuity} states that the conditional cumulative hazards \({\Alpha}_{11}(t\mid 0,x)\) and \({\Alpha}_{12}(t\mid 0,x)\) under the control treatment do not share any discontinuity point for any baseline covariates in the target population.
When there are jump points in the distribution function of the underlying event time \(T\) for a countable set of timepoints, the assumption implies that the probability that events of type \(1\) and type \(2\) are observed at the same time is \(0\).
By Assumption~\ref{asn:transportability-hazard}, this also implies that \(\triangle\Alpha_{01}(t\mid x)\triangle\Alpha_{12}(t\mid 0,x)=0\) for \(x\in\mathcal{X}_1\cap\mathcal{X}_0\).
The second part of the assumption can be interpreted analogously.
Assumption~\ref{asn:discontinuity} is certainly satisfied if \(T\) has a continuous distribution.
The assumption is also satisfied if the conditional cumulative hazard of either cause is continuous.
We will revisit Assumption~\ref{asn:discontinuity} when we construct estimators for the target parameters.

For \(a\in\{0,1\}\) and \(j,k\in\{1,2\}\), let \(N_j(t) = \I(\tilde{T}\leq t, \tilde{J}=j)\) denote the counting process for the observed event of type \(j\) and define
\begin{align*}
  H_{\dott}(t\mid x) &= \pi(x)e_1(0\mid x)(S_1S_1^c)(t\mid 0,x) + \{1-\pi(x)\}(S_0S_0^c)(t\mid x), \\
  H_1(t\mid a,x) &= e_1(a\mid x)(S_1S_1^c)(t\mid a,x), \\
  W_{kj}(t\mid a,x) &= \I(j=k)S_{1}(t\!-\!\mid a,x) - \frac{F_{1j}(\tau\mid a,x)-F_{1j}(t\mid a,x)}{1-\triangle\Alpha_{1k}(t\mid a,x)}.
\end{align*}
\begin{lemma}[Semiparametric efficiency bounds]
  \label{lem:eif}
  Suppose Assumptions~\ref{asn:censoring} and \ref{asn:discontinuity} hold.
  For \(a\in\{0,1\}\) and \(j\in\{1,2\}\), the efficient influence function of \(\theta_{j}(a)\) at \(P\in\mathcal{P}\) is
  \begin{align*}
    \label{eqn:eif}
\varphi_{j}(a)(O) &= \frac{\I(A=a)}{\alpha}\int_{0}^{\tau}\bigg\{\frac{\I(a=0)\pi(X)}{H_{\dott}(t\!-\!\mid X)}+\frac{\I(a=1)D}{H_{1}(t\!-\!\mid 1,X)}\bigg\}W_{1j}(t\mid a,X)\\
                  &\hphantom{=\frac{\I(A=a)}{\alpha}\int_{0}^{\tau}}\qquad \big\{\d N_1(t)-I(\tilde{T}\geq t)\d\Alpha_{11}(t\mid a,X)\big\} \\
                  &\hphantom{=}\quad +\frac{D}{\alpha}\int_{0}^{\tau}\frac{\I(A=a)}{H_1(t\!-\!\mid a,X)}W_{2j}(t\mid a,X)\big\{\d N_2(t)-I(\tilde{T}\geq t)\d\Alpha_{12}(t\mid a,X)\big\} \\
                      &\hphantom{=}\quad + \frac{D}{\alpha}\{{F}_{1j}(\tau\mid a,X)-\theta_{j}(a)\}.
  \end{align*}
  The semiparametric efficiency bound of \(\theta_{j}(a)\) at \(P\in\mathcal{P}\) is \(\E_{P}\varphi_{j}^2(a)\).
\end{lemma}

The efficient influence functions of the parameters \(\theta_{j}(1)\) are identical to those presented by Eq. (4) in
\citet{rytgaard2023estimation}, with the only difference being that they are restricted to the distribution on the RCT population.
Since the nuisance parameters in \(\varphi_{j}(1)(O)\) are all variationally independent of the cumulative hazards \({\Alpha}_{11}(t\mid 0,x)\) and \({\Alpha}_{01}(t\mid x)\), Assumption~\ref{asn:transportability-hazard} does not change the characterization of the efficient estimators of \(\theta_{j}(1)\).

On the other hand, comparing the efficient influence function \(\varphi_{1}(0)(O)\) with the influence function of \(\theta_{1}(0)\) without using the information of external controls, we notice that the inverse weight \(1/H_{\dott}(t\!-\!\mid x)\) is applied for efficient use of data.
We can write
\[
  H_{\dott}(t\mid x)= \Pr(T>t, C>t,A=0\mid X=x)=P(A=0\mid X=x)P(\tilde{T}> t\mid A=0,X=x),
\]
which is a product of the probability of receiving the control treatment and the survival function of an event of any type, including censoring, defined on the artificial population conjoining the whole external population and the subset of the target population under the control treatment.
It should be noted, however, that the function \(W_{11}(t\mid 0,x)\) is identifiable from the target population only.
Therefore, the predictable process in the event-of-interest martingale integral from \(\varphi_{1}(0)(O)\) is a combination of pooled and unpooled quantities across populations.

\begin{corollary}
  \label{cor:variance-reduction}
  Under the same conditions in Lemma~\ref{lem:eif}, the semiparametric efficiency bound of \(\theta_{1}(0)\) under \(\mathcal{P}\) is at least as low as that under the model where restriction \eqref{eqn:restriction} is removed.
  The reduction is
  \begin{multline*} \E\bigg[\frac{\pi(X)\{1-\pi(X)\}}{\alpha^{2}}\int_{0}^{\tau}\frac{(S_0S_0^c)(t\!-\!\mid X)}{H_1(t\!-\!\mid 0,X)H_{\dott}(t\!-\!\mid X)}\\
    W_{11}^2(t\mid 0,X)\{1-\triangle\Alpha_{11}(t\mid 0,X)\}\d\Alpha_{11}(t\mid 0,X)\bigg].
  \end{multline*}
\end{corollary}

In words, incorporating the external controls helps drop the lowest possible variance attainable by a regular estimator of the target parameter \(\theta_{1}(0)\) under the transportability assumption, if two conditions are met.
First, there is an overlap in the distributions of the baseline covariates between the populations.
Second, in this overlapped population, there is a non-trivial time span in the observation period during which an individual is at risk of experiencing the event of interest.

Corollary~\ref{cor:variance-reduction} shows that the variance reduction is accumulated over time with respect to the cumulative hazard \(\Alpha_{11}(t\mid 0,x)\), and the time-varying factors that determine the size of variance reduction cannot be teased apart.
We give some intuition on when the use of external controls provides large precision gain.
The product integral of any \(\Alpha\in\mathcal{A}\) is denoted by
\((\Pi\Alpha)(t)=\Pi_{s\in(0,t]}\{1-\d\Alpha(t)\}\).
Note that
\begin{multline*}
  \frac{(S_0S_0^c)(t\!-\!\mid x)}{H_{\dott}(t\!-\!\mid x)}=\bigg\{\pi(x)e_1(0\mid x)\frac{(\Pi\Alpha_{12}S_1^c)(t\!-\!\mid 0,x)}{(\Pi\Alpha_{02}S_0^c)(t\!-\!\mid x)}+\{1-\pi(x)\}\bigg\}^{-1}\\
  \I\big\{(S_0S_0^c)(t\!-\!\mid x)>0\big\}.
\end{multline*}
All other factors being equal, the reduction is more pronounced when the ratio between the product of product integrals
\[
\frac{(\Pi\Alpha_{12}S_1^c)(t\!-\!\mid 0,x)}{(\Pi\Alpha_{02}S_0^c)(t\!-\!\mid x)}
\]
is smaller.
In the extreme scenario where the said ratio is simply \(0\), the variance reduction formula gives
\begin{multline*}
  \E\bigg[\frac{\pi(X)}{\alpha^{2}}\I\{\pi(X)<1\}\int_{0}^{\tau}\I\big\{(S_0S_0^c)(t\!-\!\mid X)>0\big\}\\
  \frac{W_{11}^2(t\mid 0,X)}{H_1(t\!-\!\mid 0,X)}\{1-\triangle\Alpha_{11}(t\mid 0,X)\}\d\Alpha_{11}(t\mid 0,X)\bigg].
\end{multline*}
Effectively, the maximum possible reduction is the portion of asymptotic variance resulting from the martingale \(N_1(t)-\int_0^tI(\tilde{T}\geq s)\d\Alpha_{11}(s\mid 0,X)\) on the region where the indicator \(\I\big\{\pi(X)<1,(S_0S_0^c)(t\!-\!\mid X)>0\big\}\) stays \(1\).
In practical terms, when the hazard of the competing risk is much higher for subjects under the control treatment or when censoring occurs much earlier in the target population, the variance reduction is larger.
Intuitively, it is most beneficial to incorporate the external controls on the ground of hazard transportability when the hazard of the event of interest cannot be estimated well from the target population alone otherwise, due to the lack of such events in the observed data.

\subsection{Estimation}
In the following, we discuss the construction of estimators that asymptotically achieve the semiparametric efficiency bounds in Lemma~\ref{lem:eif}.
We present results for the parameter \(\theta_1(0)\) only.
An estimator for \(\theta_2(0)\) and its properties can be derived analogously.
The estimators of the parameters \(\theta_1(1)\) and \(\theta_2(1)\) do not involve the external control sample, and thus the estimation strategy for these parameters follows directly from \citet{rytgaard2023estimation}.

Suppose for the nuisance parameters, we have estimators
\[
  \big\{\hat{\Alpha}_{\dott 1}(t\mid 0,x),\hat{\Alpha}_{12}(t\mid 0,x),\hat{\Alpha}_{02}(t\mid x),\hat{\Alpha}^c_{1}(t\mid 0,x),\hat{\Alpha}^c_{0}(t\mid x)\big\}\subset\mathcal{A}
\]
and that \(\hat{e}_1(0\mid x)\) and \(\hat\pi(x)\) are valid probabilities.
The cumulative incidence function of cause \(1\) in the RCT sample are estimated by
\[
  \hat{F}_{11}(t\mid 0,x)=\int_{0}^{t}\hat{S}_1(s\!-\!\mid 0,x)\d\hat{\Alpha}_{\dott 1}(s\mid 0,x),
\]
where the integral is in the Lebesgue–Stieltjes sense, and the conditional survival function of the all-cause event time \(T\) is estimated by
\[
\hat{S}_1(t\mid 0,x)=(\Pi\hat{\Alpha}_{\dott 1}\Pi\hat{\Alpha}_{12})(t\mid 0,x).
\]
The survival functions of the censoring time are the product integrals \(\hat{S}^{c}_1=\Pi\hat{\Alpha}_1^c\) and \(\hat{S}^{c}_0=\Pi\hat{\Alpha}_0^c\), respectively.
Observing the efficient influence function \(\varphi_1(0)\) given in Lemma~\ref{lem:eif}, we define the uncentered efficient influence function and its plug-in version as
\begin{align*}
  \ell_1(0)(O) &= \varphi_1(0)(O) + \frac{D}{\alpha}\theta_1(0),\\
  \hat{\ell}_1(0)(O) &= \frac{1-A}{\hat\alpha}\hat\pi(X)\int_{0}^{\tau}\frac{\hat{W}_{\dott 1}(t\mid 0,X)}{\hat{H}_{\dott}(t\!-\!\mid X)}\big\{\d N_{1}(t)-\I(\tilde{T}\geq t)\d\hat{\Alpha}_{\dott 1}(t\mid 0,X)\big\} \\
               &\hphantom{=}\quad +\frac{D(1-A)}{\hat\alpha}\int_{0}^{\tau}\frac{\hat{W}_{21}(t\mid 0,X)}{\hat{H}_1(t\mid 0,X)}\big\{\d N_{2}(t)-\I(\tilde{T}\geq t)\d\hat{\Alpha}_{12}(t\mid 0,X)\big\} \\
               &\hphantom{=}\quad + \frac{D}{\hat\alpha}\hat{F}_{11}(\tau\mid 0,X),
\end{align*}
where
\begin{align*}
\hat{S}_{0}(t\mid x)&=(\Pi\hat\Alpha_{\dott 1})(t\mid 0,x)(\Pi\hat\Alpha_{02})(t\mid x),\\
\hat{H}_{\dott}(t\mid x) &= \hat\pi(x)\hat{e}_1(0\mid x)(\hat{S}_1\hat{S}_1^c)(t\mid 0,x) + \{1-\hat\pi(x)\}(\hat{S}_0\hat{S}_0^c)(t\mid x),\\
\hat{H}_{1}(t\mid x) &= \hat{e}_1(0\mid x)(\hat{S}_1\hat{S}_1^c)(t\mid 0,x),\\
  \hat{W}_{\dott 1}(t\mid 0,x)&=\hat{S}_1(t\!-\!\mid 0,x)-\frac{\hat{F}_{11}(\tau\mid 0,x)-\hat{F}_{11}(t\mid 0,x)}{1-\triangle\hat{\Alpha}_{\dott 1}(t\mid 0,x)},\\
    \hat{W}_{21}(t\mid 0,x)&=-\frac{\hat{F}_{11}(\tau\mid 0,x)-\hat{F}_{11}(t\mid 0,x)}{1-\triangle\hat{\Alpha}_{12}(t\mid 0,x)}.
\end{align*}
We propose the estimator
\[
  \hat\theta_1(0) = \frac{1}{n}\sum_{i=1}^{n}\hat\ell_1(0)(O_i)
\]
of \(\theta_1(0)\).

\begin{assumption}[Probability limits]
\label{asn:plim}
    There exist probability limits \(0\leq \bar\pi(x)\leq 1\), \(0\leq \bar{e}_1(0\mid x)\leq 1\) such that \(\|(\hat\pi-\bar\pi)(X)\|_P=o_{P}(1)\), \(\|(\hat{e}_1-\bar{e}_1)(0\mid X)\|_P=o_P(1)\),
    and
    \[
      \big\{\bar{\Alpha}_{\dott 1}(t\mid 0,x),\bar{\Alpha}_{12}(t\mid 0,x),\bar{\Alpha}_{02}(t\mid x),\bar{\Alpha}_{1}^c(t\mid 0,x),\bar{\Alpha}_{0}^c(t\mid x)\big\}\subset\mathcal{A}
    \]
    such that
    \begin{align*}
      \bigg\|\I\{\pi(X)>0\}\sup_{t\in(0,\tau]}|\hat{\Alpha}_{\dott 1}-\bar{\Alpha}_{\dott 1}|(t\mid 0,X)\bigg\|_{P} &= o_{P}(1), \\
      \bigg\|\I\{\pi(X)>0\}\sup_{t\in(0,\tau]}|\hat{\Alpha}_{12}-\bar{\Alpha}_{12}|(t\mid 0,X)\bigg\|_{P} &= o_{P}(1), \\
      \bigg\|\I\{0<\pi(X)<1\}\sup_{t\in(0,\tau]}|\hat{\Alpha}_{02}-\bar{\Alpha}_{02}|(t\mid X)\bigg\|_{P} &= o_{P}(1), \\
      \bigg\|\I\{\pi(X)>0\}\sup_{t\in(0,\tau]}|\hat{\Alpha}_{1}^c-\bar{\Alpha}_{1}^c|(t\mid 0,X)\bigg\|_{P} &= o_{P}(1), \\
      \bigg\|\I\{0<\pi(X)<1\}\sup_{t\in(0,\tau]}|\hat{\Alpha}_{0}^c-\bar{\Alpha}_{0}^c|(t\mid X)\bigg\|_{P} &= o_{P}(1).
    \end{align*}
\end{assumption}

\begin{theorem}[Asymptotic behavior]
  \label{thm:asymptotic}
  Suppose Assumptions~\ref{asn:censoring} and \ref{asn:plim} as well as Assumption~\ref{asn:regularity} in the Supplementary Material hold.
  Then \(\hat\theta_1(0)\overset{\mathrm{p}}{\to}\theta_1(0)\) if
  \begin{enumerate}[nosep,label=(\roman*)]
  \item \(\bar{\Alpha}_{\dott 1} = \Alpha_{\dott 1}\) and \(\bar{\Alpha}_{12} = \Alpha_{12}\);
  \item \(\bar{\Alpha}_{\dott 1} = \Alpha_{\dott 1}\), \(\bar{e}_1=e_1\), and \(\bar\pi=\pi\); or
  \item \(\bar{\Alpha}_{12} = \Alpha_{12}\), \(\bar{\Alpha}_{02} = \Alpha_{02}\), \(\bar{\Alpha}^c_{1} = \Alpha^c_{1}\), \(\bar{\Alpha}^c_{0} = \Alpha^c_{0}\), \(\bar{e}_1=e_1\), and \(\bar\pi=\pi\).
  \end{enumerate}
  Moreover,
  \[
    \hat\theta_1(0)-\theta_1(0)=\frac{1}{n}\sum_{i=1}^{n}\varphi_1(0)(O_i)+o_{P}(n^{-1/2})
  \]
  if \(\bar{\Alpha}_{\dott 1} = \Alpha_{\dott 1}\), \(\bar{\Alpha}_{12} = \Alpha_{12}\), \(\bar{\Alpha}_{02} = \Alpha_{02}\), \(\bar{\Alpha}^c_{1} = \Alpha^c_{1}\), \(\bar{\Alpha}^c_{0} = \Alpha^c_{0}\), \(\bar{e}_1=e_1\), \(\bar\pi=\pi\), and Assumption~\ref{asn:rate} in the Supplementary Material is satisfied.
\end{theorem}

The first part of Theorem~\ref{thm:asymptotic} shows that the estimator \(\hat\theta_1(0)\) constructed from the efficient influence function is triply robust against model misspecification.
The consistency of \(\hat\theta_1(0)\) hinges on correct estimation of at least one of the cause-specific hazards, namely \(\Alpha_{\dott 1}(t\mid 0,x)\) or \(\Alpha_{12}(t\mid 0,x)\).
In particular, if the cause 1 hazard does not converge to the underlying hazard, the cause 2 hazards in both populations need to be modeled correctly.
Conditions for the asymptotic linearity of \(\hat\theta_1(0)\) are given in the second part of Theorem~\ref{thm:asymptotic}.
Apart from requiring the consistency of all nuisance models in their respective sense, the von Mises expansion of \(\hat\theta_1(0)\) around the true parameter \(\theta_1(0)\) demands that the two remainder terms in Assumption~\ref{asn:rate} converge as fast as \(o_{P}(n^{-1/2})\); see Remark~\ref{rem:rate} for details.

The estimator \(\hat\theta_1(0)\) attains the semiparametric efficiency bound in the model where the conditional cause 1 hazard under the control treatment is transportable.
However, there is no free lunch.
Compared to estimators that do not rely on the external controls, the proposed data fusion estimator involve additional nuisance models for the selection score \(\pi\) and the cumulative hazards \(\Alpha_{02}\) and \(\Alpha_0^c\).
If these models are not correctly estimated, we have no guarantee that \(\hat\theta_1(0)\) will be more efficient than estimators based solely on the RCT sample.

A final remark should be made in connection with Assumption~\ref{asn:discontinuity}.
When tied event times are observed for event types \(1\) and \(2\), the plug-in estimators based on the efficient influence function under Assumption~\ref{asn:discontinuity} might be unfounded.
We can avoid this issue if the event times are continuous by nature, and it is harmless to break the ties by numerical perturbations.
Otherwise, we can turn to fully discrete-time methods \citep{benkeser2018improved} or derive estimators based on Proposition~\ref{ppn:eif-discontinuity} to handle mixed event time distributions.

\subsection{Restricted mean time lost}
   
Another interpretable parameter in competing risks analysis is the \(\tau\)-restricted mean time lost to cause \(j\) \citep{andersen2013decomposition}, defined as
\[
  \gamma_{j}(a)=\E(\I\{J(a)=j\}[\tau-\{T(a)\wedge \tau\}]\mid D=1).
\]
We can extend the definition of the parameter \(\theta_{j}(a)\) to the population cumulative incidence of event type \(j\) under intervention \(a\) at time \(t\), which is \(\theta_j(a,t)=\Pr\{T(a)\leq t,J(a)=j\mid D=1\}\).
Given Assumptions~\ref{asn:causal}--\ref{asn:censoring}, the parameter \(\gamma_{j}(a)\) is identifiable as \(\gamma_{j}(a)= \int_{0}^{\tau}\theta_{j}(a,t)\d t\), where \(\theta_j(a,t)\) is treated as an observed data parameter.
If we view \(\theta_j(a,t)\) as a function of time, then \(\gamma_j(a)\) is simply the area under the cumulative incidence function capped at \(\tau\).

Restricted mean times lost are Hadamard differentiable functionals of the cumulative incidence functions.
Hence, their efficient influence functions can be obtained from Lemma~\ref{lem:eif} by the chain rule.

\begin{corollary}
  \label{cor:eif-rmtl}
  Under the same conditions as in Lemma~\ref{lem:eif}, the efficient influence function of \(\gamma_{j}(a)\) at \(P\in\mathcal{P}\) is \(\psi_{j}(a)(O)=\int_{0}^{\tau}\varphi_{j}(a,t)(O)\d t\), where \(\varphi_{j}(a,t)\) is the efficient influence function of \(\theta_j(a,t)\).
\end{corollary}

Fusion estimators for \(\gamma_{j}(a)\) are straightforward integrals of the fusion estimators \(\hat\theta_{j}(a,t)\) for \(\theta_j(a,t)\) over time \(t\).
The asymptotics of these estimators can be established under conditions similar to those inside Theorem~\ref{thm:asymptotic}.
Particularly for asymptotic linearity of \(\hat{\gamma}_1(0)\), the rate conditions in Assumption~\ref{asn:rate} should be modified according to Remark~\ref{rem:rate-rmtl} in the Supplementary Material.

\section{Simulation study}

We investigated the performance of the fusion estimators compared to the RCT-only estimators in a simulation study.
The data at baseline \((X,D,A)\) were generated sequentially in the following manner:
\begin{align*}
  X &\sim 2\Phi[\mathrm{Normal}\{(0,0,0)^\T,\Sigma\}]-1,\\
  D\mid X &\sim\mathrm{Bernoulli}\{\pi(X)\}, \\
  A\mid (X,D) &\sim \mathrm{Bernoulli}(0.5D),
\end{align*}
where \(\pi(X)=\mathrm{expit}(-0.2+0.4X_1+0.2X_2+0.3X_3)\), \(\Phi\) is the distribution function of the standard normal distribution, and the covariance matrix is
\[
  \Sigma=
  \begin{pmatrix}
    1 & 0.25 & 0.25 \\
    0.25 & 1 & 0.25 \\
    0.25 & 0.25 & 1
  \end{pmatrix}.
\]
The uncensored event times were simulated from distributions with the following multiplicative hazards:
\begin{align*}
  \alpha_{11}(t\mid A,X) &= \alpha_{1}(t)\exp(0.5A + 0.2 X_{1}+ 0.7 X_{3}), \\
  \alpha_{12}(t\mid A,X) &= \alpha_{2}(t)\exp(1 + 0.05A + 0.8 X_{1}+0.5 X_{2} ), \\
  \alpha_{01}(t\mid X) &= \alpha_{1}(t)\exp(0.2 X_{1} + 0.7 X_{3}), \\
  \alpha_{02}(t\mid X) &= \alpha_{2}(t)\exp(0.5 X_{1} + 0.8 X_{2} - 0.3X_3),
\end{align*}
where the baseline hazards \(\alpha_{1}(t)\) and \(\alpha_{2}(t)\) both correspond to the hazard of the Weibull distribution with shape parameter \(0.7\) and scale parameter \(0.2\).
In other words, \(\alpha_{d}(t)=0.2 \cdot 0.7 t^{0.7-1}\).
The censoring times were simulated from distributions with the following multiplicative hazards:
\begin{align*}
  \alpha^{c}_1(t\mid A,X) &= \alpha^{c}(t)\exp\{0.5+0.05(1-A)X_1-0.05X_{3}\}, \\
  \alpha^{c}_0(t\mid X) &= \alpha^{c}(t)\exp(0.05X_{2}),
\end{align*}
where the baseline hazard \(\alpha^{c}(t)\) is the hazard of the Weibull distribution with shape parameter \(0.7\) and scale parameter \(0.24\).
Under this data generating mechanism, the proportion of samples from the external control population was around \(55\%\).

The target population selection score \(\hat{\pi}(x)\) and the propensity score of treatment in the target population \(\hat{e}_1(a\mid x)\) were estimated using logistic regressions.
The cause 1 hazard under the control treatment \(\d\hat{\Alpha}_{\dott 1}(t\mid 0,x)\) was fitted with a Cox model combining all samples under the control treatment and the event indicator \(\I(\tilde{J}=1)\).
The cause 2 hazards under the control treatment \(\d\hat{\Alpha}_{d2}(t\mid 0,x)\) were fitted with a Cox model within the respective population using the event indicator \(\I(\tilde{J}=2)\).
The two cause-specific hazards under the active treatment \(\d\hat{\Alpha}_{1j}(t\mid 1,x)\) were obtained with a multi-state Cox model in the RCT population using the state indicator \(\tilde{J}\).
The hazards of the censoring were fitted with a Cox model for each treatment within the respective population using the event indicator \(\I(\tilde{J}=0)\).
The nuisance function estimates \(\hat{S}_d,\hat{F}_{dj},\hat{S}^{c}_d\) were subsequently computed using the hazard estimates.

As an example, we present the nuisance estimators required for the estimator \(\hat{\theta}_{1}(0)\), which included \(\hat{S}_1(t\mid 0,x)\), \(\hat{F}_{11}(t\mid 0,x)\), \(\hat{S}^{c}_1(t\mid 0,x)\), \(\hat{S}_0(t\mid x)\), and \(\hat{S}^{c}_0(t\mid x)\).
The cumulative hazard estimates from Cox models are c\`{a}dl\`{a}g step functions.
The approximation \(\triangle\hat{\Alpha}_{\dott 1}(s\mid 0,x)\approx 1-\exp\{-\triangle\hat{\Alpha}_{\dott 1}(s\mid 0,x)\}\) was applied due to possible jumps whose sizes exceed one, ensuring that it fell between \(0\) and \(1\).
The survival function of the composite event in the RCT population was approximated by \(\hat{S}_1(t\mid 0,x)=\exp\bigl\{-\hat{\Alpha}_{\dott 1}(t\mid 0,x)-\hat{\Alpha}_{12}(t\mid 0,x)\bigr\}\).
The cumulative incidence function of the event of interest was computed using the Lebesgue-Stieltjes integral \(\hat{F}_{11}(t\mid 0,x)=\int_{0}^{t}\hat{S}_1(s\!-\!\mid 0,x)\d\hat{\Alpha}_{\dott 1}(s\mid 0,x)\).
Similarly, the survival function of the composite event in the external control population was \(\hat{S}_0(t\mid x)=\exp\bigl\{-\hat{\Alpha}_{\dott 1}(t\mid 0,x)-\hat{\Alpha}_{02}(t\mid x)\bigr\}\).
The Cox-estimated cumulative hazard \(\hat{\Alpha}_{\dott 1}(t\mid 0,x)\) did not share any discontinuity points with \(\hat{\Alpha}_{02}(t\mid x)\), since there were no ties among event times of different types.
The survival functions of the censoring time are \(\hat{S}^{c}_1(t\mid 0,x)=\exp\bigl\{-\hat{\Alpha}^{c}_1(t\mid 0,x)\bigr\}\) and \(\hat{S}^{c}_0(t\mid x)=\exp\bigl\{-\hat{\Alpha}^{c}_0(t\mid x)\bigr\}\), respectively.


We simulated data of sample size \(n\in\{750,1500\}\) from the described data generating mechanism.
The proposed estimator \(\hat{\theta}_{1}(0,t)\) for the cumulative incidence of the event of interest under control treatment was computed for three time points \(t\in\{0.25, 1, 2\}\).
The standard error of \(\hat{\theta}_{1}(0,t)\) was estimated by \(n^{-1/2}\) times the empirical \(L_2\)-norm of the efficient influence function \(\varphi_1(0)\).
The estimators \(\hat{\theta}_2(0,t)\), \(\hat{\theta}_1(1,t)\), and \(\hat{\theta}_2(1,t)\) for other cumulative incidences and the estimators \(\hat\theta_1\{t\}\) and \(\hat\theta_2\{t\}\) for the average treatment effects were also computed using the respective efficient influence functions.
To demonstrate the gain in precision, we compared the estimated asymptotic variance of the estimators above to the estimators that would be efficient if only the RCT data was available.
The exact expressions of all other estimators can be found in Supplementary Materials \S\ref{sec:sim-app}.
As an alternative effect measure, we also considered the treatment effect as the difference between restricted mean times lost capped at \(t\in\{0.25, 1, 2\}\).
The calculations were repeated \(1000\) times for each sample size.

Summary statistics of selected estimators from the simulation study are reported in Tables~\ref{tab:sim-cif}--\ref{tab:sim-rmtl}.
Results for the remaining estimators are deferred to Tables~\ref{tab:sim-cif-app-1}--\ref{tab:sim-rmtl-app-2} in the Supplementary Material.
All estimators have small empirical bias.
The averages of the plug-in standard errors align with the empirical root mean squared errors.
The coverage of the \(95\%\)-confidence intervals constructed from these plug-in standard errors appears largely correct.
The percentage reduction in variance is higher for the control parameters \(\theta_1(0,t)\) and \(\gamma_1(0,t)\) than it is for the treatment effects \(\theta_1\{t\}\) and \(\gamma_1\{t\}\).
This should be expected since the parameters \(\theta_1(1,t)\) and \(\gamma_1(1,t)\), and thus their estimators, do not use information from external controls.

\begin{table}
  \caption{Simulation results for estimators of cumulative incidences.}
  \label{tab:sim-cif}
  \footnotesize
  \centering

\begin{tabular}{rlrlrrrrrr}
{\(n\)} & {Estimand} & {\(t\)} & {Type} & {Mean} & {Bias} & {RMSE} & {SE} & {Coverage} & {Reduction}\\
\(750\) & \(\theta_1(0,t)\) & \(0.25\) & \(+\) & \(0.07\) & \(5.73\) & \(1.17\) & \(1.15\) & \(94.5\) & \(66.84\)\\
 &  &  & \(-\) & \(0.07\) & \(8.20\) & \(2.15\) & \(2.05\) & \(92.7\) & \(\cdot\)\\
 &  & \(1\) & \(+\) & \(0.14\) & \(3.66\) & \(1.60\) & \(1.60\) & \(95.1\) & \(69.42\)\\
 &  &  & \(-\) & \(0.14\) & \(5.50\) & \(2.99\) & \(2.92\) & \(92.3\) & \(\cdot\)\\
 &  & \(2\) & \(+\) & \(0.19\) & \(2.03\) & \(1.82\) & \(1.84\) & \(94.8\) & \(69.20\)\\
 &  &  & \(-\) & \(0.19\) & \(-4.92\) & \(3.46\) & \(3.35\) & \(93.5\) & \(\cdot\)\\
 & \(\theta_1(t)\) & \(0.25\) & \(+\) & \(0.04\) & \(-3.03\) & \(2.70\) & \(2.76\) & \(95.1\) & \(27.35\)\\
 &  &  & \(-\) & \(0.04\) & \(-5.50\) & \(3.17\) & \(3.25\) & \(96.3\) & \(\cdot\)\\
 &  & \(1\) & \(+\) & \(0.08\) & \(8.53\) & \(3.68\) & \(3.77\) & \(95.6\) & \(29.58\)\\
 &  &  & \(-\) & \(0.08\) & \(6.69\) & \(4.50\) & \(4.50\) & \(95.1\) & \(\cdot\)\\
 &  & \(2\) & \(+\) & \(0.09\) & \(14.66\) & \(4.07\) & \(4.23\) & \(96.2\) & \(30.43\)\\
 &  &  & \(-\) & \(0.09\) & \(21.61\) & \(5.04\) & \(5.08\) & \(94.7\) & \(\cdot\)\\
\(1500\) & \(\theta_1(0,t)\) & \(0.25\) & \(+\) & \(0.07\) & \(1.79\) & \(0.81\) & \(0.81\) & \(94.3\) & \(68.31\)\\
 &  &  & \(-\) & \(0.07\) & \(3.51\) & \(1.42\) & \(1.46\) & \(94.9\) & \(\cdot\)\\
 &  & \(1\) & \(+\) & \(0.14\) & \(3.60\) & \(1.09\) & \(1.13\) & \(95.5\) & \(70.17\)\\
 &  &  & \(-\) & \(0.14\) & \(-0.17\) & \(2.07\) & \(2.07\) & \(94.2\) & \(\cdot\)\\
 &  & \(2\) & \(+\) & \(0.19\) & \(2.23\) & \(1.25\) & \(1.30\) & \(96.0\) & \(70.13\)\\
 &  &  & \(-\) & \(0.19\) & \(-5.99\) & \(2.38\) & \(2.39\) & \(94.8\) & \(\cdot\)\\
 & \(\theta_1(t)\) & \(0.25\) & \(+\) & \(0.04\) & \(-3.05\) & \(2.00\) & \(1.95\) & \(94.1\) & \(27.73\)\\
 &  &  & \(-\) & \(0.04\) & \(-4.77\) & \(2.28\) & \(2.30\) & \(95.2\) & \(\cdot\)\\
 &  & \(1\) & \(+\) & \(0.07\) & \(-6.90\) & \(2.76\) & \(2.66\) & \(93.6\) & \(29.93\)\\
 &  &  & \(-\) & \(0.07\) & \(-3.13\) & \(3.24\) & \(3.18\) & \(94.0\) & \(\cdot\)\\
 &  & \(2\) & \(+\) & \(0.09\) & \(-12.92\) & \(3.06\) & \(2.99\) & \(94.2\) & \(30.98\)\\
 &  &  & \(-\) & \(0.09\) & \(-4.70\) & \(3.59\) & \(3.60\) & \(95.1\) & \(\cdot\)\\
\end{tabular}

  \medskip
  \begin{flushleft}
    Type: fusion estimator (\(+\)) or RCT-only estimator (\(-\)); Mean: average of estimates; Bias: Monte-Carlo bias, \(10^{-4}\); RMSE: root mean squared error, \(10^{-2}\); SE: average of standard error estimates, \(10^{-2}\); Coverage: \(95\%\) confidence interval coverage, \(\%\); Reduction: average of percentage reduction in squared standard error estimates, \(\%\).
  \end{flushleft}
\end{table}

\begin{table}
  \caption{Simulation results for estimators of restricted mean times lost.}
  \label{tab:sim-rmtl}
  \footnotesize
  \centering

\begin{tabular}{rlrlrrrrrr}
{\(n\)} & {Estimand} & {\(t\)} & {Type} & {Mean} & {Bias} & {RMSE} & {SE} & {Coverage} & {Reduction}\\
\(750\) & \(\gamma_1(0,t)\) & \(0.25\) & \(+\) & \(0.01\) & \(-0.20\) & \(0.20\) & \(0.20\) & \(94.6\) & \(64.49\)\\
 &  &  & \(-\) & \(0.01\) & \(0.66\) & \(0.36\) & \(0.35\) & \(92.1\) & \(\cdot\)\\
 &  & \(1\) & \(+\) & \(0.10\) & \(-4.03\) & \(1.17\) & \(1.16\) & \(94.6\) & \(67.61\)\\
 &  &  & \(-\) & \(0.10\) & \(0.12\) & \(2.12\) & \(2.07\) & \(93.0\) & \(\cdot\)\\
 &  & \(2\) & \(+\) & \(0.26\) & \(-18.36\) & \(2.78\) & \(2.74\) & \(93.4\) & \(68.04\)\\
 &  &  & \(-\) & \(0.26\) & \(-19.14\) & \(5.06\) & \(4.89\) & \(92.4\) & \(\cdot\)\\
 & \(\gamma_1(t)\) & \(0.25\) & \(+\) & \(0.01\) & \(-0.36\) & \(0.46\) & \(0.47\) & \(94.9\) & \(26.45\)\\
 &  &  & \(-\) & \(0.01\) & \(-1.21\) & \(0.55\) & \(0.56\) & \(95.9\) & \(\cdot\)\\
 &  & \(1\) & \(+\) & \(0.05\) & \(-0.03\) & \(2.62\) & \(2.73\) & \(96.0\) & \(28.21\)\\
 &  &  & \(-\) & \(0.05\) & \(-4.18\) & \(3.14\) & \(3.22\) & \(94.3\) & \(\cdot\)\\
 &  & \(2\) & \(+\) & \(0.14\) & \(2.81\) & \(6.04\) & \(6.33\) & \(96.5\) & \(29.05\)\\
 &  &  & \(-\) & \(0.14\) & \(3.58\) & \(7.37\) & \(7.52\) & \(95.3\) & \(\cdot\)\\
\(1500\) & \(\gamma_1(0,t)\) & \(0.25\) & \(+\) & \(0.01\) & \(-0.08\) & \(0.14\) & \(0.14\) & \(94.2\) & \(66.55\)\\
 &  &  & \(-\) & \(0.01\) & \(0.15\) & \(0.25\) & \(0.25\) & \(93.7\) & \(\cdot\)\\
 &  & \(1\) & \(+\) & \(0.10\) & \(-2.50\) & \(0.81\) & \(0.82\) & \(95.0\) & \(68.40\)\\
 &  &  & \(-\) & \(0.10\) & \(-3.70\) & \(1.46\) & \(1.47\) & \(94.7\) & \(\cdot\)\\
 &  & \(2\) & \(+\) & \(0.26\) & \(-9.76\) & \(1.88\) & \(1.94\) & \(95.2\) & \(68.73\)\\
 &  &  & \(-\) & \(0.26\) & \(-16.13\) & \(3.48\) & \(3.49\) & \(95.2\) & \(\cdot\)\\
 & \(\gamma_1(t)\) & \(0.25\) & \(+\) & \(0.01\) & \(-0.62\) & \(0.34\) & \(0.34\) & \(94.6\) & \(26.69\)\\
 &  &  & \(-\) & \(0.01\) & \(-0.85\) & \(0.39\) & \(0.40\) & \(94.9\) & \(\cdot\)\\
 &  & \(1\) & \(+\) & \(0.05\) & \(-4.73\) & \(2.01\) & \(1.94\) & \(93.5\) & \(28.43\)\\
 &  &  & \(-\) & \(0.05\) & \(-3.54\) & \(2.31\) & \(2.29\) & \(94.8\) & \(\cdot\)\\
 &  & \(2\) & \(+\) & \(0.13\) & \(-17.48\) & \(4.67\) & \(4.49\) & \(93.0\) & \(29.33\)\\
 &  &  & \(-\) & \(0.14\) & \(-11.11\) & \(5.39\) & \(5.35\) & \(94.6\) & \(\cdot\)\\
\end{tabular}

  \medskip
  \begin{flushleft}
    {Type: fusion estimator (\(+\)) or RCT-only estimator (\(-\)); Mean: average of estimates; Bias: Monte-Carlo bias, \(10^{-4}\); RMSE: root mean squared error, \(10^{-2}\); SE: average of standard error estimates, \(10^{-2}\); Coverage: \(95\%\) confidence interval coverage, \(\%\); Reduction: average of percentage reduction in squared standard error estimates, \(\%\).}
  \end{flushleft}
\end{table}

\section{Real data example}

In this data example, we use data from the clinical trials SUSTAIN-6 (ClinicalTrials.gov ID NCT01720446, \citealp{marso2016semaglutide}) and LEADER (ClinicalTrials.gov ID NCT01179048, \citealp{marso2016liraglutide}).
The overall objective is to incorporate the controls collected in LEADER (\(D=0\)) in the statistical analysis on the study population of SUSTAIN-6 (\(D=1\)) to boost the precision of estimates.
The number of subjects randomized to placebo is \(1649\) in SUSTAIN-6 and \(4672\) in LEADER.
The placebos are both subcutaneous injections matched to their corresponding active treatment.
The frequency of injection is once daily in LEADER but once weekly in SUSTAIN-6.
We proceed by regarding the three placebos as the same intervention.

\begin{table}
  \caption{Numbers of randomized subjects and events by arm in SUSTAIN-6 and LEADER.}
  \label{tab:sample-size}
  \footnotesize
  \centering
  \begin{tabular}{lrrrrrr}
    & \multicolumn{4}{c}{SUSTAIN-6} & \multicolumn{2}{c}{LEADER} \\
    & \multicolumn{4}{c}{(once-weekly)} & \multicolumn{2}{c}{(once-daily)} \\
    & \multicolumn{2}{c}{Semaglutide} & \multicolumn{2}{c}{Placebo} & Liraglutide & Placebo \\
    & \(1.0\) mg & \(0.5\) mg &  \(1.0\) mg & \(0.5\) mg & \(1.8\) mg & \(1.8\) mg \\
    Total & \(822\) & \(826\) & \(825\) & \(824\) & \(4668\) & \(4672\) \\
    Non-fatal cardiovascular event & \(29\) & \(38\) & \(48\) & \(53\) & \(242\) & \(271\) \\
    All-cause death & \(23\) & \(24\) & \(21\) & \(27\) & \(135\) & \(155\)
  \end{tabular}

  \medskip
  {Total: total number of randomized subjects at baseline; the other numbers count the non-fatal cardiovascular events and all-cause deaths  on or before day \(728\).}
\end{table}

We define the event of interest as the composite event of nonfatal myocardial infarction or nonfatal stroke (\(J=1\)), which we refer to as the non-fatal cardiovascular event.
We treat death from all causes as the competing event (\(J=0\)).
Time-zero in the analysis is the time of treatment or placebo randomization.
The first set of parameters we considered were the cumulative incidences \(\theta_{j}(a,t)\) for both events at week \(26\), week \(52\), week \(78\), and week \(104\) in the study population of SUSTAIN-6 of once-weekly semaglutide, 1.0 mg (\(A=1\)), as well as the average treatment effects \(\theta_j\{t\}=\theta_j(1,t)-\theta_j(0,t)\).
We set the limit of the time span as the end of the follow-up period in SUSTAIN-6.
We also considered the restricted mean times lost to the events \(\gamma_j(a)\) capped at week \(26\), week \(52\), week \(78\), and week \(104\) and the corresponding effect \(\gamma_j\{t\}=\gamma_j(1,t)-\gamma_j(0,t)\).
See Table~\ref{tab:sample-size} for a breakdown of sample sizes by randomization arm and the numbers of events observed until week \(104\).

The main analysis was carried out under the transportability assumption that the cause-specific hazard of the event of interest under placebo, conditioning on relevant baseline covariates, is the same in the study population of SUSTAIN-6 and that of LEADER.
The baseline covariates \(X\) to adjust for included age, sex, weight, duration of type-2 diabetes, glycated hemoglobin level, systolic and diastolic blood pressure, level of low-density lipoprotein cholesterol, smoking status, as well as history of ischemic heart disease, myocardial infarction, heart failure, ischemic stroke, and hypertension.
The inclusion-exclusion criteria for the studies are highly comparable.
Therefore, the transportability assumption implies that any difference in the marginal hazard of the event of interest can be attributed to the differences in the rates of death, in the rates of censoring, and/or in the baseline characteristics induced by sampling.
The causal and transportability assumptions are compatible with the local independence graph \citep{didelez2008graphical} without right-censoring in Figure~\ref{fig:data-example}.
The cause-specific hazards and hazards of censoring were estimated with the Cox proportional hazards model.
The hazards in the RCT sample were fitted separately within the treatment arms to ensure full treatment-covariate interaction and per-stratum baseline hazards.

\begin{figure}
  \centering
  \scalebox{0.8}{
  \begin{tikzpicture}
    \node[vertex] (A) at (162+72:1.5) {\(A\)};
    \node[vertex] (X) at (162+4*72:1.5) {\(X\)};
    \node[vertex] (D) at (162+5*72:1.5) {\(D\)};
    \node[vertex] (N1) at (162+2*72:1.5) {\(N_1\)};
    \node[vertex] (N2) at (162+3*72:1.5) {\(N_2\)};
    \node[above right=-1ex of X] {Baseline covariates};
    \node[right=0ex of N2] {All-cause death};
    \node[below right=-1ex of N1,align=left] {Non-fatal adverse\\[-.2pc] cardiovascular event};
    \node[below left=-1ex of A] {Treatment assignment};
    \node[left=0ex of D] {SUSTAIN-6/LEADER};
    \path[-latex]
    (A) edge (N1)
    (A) edge (N2)
    (D) edge (X)
    (D) edge (A)
    (D) edge (N2)
    (X) edge (N1)
    (X) edge (N2)
    (N1) edge[bend left,in=160,out=20] (N2)
    (N2) edge[bend left,in=160,out=20] (N1);
  \end{tikzpicture}
  }
  
  \caption{Hypothesized local independence graph of the variables used in the data example.
    The nodes \(N_{j}(t)\) are uncensored versions of the counting processes.}
  \label{fig:data-example}
\end{figure}
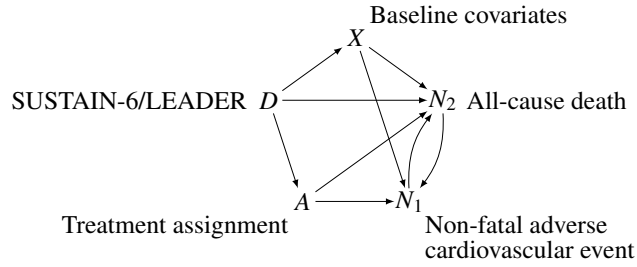

The results are reported in Tables~\ref{tab:analysis-cif}--\ref{tab:analysis-rmtl}.
We highlight the results at week \(104\).
The fusion estimate of \(\theta_1\{104\}\) demonstrates a decrease of \(2.72\) percentage points [\(95\%\)-confidence interval: \((-4.33,-1.12)\)] in the cumulative incidence of non-fatal cardiovascular event by semaglutide.
There appears to be no evidence for semaglutide's effect on the cumulative incidence of all-cause death \(\theta_2\{104\}\).
Semaglutide does not seem to lower the risk of non-fatal cardiovascular event because of an increased risk of cardiovascular death.
The restricted mean time lost to non-fatal cardiovascular event \(\gamma_1\{104\}\) reduces by \(1.15\) week [\(95\%\)-confidence interval: \((-2.24,-0.07)\)] with semaglutide.
Again, semaglutide does not appear to change the restricted mean time lost to all-cause death.
The estimates and confidence intervals for time points before week \(104\) do not hint at any treatment effect on non-fatal cardiovascular event, except for \(\theta_{1}\{78\}\).
The results for treatment-specific parameters are displayed in Tables~\ref{tab:analysis-cif-app}--\ref{tab:analysis-rmtl-app} of the Supplementary Material.

For treatment effects of semaglutide in terms of \(\theta_1\{t\}\) and \(\gamma_1\{t\}\), the fusion point estimate and the RCT-only estimates are rather comparable.
On the other hand, the length of confidence intervals is shortened by approximately \(9\%\) for all treatment effects.
Despite the inclusion of external controls amounting to nearly three times the controls in SUSTAIN-6, the precision gain is hardly impressive.
We believe the information bottleneck is the lack of subjects receiving the active treatment, since the size of the placebo group in SUSTAIN-6 was already twice as large as that of the semaglutide \(1.0\) mg group.
This is supported by the observation that the percentage reduction in standard errors is above \(20\%\) for the under-placebo parameters \(\theta_1(0,t)\) and \(\gamma_1(0,t)\).

To mimic the setup where the size of the control arm is much smaller than the size of the treatment arm, we randomly discarded \(75\%\) of the controls from SUSTAIN-6 and repeated the analysis.
The resulting fusion estimators for the treatment effects exhibited some deviation from the RCT-only estimators, but a large precision gain at approximately \(45\)--\(50\%\) was observed; see Tables~\ref{tab:analysis-cif-fewer-controls-app} and \ref{tab:analysis-rmtl-fewer-controls-app}.
Finally, to evaluate the impact of omitted variable bias, we performed a sensitivity analysis by removing history of cardiovascular diseases from the set of baseline covariates.
While the reduction in standard errors was twice as large compared to the main analysis, there was also a more substantial difference between point estimates of \(\theta_1\{104\}\) and \(\gamma_1\{t\}\); see Tables~\ref{tab:analysis-cif-sensitivity-no-history-app} and \ref{tab:analysis-rmtl-sensitivity-no-history-app}.
It is thus unclear whether the transportability assumption holds at all with this restricted set of baseline covariates.
Further details on the data example and results from the additional analysis are available in Supplementary Material \S\ref{sec:data-example-app}.

\begin{table}
  \caption{Cumulative incidence differences in the real data example.}
  \label{tab:analysis-cif}
  \footnotesize
  \centering
  
  \begin{tabular}{lrlrrr}
    Estimand & \(t\) (weeks) & Type & Estimate (\%) & 95\%-CI (\%) & Reduction\\
    \(\theta_1\{t\}\) & 26 & \(+\) & \(-0.26\) & \((-1.27,0.75)\) & \(8.31\)\\
             &  & \(-\) & \(-0.39\) & \((-1.49,0.71)\) & .\\
             & 52 & \(+\) & \(-0.61\) & \((-1.87,0.64)\) & \(8.65\)\\
             &  & \(-\) & \(-0.73\) & \((-2.10,0.65)\) & .\\
             & 78 & \(+\) & \(-1.99\) & \((-3.40,-0.57)\) & \(9.18\)\\
             &  & \(-\) & \(-1.78\) & \((-3.34,-0.22)\) & .\\
             & 104 & \(+\) & \(-2.72\) & \((-4.33,-1.12)\) & \(9.76\)\\
             &  & \(-\) & \(-2.56\) & \((-4.33,-0.78)\) & .\\
    \(\theta_2\{t\}\) & 26 & \(+\) & \(-0.49\) & \((-1.05,0.07)\) & \(-0.00\)\\
             &  & \(-\) & \(-0.49\) & \((-1.05,0.07)\) & .\\
             & 52 & \(+\) & \(-0.38\) & \((-1.26,0.50)\) & \(-0.00\)\\
             &  & \(-\) & \(-0.38\) & \((-1.26,0.50)\) & .\\
             & 78 & \(+\) & \(0.05\) & \((-1.17,1.28)\) & \(-0.00\)\\
             &  & \(-\) & \(0.05\) & \((-1.17,1.28)\) & .\\
             & 104 & \(+\) & \(-0.21\) & \((-1.71,1.28)\) & \(-0.00\)\\
             &  & \(-\) & \(-0.21\) & \((-1.71,1.28)\) & .\\
  \end{tabular}
  \medskip
  \begin{flushleft}
    {Type: fusion estimator (\(+\)) or RCT-only estimator (\(-\)); CI: confidence interval; Reduction: percentage reduction CI length, \(\%\).}
  \end{flushleft}
\end{table}

\begin{table}
  \caption{Restricted mean time lost differences in the real data example.}
  \label{tab:analysis-rmtl}
  \footnotesize
  \centering
  
  \begin{tabular}{lrlrrr}
    Estimand & \(t\) (weeks) & Type & Estimate (weeks) & 95\%-CI (weeks) & Reduction\\
    \(\gamma_1\{t\}\) & 26 & \(+\) & \(-0.10\) & \((-0.23,0.03)\) & \(10.32\)\\
             &  & \(-\) & \(-0.11\) & \((-0.25,0.04)\) & .\\
             & 52 & \(+\) & \(-0.22\) & \((-0.62,0.18)\) & \(9.02\)\\
             &  & \(-\) & \(-0.26\) & \((-0.70,0.18)\) & .\\
             & 78 & \(+\) & \(-0.59\) & \((-1.31,0.13)\) & \(8.87\)\\
             &  & \(-\) & \(-0.62\) & \((-1.41,0.18)\) & .\\
             & 104 & \(+\) & \(-1.15\) & \((-2.24,-0.07)\) & \(9.00\)\\
             &  & \(-\) & \(-1.14\) & \((-2.33,0.04)\) & .\\
    \(\gamma_2\{t\}\) & 26 & \(+\) & \(-0.05\) & \((-0.13,0.03)\) & \(-0.00\)\\
             &  & \(-\) & \(-0.05\) & \((-0.13,0.03)\) & .\\
             & 52 & \(+\) & \(-0.12\) & \((-0.37,0.13)\) & \(-0.00\)\\
             &  & \(-\) & \(-0.12\) & \((-0.37,0.13)\) & .\\
             & 78 & \(+\) & \(-0.21\) & \((-0.68,0.27)\) & \(-0.00\)\\
             &  & \(-\) & \(-0.21\) & \((-0.68,0.27)\) & .\\
             & 104 & \(+\) & \(-0.24\) & \((-1.02,0.53)\) & \(-0.00\)\\
             &  & \(-\) & \(-0.24\) & \((-1.02,0.53)\) & .\\
  \end{tabular}
  \medskip
  \begin{flushleft}
    {Type: fusion estimator (\(+\)) or RCT-only estimator (\(-\)); CI: confidence interval; Reduction: percentage reduction CI length, \(\%\).}
  \end{flushleft}
\end{table}

\section{Discussion}

In this work, we assume transportability of the conditional cause-specific hazard of the event of interest between the RCT population and the external control population.
We have considered estimation of the cumulative incidence functions and restricted mean times lost with external controls.
In fact, this assumption also allows us to derive estimators with improved precision for other estimands in competing risks analysis, including the average hazard with survival weights \citep{uno2023ratio} and separable effects \citep{stensrud2022separable}.
We comment in Supplementary Material \S\ref{sec:weaker-app} that weaker transportability assumptions for competing risks analysis can be difficult to interpret.

In practice, the risk of introducing bias to RCT data when incorporating external controls should be evaluated.
One approach is to carry out the analysis with the data fusion estimator for precision gain.
Then, post-hoc model diagnostics such as likelihood ratio tests or other omnibus tests may be performed to assess the validity of the transportability assumption.
A more principled approach may be to integrate the estimated bias to make an informed decision of whether the RCT-only estimator should be retained.
For instance, following \citet{yang2023elastic}, a test-then-pool estimator for the cumulative incidence \(\theta_{1}(0)\) can be constructed from the estimators \(\hat\theta_{1}(0)\) with and without external controls via a score test.
This is left for future work.

We focus on treatment policy estimands, which ignore treatment trajectories after randomization.
Consequently, we do not adjust for post-baseline variables.
However, by omitting these variables, we may fail to establish transportability of the cause-specific hazard.
In SUSTAIN-6 and LEADER, when subjects experienced non-fatal adverse events, they could receive drop-in medications.
If these decisions were based on different treatment guidelines and policies between the two populations, subjects with similar baseline characteristics might have rather different event rates.
This is a particular concern for using historical controls in RCTs.
Future research can focus on fusion estimates that allow for history beyond baseline.

\bibliography{./competing-risk-external-control.bib}

\newpage

\part{}
\begin{center}
  {\large Supplementary material for ``Improving precision of cumulative incidence estimates in randomized controlled trials with external controls''}
\end{center}

\bigskip

\renewcommand{\theequation}{S\arabic{equation}}%
\renewcommand{\thesection}{S\arabic{section}}%
\renewcommand{\thetable}{S\arabic{table}}%
\renewcommand{\thefigure}{S\arabic{figure}}%
\renewcommand{\thetheorem}{S\arabic{theorem}}%
\renewcommand{\theproposition}{S\arabic{proposition}}%
\renewcommand{\thelemma}{S\arabic{lemma}}%
\renewcommand{\theassumption}{S\arabic{assumption}}%
\renewcommand{\theremark}{S\arabic{remark}}%
\renewcommand{\thecorollary}{S\arabic{corollary}}%
\renewcommand{\thepage}{S\arabic{page}}

\setcounter{equation}{0}
\setcounter{section}{0}
\setcounter{table}{0}
\setcounter{figure}{0}
\setcounter{theorem}{0}
\setcounter{proposition}{0}
\setcounter{lemma}{0}
\setcounter{assumption}{0}
\setcounter{remark}{0}
\setcounter{corollary}{0}
\setcounter{page}{1}

\section{Notations on the observed data distribution}
\label{sec:notation-app}

In the main text, we make use of three models: the counterfactual data distribution, the uncensored data distribution, and the observed data distribution.
All three models encompass the population indicator \(D\), the baseline covariates \(X\), and the treatment \(DA\), whereas the counterfactual data distribution contains the potential outcomes \(\{T(1),T(0),J(1),J(0)\}\), the uncensored data distribution contains the uncensored event time and event type plus the censoring time \((T,J,C)\), and the observed data distribution contains the censored event time and event type \((\tilde{T},\tilde{J})\).
With Assumption~\ref{asn:causal}, we identify the causal parameters in the uncensored data distribution.
We now connect the observed data quantities to their uncensored counterparts using Assumption~\ref{asn:censoring}.

Recall the (observed) event counting process \(N_{j}(t)= \I(\tilde{T}\leq t,\tilde{J}=j)\) for \(j=1,2\).
Let \(N^{c}(t)=\I(\tilde{T}\leq t,\tilde{J}=0)\) be the censoring counting process.
Define \((\mathcal{F}_{t})_{t\in(0,\tau]}\) as the filtration in which the \(\sigma\)-algebra \(\mathcal{F}_{t}=\sigma\big[\big\{N_{1}(s),N_{2}(s),N^{c}(s),DA,X,D; 0< s\leq t\big\}\big]\) contains the observed information up to time \(t\) (inclusive).
The event counting process \(N_{j}(t)\) has a compensator such that
\begin{equation*}
\begin{aligned}
  \tilde{M}_{Dj}(t\mid A,X) &= N_{j}(t)-\tilde{\Lambda}_{Dj}(t\mid A,X),&\quad D&=1\\
  \tilde{M}_{Dj}(t\mid X) &= N_{j}(t)-\tilde{\Lambda}_{Dj}(t\mid X),&\quad D&=0,
\end{aligned}
\end{equation*}
is a martingale adapted to \((\mathcal{F}_{t})\).
Standard results in time-to-event analysis shows that the compensator satisfies a multiplicative hazard structure such that the increment of the compensator factorizes as
\begin{equation*}
\begin{aligned}
  \d\tilde{\Lambda}_{Dj}(t\mid A,X) &= \I(\tilde{T}\geq t)\d\tilde{\Alpha}_{Dj}(t\mid A,X), &\quad D&=1, \\
  \d\tilde{\Lambda}_{Dj}(t\mid X) &= \I(\tilde{T}\geq t)\d\tilde{\Alpha}_{Dj}(t\mid X), &\quad D&=0,
\end{aligned}
\end{equation*}
where
\begin{align*}
  \d\tilde{\Alpha}_{1j}(t\mid a,x) &= \frac{\d P(\tilde{T}\leq t,\tilde{J}=j\mid A=a,X=x,D=1)}{P(\tilde{T}\geq t\mid A=a,X=x,D=1)},\\
  \d\tilde{\Alpha}_{0j}(t\mid x) &= \frac{\d P(\tilde{T}\leq t,\tilde{J}=j\mid X=x,D=0)}{P(\tilde{T}\geq t\mid X=x,D=0)}.
\end{align*}

Consider the filtration \((\mathcal{G}_t)\) where \(\mathcal{G}_{t}= \sigma\big[\big\{N_{1}(s+),N_{2}(s+),N^{c}(s),DA,X,D;0< s\leq t\big\}\big]\).
The quantity associated with the observed censoring counting process
\[
  D\tilde{M}_{D}^c(t\mid A,X)+(1-D)\tilde{M}_{D}^c(t\mid X)
\]
is a martingale adapted to \((\mathcal{G}_{t})\), where \(\tilde{M}^{c}_D(t\mid A,X)= N^{c}(t)-\tilde{\Lambda}^{c}_D(t\mid A,X)\), \(\tilde{M}^{c}_D(t\mid X)= N^{c}(t)-\tilde{\Lambda}^{c}_D(t\mid X)\),
\[
  \d\tilde{\Lambda}_{D}^c(t\mid A,X)=\big\{\I(\tilde{T}\geq t,\tilde{J}=0)+\I(\tilde{T}>t,\tilde{J}\neq 0)\big\}\d\tilde{\Alpha}_{D}^c(t\mid A,X)
\]
and
\begin{align*}
  \d\tilde{\Alpha}_{1}^c(t\mid a,x) &= \frac{\d P(\tilde{T}\leq t,\tilde{J}=0\mid A=a,X=x,D=1)}{P[\{\tilde{T}\geq t,\tilde{J}=0\}\cup\{\tilde{T}>t,\tilde{J}\neq 0\}\mid A=a,X=x,D=1]},\\
  \d\tilde{\Alpha}_{0}^c(t\mid x) &= \frac{\d P(\tilde{T}\leq t,\tilde{J}=0\mid X=x,D=0)}{P[\{\tilde{T}\geq t,\tilde{J}=0\}\cup\{\tilde{T}>t,\tilde{J}\neq 0\}\mid X=x,D=0]}.
\end{align*}


Under Assumption~\ref{asn:censoring}, the cause-specific hazards defined on the uncensored data distribution are identifiable from the observed data with \(\tilde{\Alpha}_{dj}(t\mid a,x)=\Alpha_{dj}(t\mid a,x)\), and so is the censoring hazard with \(\tilde{\Alpha}_{d}^c(t\mid a,x)=\Alpha_{d}^c(t\mid a,x)\).
Therefore, the survival function of the composite event and the cumulative incidence function of event type \(j\) is subsequently identifiable in the observed data distribution as the product integral
\[
  S_d(t\mid a,x)=\{\Pi(\Alpha_{d1}+\Alpha_{d2})\}(t\mid a,x)=\big\{\Pi\big(\tilde{\Alpha}_{d1}+\tilde{\Alpha}_{d2}\big)\big\}(t\mid a,x)=\tilde{S}_d(t\mid a,x)
\]  
and the Lebesgue-Stieltjes integral
\[
  F_{dj}(t\mid a,x)= \int_{0}^{t}S_d(s\!-\!\mid a,x)\d\Alpha_{dj}(s\mid a,x)=\int_{0}^{t}\tilde{S}_d(s\!-\!\mid a,x)\d\tilde{\Alpha}_{dj}(s\mid a,x)=\tilde{F}_{dj}(t\mid a,x).
\]

\section{Proofs}
\label{sec:proof-app}

\subsection{Proof of Lemma~\ref{lem:eif}}

We show the efficient influence function of \(\theta_1(0)\) without Assumption~\ref{asn:discontinuity}.
Define
\begin{align*}
G_1(t\mid a,x) &= P(\tilde{T}>t\mid A=a,X=x,D=1),\\
G_0(t\mid x)&=P(\tilde{T}>t\mid X=x,D=0),\\
\tilde{F}_{1j}(t\mid a,x)&= \int_{0}^{t}\tilde{S}_1(s\!-\!\mid a,x)\d\tilde{\Alpha}_{1j}(s\mid a,x).
\end{align*}
Lemma~\ref{lem:eif} is a special case of the following result.

\begin{proposition}
  \label{ppn:eif-discontinuity}
  The efficient influence function of \(\theta_{1}(0)\) at \(P\in\mathcal{P}\) is
  \begin{align}
      \varphi_{1}(0)(O) &= \frac{D(1-A)\pi(X)}{\alpha}\int_{0}^{\tau}\frac{w_1(t\mid X)}{G_1(t\!-\!\mid A,X)}g_{11}(t\mid A,X)\d\tilde{M}_{11}(t\mid A,X) \nonumber\\
      &\hphantom{=}\quad + \frac{(1-D)\pi(X)}{\alpha}\int_{0}^{\tau}\frac{w_0(t\mid X)}{G_0(t\!-\!\mid X)}g_{11}(t\mid 0,X)\d\tilde{M}_{01}(t\mid X) \nonumber\\
      &\hphantom{=}\quad + \frac{D}{\alpha}\frac{(1-A)}{e_1(0\mid X)}\int_{0}^{\tau}\frac{1}{G_1(t\!-\!\mid A,X)}g_{21}(t\mid A,X)\d\tilde{M}_{12}(t\mid A,X)\nonumber\\
      &\hphantom{=}\quad + \frac{D}{\alpha}\big\{\tilde{F}_{11}(\tau\mid 0,X)-\theta_{1}(0)\big\},    \label{eqn:eif-discontinuity}
  \end{align}
  where
  \begin{align*}
    g_{k1}(t\mid A,X) &= \I(k=1)\tilde{S}_1(t\!-\!\mid A,X)-\frac{\tilde{F}_{11}(\tau\mid A,X)-\tilde{F}_{11}(t\mid A,X)}{1-\triangle\tilde{\Alpha}_{11}(t\mid A,X)-\triangle\tilde{\Alpha}_{12}(t\mid A,X)} \\
    w_{\dott}(t\mid X) &= \{1-\triangle\tilde{\Alpha}_{11}(t\mid 0,X)\}g_{11}(t\mid 0,X)\bigg\{\frac{1-\pi(X)}{G_1(t\!-\!\mid 0,X)}+\frac{\pi(X)e_1(0\mid X)}{G_0(t\!-\!\mid X)}\bigg\}\\
    w_1(t\mid X) &= \frac{1}{w_{\dott}(t\mid X)}\bigg\{\frac{1-\triangle\tilde{\Alpha}_{11}(t\mid 0,X)}{G_0(t\!-\!\mid X)}g_{11}(t\mid 0,X) \\
    &\hphantom{= \frac{1}{w_{\dott}(t\mid X)}\bigg\{}\quad+\frac{\triangle\tilde{\Alpha}_{12}(t\mid 0,X)}{G_1(t\!-\!\mid 0,X)}\frac{1-\pi(X)}{\pi(X)e_1(0\mid X)}g_{21}(t\mid 0,X)\bigg\}\\
    w_0(t\mid X) &= \frac{1}{w_{\dott}(t\mid X)}\bigg\{\frac{1-\triangle\tilde{\Alpha}_{11}(t\mid 0,X)}{G_1(t\!-\!\mid 0,X)}g_{11}(t\mid 0,X)-\frac{\triangle\tilde{\Alpha}_{12}(t\mid 0,X)}{G_1(t\!-\!\mid 0,X)}g_{21}(t\mid 0,X)\bigg\}.
  \end{align*}
  The efficient influence function of \(\theta_{2}(0)\) at \(P\in\mathcal{P}\) is
  \begin{align}
      \varphi_{2}(0)(O) &= \frac{D(1-A)\pi(X)}{\alpha}\int_{0}^{\tau}\frac{w_1(t\mid X)}{G_1(t\!-\!\mid A,X)}g_{12}(t\mid A,X)\d\tilde{M}_{11}(t\mid A,X) \nonumber\\
      &\hphantom{=}\quad + \frac{(1-D)\pi(X)}{\alpha}\int_{0}^{\tau}\frac{w_0(t\mid X)}{G_0(t\!-\!\mid X)}g_{12}(t\mid 0,X)\d\tilde{M}_{01}(t\mid X) \nonumber\\
      &\hphantom{=}\quad + \frac{D}{\alpha}\frac{(1-A)}{e_1(0\mid X)}\int_{0}^{\tau}\frac{1}{G_1(t\!-\!\mid A,X)}g_{22}(t\mid A,X)\d\tilde{M}_{12}(t\mid A,X)\nonumber\\
      &\hphantom{=}\quad + \frac{D}{\alpha}\big\{\tilde{F}_{12}(\tau\mid 0,X)-\theta_{2}(0)\big\}. 
    \label{eqn:eif-discontinuity-2}
    \end{align}
\end{proposition}

\begin{proof}
  We define the observed data distribution for the censoring time as \(Q_{10}(t\mid a,x)= P(\tilde{T}\leq t,\tilde{J}=0\mid A=a,X=x,D=1)\) and \(Q_{00}(t\mid x)= P(\tilde{T}\leq t,\tilde{J}=0\mid X=x,D=0)\) for \(a\in\{0,1\}\).
  The observed data density can be factorized as
  \begin{multline*}
    \d P(\tilde{T},\tilde{J},A,X,D) = \{\d Q_{1\tilde{J}}(\tilde{T}\mid A,X)e_1(A\mid X)\pi(X)\}^{D}[\d Q_{0\tilde{J}}(\tilde{T}\mid X)\{1-\pi(X)\}]^{1-D}\d P(X).
  \end{multline*}
  Consider the parametric submodel for the observed data density for the event time and event type:
  \begin{align*}
    \d Q_{1k}(t\mid a,x;\varepsilon) &= \d Q_{1k}(t\mid a,x)\{1+\varepsilon h_{1}(t,k,a,x)\}, \\
    \d Q_{0k}(t\mid x;\varepsilon) &= \d Q_{0k}(t\mid x)\{1+\varepsilon h_{0}(t,k,x)\}, \\
  \end{align*}
  for \(a\in\{0,1\}\) and \(k\in\{0,1,2\}\), where \(h_{1}(\tilde{T},\tilde{J},A,X)\) and \(h_{0}(\tilde{T},\tilde{J},X)\) are functions with finite variance that satisfy \(E\{h_{1}(\tilde{T},\tilde{J},A,X)\mid A,X,D=1\}=0\) and \(E\{h_{0}(\tilde{T},\tilde{J},X)\mid X,D=0\}=0\).
  The submodel must further obey the restriction \(\d\tilde{\Alpha}_{11}(t\mid 0,x;\varepsilon) = \d\tilde{\Alpha}_{01}(t\mid x;\varepsilon)\) for \(t\in(0,\tau]\), or equivalently,
  \begin{equation}
    \label{eqn:restriction-submodel}
    \frac{\d Q_{11}(t\mid 0,x;\varepsilon)}{G_1(t\!-\!\mid 0,x;\varepsilon)} = \frac{\d Q_{01}(t\mid x;\varepsilon)}{G_0(t\!-\!\mid x;\varepsilon)}.
  \end{equation}
  The Gateaux derivative of the cumulative hazard increment \(\d \tilde\Alpha_{1k}(t\mid a,x;\varepsilon)\) for \(k\in\{0,1,2\}\) is
  \begin{align*}
    \MoveEqLeft\frac{\d}{\d \varepsilon}\d \tilde\Alpha_{1k}(t\mid a,x;\varepsilon)\bigg\vert_{\varepsilon=0}\\
    &= \frac{\d}{\d \varepsilon}\frac{\d Q_{1k}(t\mid a,x;\varepsilon)}{G_1(t\!-\!\mid a,x;\varepsilon)}\bigg\vert_{\varepsilon=0} \\
    &= \frac{1}{G_1(t\!-\!\mid a,x)}\frac{\d}{\d \varepsilon}\d Q_{1k}(t\mid a,x;\varepsilon)\bigg\vert_{\varepsilon=0} - \frac{\d Q_{1k}(t\mid a,x)}{G_1^2(t\!-\!\mid a,x)}\frac{\d}{\d \varepsilon}G_1(t\!-\!\mid a,x;\varepsilon)\bigg\vert_{\varepsilon=0} \\
    &= \frac{1}{G_1(t\!-\!\mid a,x)}\frac{\d}{\d \varepsilon}\d Q_{1k}(t\mid a,x;\varepsilon)\bigg\vert_{\varepsilon=0} \\
    &\hphantom{==}\quad - \frac{\d Q_{1k}(t\mid a,x)}{G_1^2(t\!-\!\mid a,x)}\frac{\d}{\d \varepsilon}\int_{u\in[t,\infty)}\sum_{j\in\{0,1,2\}}\d Q_{1j}(u\mid a,x;\varepsilon)\bigg\vert_{\varepsilon=0}\\
    &= \d\tilde\Alpha_{1k}(t\mid a,x)\bigg\{h_1(t,k,a,x) - \sum_{j\in\{0,1,2\}}\int_{u\in [0,\infty)}h_1(u,j,a,x)\frac{\d Q_{1j}(u\mid a,x)}{G_1(t\!-\!\mid a,x)}\bigg\}.
  \end{align*}
  Similarly, for \(k\in\{0,1,2\}\), we have
  \[
    \frac{\d}{\d \varepsilon}\d \tilde\Alpha_{0k}(t\mid x;\varepsilon)\bigg\vert_{\varepsilon=0}
    = \d\tilde\Alpha_{0k}(t\mid x)\bigg\{h_0(t,k,x) - \sum_{j\in\{0,1,2\}}\int_{u\in [t,\infty)}h_0(u,j,x)\frac{\d Q_{0j}(u\mid x)}{G_0(t\!-\!\mid x)}\bigg\}.
  \]
  Therefore, differentiating both sides of \eqref{eqn:restriction-submodel} with respect to \(\varepsilon\) and evaluating at zero, the restriction on the scores of the hazards is
  \begin{multline*}
    \d\tilde{\Alpha}_{11}(t\mid 0,x)\bigg\{h_{1}(t,1,0,x)-\int_{u\in[t,\infty)}\sum_{k\in\{0,1,2\}}h_{1}(u,k,0,x)\frac{\d Q_{1k}(u\mid 0,x)}{G_1(t\!-\!\mid 0,x)}\bigg\} \\
    =\d\tilde{\Alpha}_{01}(t\mid x)\bigg\{h_{0}(t,1,x)-\int_{u\in[t,\infty)}\sum_{k\in\{0,1,2\}}h_{0}(u,k,x)\frac{\d Q_{0k}(u\mid x)}{G_0(t\!-\!\mid x)}\bigg\}.
  \end{multline*}
  With some algebra, the score restriction can be expressed as a conditional expectation restriction:
  \begin{multline}
    \label{eqn:restriction-expectation}
    \E\bigg[h_{1}(\tilde{T},\tilde{J},A,X)\frac{\d \tilde{M}_{11}(t\mid A,X)}{G_1(t\!-\!\mid A,X)}\,\bigg\vert\,A=0,X,D=1\bigg] \\
    = \E\bigg[h_{0}(\tilde{T},\tilde{J},X)\frac{\d \tilde{M}_{01}(t\mid X)}{G_0(t\!-\!\mid X)}\,\bigg\vert\,X,D=0\bigg].
  \end{multline}
  For the rest of the components in the observed data density, we also choose appropriate perturbation functions, or score functions, such that \(\d P(\tilde{T},\tilde{J},A,X,D;\varepsilon)\) equals \(\d P(\tilde{T},\tilde{J},A,X,D)\) when \(\varepsilon=0\).
  The closed linear subspace of all possible choices of these perturbation functions is the tangent space of the model \(\mathcal{P}\) at \(P\), which is
  \[
    \dot{\mathcal{P}} = \dot{\mathcal{P}}_{1} \oplus \dot{\mathcal{P}}_{2} \oplus \dot{\mathcal{P}}_{3} \oplus \dot{\mathcal{P}}_{4},
  \]
  where \(\mathcal{V}_1\oplus\mathcal{V}_2\) denotes the direct sum of vector spaces \(\mathcal{V}_1\) and \(\mathcal{V}_2\),
  \begin{align*}
    \dot{\mathcal{P}}_{1} &= \big\{Dh_{1}(\tilde{T},\tilde{J},A,X)+(1-D)h_{0}(\tilde{T},\tilde{J},X):\E\{h_{1}(\tilde{T},\tilde{J},A,X)\mid A,X,D=1\}=0, \\
           &\hphantom{=\big\{}\quad \E\{h_{0}(\tilde{T},\tilde{J},X)\mid X,D=0\}=0,h_{1}(\tilde{T},\tilde{J},A,X)\text{ and }h_{0}(\tilde{T},\tilde{J}, X)\text{ satisfy \eqref{eqn:restriction-expectation}}\big\}, \\
    \dot{\mathcal{P}}_{2} &= \big\{Dh_1(A,X):\E\{h_1(A,X)\mid X,D=1\}=0\big\}, \\
    \dot{\mathcal{P}}_{3} &= \big\{h(D,X):\E\{h(D,X)\mid X\}=0\big\}, \\
    \dot{\mathcal{P}}_{4} &= \big\{h(X):\E\{h(X)\}=0\big\}.
  \end{align*}
  The decomposition of the tangent space follows from the product structure of the observed data likelihood.
  Differentiating the target parameter along some submodel \(\{P_{\varepsilon}\}\) with score function
  \[
    h(O)=Dh_{1}(\tilde{T},\tilde{J},A,X)+(1-D)h_{0}(\tilde{T},\tilde{J},X)+Dh(A,X)+h(D,X)+h(X),
  \]
  we have
  \begin{align}
    \MoveEqLeft \frac{\d }{\d\varepsilon}\theta_{1}(0;\varepsilon)\bigg\vert_{\varepsilon=0} \nonumber\\
    &= \frac{\d }{\d\varepsilon}\int_{\mathcal{X}}\int_{0}^{\tau}\d\tilde{\Alpha}_{11}(t\mid 0,x;\varepsilon)\tilde{S}_1(t\!-\!\mid 0,x;\varepsilon)\d P(x\mid D=1;\varepsilon)\bigg\vert_{\varepsilon=0} \nonumber\\
    &= \int_{\mathcal{X}}\int_{0}^{\tau}\frac{\d }{\d\varepsilon}\d\tilde{\Alpha}_{11}(t\mid 0,x;\varepsilon)\bigg\vert_{\varepsilon=0}\tilde{S}_1(t\!-\!\mid 0,x)\d P(x\mid D=1) \label{eqn:derivative-1}\\
    &\hphantom{==}\begin{multlined}[b][0.8\textwidth]
    \quad  -\int_{\mathcal{X}}\int_{0}^{\tau}\int_{0}^{t\!-\!}\frac{\dfrac{\d }{\d\varepsilon}\d\tilde{\Alpha}_{11}(s\mid 0,x;\varepsilon)\bigg\vert_{\varepsilon=0}+\dfrac{\d}{\d\varepsilon}\d\tilde{\Alpha}_{12}(s\mid 0,x;\varepsilon)\bigg\vert_{\varepsilon=0}}{1-\triangle\tilde{\Alpha}_{11}(s\mid 0,x)-\triangle\tilde{\Alpha}_{12}(s\mid 0,x)}\\
    \d\tilde{F}_{11}(t\mid 0,x)\d P(x\mid D=1)
    \end{multlined}\label{eqn:derivative-2}\\
    &\hphantom{=}\quad +\int_{\mathcal{X}}\tilde{F}_{11}(\tau\mid 0,x)\frac{\d}{\d\varepsilon}\d P(x\mid D=1;\varepsilon)\bigg\vert_{\varepsilon=0}\label{eqn:derivative-3},
  \end{align}
  where the outermost integral is over the set \(\mathcal{X}= \mathcal{X}_{1}\cup \mathcal{X}_{0}\).

  We proceed by analyzing the terms separately.
  First, the term \eqref{eqn:derivative-1} is
  \begin{align*}
    \MoveEqLeft\int_{\mathcal{X}}\int_{0}^{\tau}h_{1}(t,1,0,x)\d\tilde{\Alpha}_{11}(t\mid 0,x)\tilde{S}_1(t\!-\!\mid 0,x)\d P(x\mid D=1)\\
      &-\begin{multlined}[t][.9\textwidth]
      \int_{\mathcal{X}}\int_{0}^{\tau}\int_{u\in[t,\infty)}\sum_{k\in\{0,1,2\}}h_{1}(u,k,0,x)\frac{\d Q_{1k}(u\mid 0,x)}{G_1(t\!-\!\mid 0,x)}\\
      \d\tilde{\Alpha}_{11}(t\mid 0,x)\tilde{S}_1(t\!-\!\mid 0,x)\d P(x\mid D=1)
      \end{multlined}\\
      &= \begin{multlined}[t][.9\textwidth]
      \int_{\mathcal{X}}E\bigg[h_{1}(\tilde{T},\tilde{J},A,X)\int_{0}^{\tau}\frac{\tilde{S}_1(t\!-\!\mid A,X)}{G_1(t\!-\!\mid A,X)}\d\tilde{M}_{11}(t\mid A,X)\,\bigg\vert\,A=0,X=x,D=1\bigg]\\
      \d P(x\mid D=1).
      \end{multlined}
  \end{align*}
  The term \eqref{eqn:derivative-2} is a sum of two terms, which for \(k\in\{1,2\}\) can be seen to be
  \begin{align*}
    \MoveEqLeft -\int_{\mathcal{X}}\int_{0}^{\tau}\int_{s\in(0,t)}\frac{\dfrac{\d }{\d\varepsilon}\d\tilde{\Alpha}_{1k}(s\mid 0,x;\varepsilon)\bigg\vert_{\varepsilon=0}}{1-\triangle\tilde{\Alpha}_{11}(s\mid 0,x)-\triangle\tilde{\Alpha}_{12}(s\mid 0,x)}\d\tilde{F}_{11}(t\mid 0,x)\d P(x\mid D=1)\\
    &= -\int_{\mathcal{X}}\int_{0}^{\tau}\int_{s\in(0,t)}\frac{h_{1}(s,k,0,x)\d\tilde{\Alpha}_{1k}(s\mid 0,x)}{1-\triangle\tilde{\Alpha}_{11}(s\mid 0,x)-\triangle\tilde{\Alpha}_{12}(s\mid 0,x)}\d\tilde{F}_{11}(t\mid 0,x)\d P(x\mid D=1)\\
    &\hphantom{==}\begin{multlined}[t][0.9\textwidth]
    \quad +\int_{\mathcal{X}}\int_{0}^{\tau}\int_{s\in(0,t)}\frac{\big\{\int_{u\in[s,\infty)}\sum_{j\in\{0,1,2\}}h_{1}(u,j,0,x)\d Q_{1j}(u\mid 0,x)\big\}}{G_1(s\!-\!\mid 0,x)\{1-\triangle\tilde{\Alpha}_{11}(s\mid 0,x)-\triangle\tilde{\Alpha}_{12}(s\mid 0,x)\}}\\
    \d\tilde{\Alpha}_{1k}(s\mid 0,x)\d\tilde{F}_{11}(t\mid 0,x)\d P(x\mid D=1)
    \end{multlined}\\
    &= -\int_{\mathcal{X}}\int_{s\in(0,\tau)}\int_{t\in(s,\tau]}\d\tilde{F}_{11}(t\mid 0,x)\frac{h_{1}(s,k,0,x)\d\tilde{\Alpha}_{1k}(s\mid 0,x)}{1-\triangle\tilde{\Alpha}_{11}(s\mid 0,x)-\triangle\tilde{\Alpha}_{12}(s\mid 0,x)}\d P(x\mid D=1)\\
    &\phantom{=}\ \begin{multlined}[t][0.9\textwidth]
    +\int_{\mathcal{X}}\int_{s\in(0,\tau)}\int_{t\in(s,\tau]}\d\tilde{F}_{11}(t\mid 0,x)\\
    \frac{\big\{\int_{u\in[s,\infty)}\sum_{j\in\{0,1,2\}}h_{1}(u,j,0,x)\d Q_{1j}(u\mid 0,x)\big\}}{G_1(s\!-\!\mid 0,x)\{1-\triangle\tilde{\Alpha}_{11}(s\mid 0,x)-\triangle\tilde{\Alpha}_{12}(s\mid 0,x)\}}\\
    \d\tilde{\Alpha}_{1k}(s\mid 0,x)\d P(x\mid D=1)
    \end{multlined}\\
    &=\begin{multlined}[t][0.9\textwidth]
    -\int_{\mathcal{X}}\E\bigg[h_{1}(\tilde{T},\tilde{J},A,X)\int_{s\in(0,\tau)}\frac{\tilde{F}_{11}(\tau\mid A,X)-\tilde{F}_{11}(s\mid A,X)}{1-\triangle\tilde{\Alpha}_{11}(s\mid A,X)-\triangle\tilde{\Alpha}_{12}(s\mid A,X)}\\
    \frac{\d N_k(s)}{G_1(s\!-\!\mid A,X)}\biggm\vert A=0,X=x,D=1\bigg]\d P(x\mid D=1)
    \end{multlined}\\
    &\hphantom{==}\begin{multlined}[t][0.9\textwidth]
    \quad+\int_{\mathcal{X}}\E\bigg[h_1(\tilde{T},\tilde{J},A,X)\int_{s\in(0,\tau)}\frac{\tilde{F}_{11}(\tau\mid A,X)-\tilde{F}_{11}(s\mid A,X)}{1-\triangle\tilde{\Alpha}_{11}(s\mid A,X)-\triangle\tilde{\Alpha}_{12}(s\mid A,X)}\\
    \frac{\I(\tilde{T}\geq s)\d\tilde{\Alpha}_{1k}(s\mid A,X)}{G_1(s\!-\!\mid A,X)}\biggm\vert A=0,X=x,D=1\bigg]\d P(x\mid D=1)
    \end{multlined}\\
    &= \begin{multlined}[t][0.9\textwidth]
    -\int_{\mathcal{X}}E\bigg[h_{1}(\tilde{T},\tilde{J},A,X)\int_{0}^{\tau}\frac{\tilde{F}_{11}(\tau\mid A,X)-\tilde{F}_{11}(t\mid A,X)}{G_1(t\!-\!\mid A,X)\{1-\triangle\tilde{\Alpha}_{11}(t\mid A,X)-\triangle\tilde{\Alpha}_{12}(t\mid A,X)\}}\\
    \d\tilde{M}_{1k}(t\mid A,X)\,\bigg\vert\,A=0,X=x,D=1\bigg]\d P(x\mid D=1).
    \end{multlined}
  \end{align*}
  The last term \eqref{eqn:derivative-3} is
  \[
    \int_{\mathcal{X}} \{\tilde{F}_{11}(\tau\mid 0,x)-\theta_{1}(0)\}\{h(1,x)+h(x)\}\d P(x\mid D=1).
  \]
  Collecting the terms yields that
  \begin{align*}
    \MoveEqLeft\frac{\d }{\d\varepsilon}\theta_{1}(0;\varepsilon)\bigg\vert_{\varepsilon=0} \\
    &= \int_{\mathcal{X}}E\bigg[\int_{0}^{\tau}\bigg\{g_{11}(t\mid A,X)\frac{\d\tilde{M}_{11}(t\mid A,X)}{G_1(t\!-\!\mid A,X)}+g_{21}(t\mid A,X)\frac{\d\tilde{M}_{12}(t\mid A,X)}{G_1(t\!-\!\mid A,X)}\bigg\} \\
    &\hphantom{= \int_{\mathcal{X}}E\bigg[}\qquad h_{1}(\tilde{T},\tilde{J},A,X)\,\bigg\vert\,A=0,X=x,D=1\bigg]\d P(x\mid D=1) \\
    &\hphantom{=}\quad +\int_{\mathcal{X}} \{\tilde{F}_{11}(\tau\mid 0,x)-\theta_{1}(0)\}\{h(1,x)+h(x)\}\d P(x\mid D=1),\\
    \intertext{and by replacing integrals with expectations, we have}
    &= E\bigg[\frac{D(1-A)}{\alpha e_1(A\mid X)}h_{1}(\tilde{T},\tilde{J},A,X) \\
    &\hphantom{= E\bigg[}\qquad\int_{0}^{\tau}\bigg\{g_{11}(t\mid A,X)\frac{\d\tilde{M}_{11}(t\mid A,X)}{G_1(t\!-\!\mid A,X)}+g_{12}(t\mid A,X)\frac{\d\tilde{M}_{12}(t\mid A,X)}{G_1(t\!-\!\mid A,X)}\bigg\}\bigg] \\
    &\hphantom{=}\quad +E\bigg[\frac{D}{\alpha}\big\{\tilde{F}_{11}(\tau\mid 0,X)-\theta_{1}(0)\big\}\{h(D,X)+h(X)\}\bigg].
  \end{align*}
  In the following we show that the function \(\varphi_{1}(0)(O)\) displayed in \eqref{eqn:eif-discontinuity} is indeed a gradient of the parameter \(\theta_{1}(0)\) by verifying that
  \[
    \E\{\varphi_{1}(0)(O)h(O)\}=\frac{\d }{\d\varepsilon}\theta_{1}(0;\varepsilon)\bigg\vert_{\varepsilon=0}
  \]
  for the score \(h(O)\) of an arbitrary submodel \(\{P_{\varepsilon}\}\subset\mathcal{P}\).
  The inner product
  \begin{align*}
    \MoveEqLeft\E\{\varphi_{1}(0)(O)h(O)\} &\\
    &= E\bigg\{\frac{(1-A)\pi(X)}{\alpha}\int_{0}^{\tau}\frac{Dw_1(t\mid X)}{G_1(t\!-\!\mid A,X)}g_{11}(t\mid A,X)\d\tilde{M}_{11}(t\mid A,X)h_{1}(\tilde{T},\tilde{J},A,X)\bigg\}\\
    &\hphantom{=}\quad +E\bigg\{\frac{\pi(X)}{\alpha}\int_{0}^{\tau}\frac{(1-D)w_0(t\mid X)}{G_0(t\!-\!\mid X)}g_{11}(t\mid 0,X)\d\tilde{M}_{01}(t\mid X)h_{0}(\tilde{T},\tilde{J},X)\bigg\}\\
    &\hphantom{=}\quad +E\bigg\{\frac{(1-A)}{\alpha e_1(A\mid X)}\int_{0}^{\tau}\frac{D}{G_1(t\!-\!\mid A,X)}g_{21}(t\mid A,X)\d\tilde{M}_{12}(t\mid A,X)h_{1}(\tilde{T},\tilde{J},A,X)\bigg\} \\
    &\hphantom{=}\quad +E\bigg[\frac{D}{\alpha}\big\{\tilde{F}_{11}(\tau\mid 0,X)-\theta_{1}(0)\big\}\{h(D,X)+h(X)\}\bigg].
  \end{align*}
  The first two terms of the right hand side of the equation can be simplified by \eqref{eqn:restriction-expectation}, so that they sum up to
  \begin{align*}
    \MoveEqLeft E\bigg\{\frac{D(1-A)\pi(X)}{\alpha}\int_{0}^{\tau}w_1(t\mid X)g_{11}(t\mid A,X)\frac{\d\tilde{M}_{11}(t\mid A,X)}{G_1(t\!-\!\mid A,X)}h_{1}(\tilde{T},\tilde{J},A,X)\bigg\} &\\
                                                                                                                                                                                            &\quad+E\bigg[\frac{(1-D)\pi(X)}{\alpha}E\bigg\{\int_{0}^{\tau}w_0(t\mid X)g_{11}(t\mid A,X)\frac{\d\tilde{M}_{11}(t\mid A,X)}{G_1(t\!-\!\mid A,X)}\\
    &\quad\hphantom{+E\bigg[\frac{(1-D)\pi(X)}{\alpha}E\bigg\{\int_{0}^{\tau}}\qquad h_{1}(\tilde{T},\tilde{J},A,X)\,\bigg\vert\, A=0,X,D=1\bigg\}\bigg]\\
    &= E\bigg\{\frac{\pi(X)}{\alpha}E\bigg(\int_{0}^{\tau}[\pi(X)e_1(A\mid X)w_1(t\mid X)+\{1-\pi(X)\}w_0(t\mid X)] \\
    &\hphantom{= E\bigg\{\frac{\pi(X)}{\alpha}E\bigg(\int_{0}^{\tau}}\qquad g_{11}(t\mid A,X)\frac{\d\tilde{M}_{11}(t\mid A,X)}{G_1(t\!-\!\mid A,X)}h_{1}(\tilde{T},\tilde{J},A,X)\,\bigg\vert\,A=0,X,D=1\bigg)\bigg\}\\
    &= E\bigg\{\frac{D(1-A)}{\alpha e_1(A\mid X)}h_{1}(\tilde{T},\tilde{J},A,X)\int_{0}^{\tau}g_{11}(t\mid A,X)\frac{\d\tilde{M}_{11}(t\mid A,X)}{G_1(t\!-\!\mid A,X)}\bigg\},
  \end{align*}
  where in the last step we used the identity
  \[
  \pi(x)e_1(0\mid x)w_1(t\mid x)+\{1-\pi(x)\}w_0(t\mid x)=1.
  \]
  This can be established by direct calculation:
  \begin{align*}
  \MoveEqLeft \pi(x)e_1(0\mid x)w_1(t\mid x)+\{1-\pi(x)\}w_0(t\mid x) \\
  &= \frac{1}{w_{\dott}(t\mid x)}\bigg[\{1-\triangle\tilde{\Alpha}_{11}(t\mid 0,x)\}g_{11}(t\mid 0,x)\bigg\{\frac{1-\pi(x)}{G_1(t\!-\!\mid 0,x)}+\frac{\pi(x)e_1(0\mid x)}{G_0(t\!-\!\mid x)}\bigg\}\bigg]\\
  &= 1.
  \end{align*}
    We have established that \(\varphi_{1}(0)(O)\) is indeed a gradient of \(\theta_{1}(0)\).
    
  In order to show that \(\varphi_1(0)\) is the efficient influence function, it remains to ascertain \(\varphi_{1}(0)(O)\) itself is a score of the model at \(P\).
  To proceed further, we decompose \(\varphi_{1}(0)(O)\) into the following functions:
  \begin{align*}
    h_{1}^{*}(\tilde{T},\tilde{J},A,X) &= \frac{\pi(X)(1-A)}{\alpha}\int_{0}^{\tau}\frac{w_1(t\mid X)g_{11}(t\mid A,X)}{G_1(t\!-\!\mid A,X)}\d\tilde{M}_{11}(t\mid A,X)\\
    &\hphantom{=}\quad + \frac{(1-A)}{\alpha e_1(A\mid X)}\int_{0}^{\tau}\frac{g_{21}(t\mid A,X)}{G_1(t\!-\!\mid A,X)}\d\tilde{M}_{12}(t\mid A,X),\\
    h_{0}^{*}(\tilde{T},\tilde{J},X) &= \frac{\pi(X)}{\alpha}\int_{0}^{\tau}\frac{w_0(t\mid X)g_{11}(t\mid 0,X)}{G_0(t\!-\!\mid X)}\d\tilde{M}_{01}(t\mid X),\\
    h^{*}(D,X) &= \frac{D-\pi(X)}{\alpha}\big\{\tilde{F}_{11}(\tau\mid 0,X)-\theta_{1}(0)\big\},\\
    h^{*}(X) &= \frac{\pi(X)}{\alpha}\big\{\tilde{F}_{11}(\tau\mid 0,X)-\theta_{1}(0)\big\},
  \end{align*}
  such that
  \[
  \varphi_{1}(0)(O) = Dh_{1}^{*}(\tilde{T},\tilde{J},A,X) + (1-D)h_{0}^{*}(\tilde{T},\tilde{J},X) + h^{*}(D,X)+h^{*}(X).
  \]
  It is trivial to show that \(h^{*}(D,X)\) and \(h^{*}(X)\) are valid scores by noting that \(E\{h^{*}(D,X)\mid X\}=0\) and \(E\{h^{*}(X)\}=0\), so \(h^{*}(D,X)\in\dot{\mathcal{P}}_{3}\) and \(h^{*}(X)\in\dot{\mathcal{P}}_{4}\).
  Since \(h_{1}^{*}(\tilde{T},\tilde{J},A,X)\) and \(h_{0}^{*}(\tilde{T},\tilde{J},X)\) are zero-mean martingales adapted to the filtration \((\mathcal{F}_{t})\) at time \(\tau\), it is also clear that
  \[
  E\{h_{1}^{*}(\tilde{T},\tilde{J},A,X)\mid A,X,D=1\}=0,\quad E\{h_{0}^{*}(\tilde{T},\tilde{J},X)\mid X,D=0\}=0.
  \]

  We will now verify that \(h_0^*\) and \(h_1^*\) fulfill \eqref{eqn:restriction-expectation} in the integral form.
  We do so by computing both sides of \eqref{eqn:restriction-expectation} substituting \(h_1^*\) for \(h_1\) and \(h_0^*\) for \(h_0\).
  For any \(0< t\leq \tau\), we calculate the conditional expectation
  \begin{align}
    \MoveEqLeft E\bigg\{h_{1}^{*}(\tilde{T},\tilde{J},A,X)\int_{0}^{t}\frac{\d \tilde{M}_{11}(s\mid A,X)}{G_1(s\!-\!\mid A,X)}\,\bigg\vert\,A=0,X=x,D=1\bigg\} &\nonumber\\
    &= \frac{1}{\alpha}E\bigg[\int_{0}^{t}\frac{\d \tilde{M}_{11}(s\mid A,X)}{G_1(s\!-\!\mid A,X)}\int_{0}^{t}\bigg\{\pi(X)\frac{w_1(s\mid X)g_{11}(s\mid A,X)}{G_1(s\!-\!\mid A,X)}\d\tilde{M}_{11}(s\mid A,X)& \nonumber\\
    &\qquad + \frac{1}{e_1(A\mid X)}\frac{g_{21}(s\mid A,X)}{G_1(s\!-\!\mid A,X)}\d\tilde{M}_{12}(s\mid A,X)\bigg\}\,\bigg\vert\,A=0,X=x,D=1\bigg],& \nonumber\\
    \intertext{and by the property of the martingale product and the corresponding predictable covariation process \citepsuppmat[Theorem 2.3.4]{fleming1991counting}, the term above is}
    &= \frac{1}{\alpha}E\bigg[\bigg\langle\int_{0}^{t}\frac{\d \tilde{M}_{11}(s\mid A,X)}{G_1(s\!-\!\mid A,X)},\int_{0}^{t}\bigg\{\pi(X)\frac{w_1(s\mid X)g_{11}(s\mid A,X)}{G_1(s\!-\!\mid A,X)}\d\tilde{M}_{11}(s\mid A,X) &\nonumber\\
    &\qquad \frac{1}{e_1(A\mid X)}\frac{g_{21}(s\mid A,X)}{G_1(s\!-\!\mid A,X)}\d\tilde{M}_{12}(s\mid A,X)\bigg\}\bigg\rangle\,\bigg\vert\,A=0,X=x,D=1\bigg] \nonumber\\
    &= \frac{1}{\alpha}E\bigg\{\int_{0}^{t}\frac{w_1(s\mid X)g_{11}(s\mid A,X)}{\{G_1(s\!-\!\mid A,X)\}^{2}}\pi(X)\d\langle\tilde{M}_{11}\rangle(s\mid A,X)\bigg\vert\,A=0,X=x,D=1\bigg\}& \nonumber\\
    &\hphantom{=}\quad +\frac{1}{\alpha}E\bigg\{\int_{0}^{t}\frac{g_{21}(s\mid A,X)}{\{G_1(s\!-\!\mid A,X)\}^{2}}\frac{\d\langle\tilde{M}_{11},\tilde{M}_{12}\rangle(s\mid A,X)}{e_1(A\mid X)}\bigg\vert\,A=0,X=x,D=1\bigg\},& \nonumber\\
     \intertext{and evaluating the predictable (co-)variation processes with the help of Theorem 2.6.1 of \citetsuppmat{fleming1991counting}, finally gives}
    &= \frac{1}{\alpha}\int_{0}^{t}\frac{\d\tilde{\Alpha}_{11}(s\mid 0,x)}{G_1(s\!-\!\mid 0,x)}&\nonumber\\
    &\qquad \bigg\{\pi(x)w_1(s\mid x)g_{11}(s\mid 0,x)\{1-\triangle\tilde{\Alpha}_{11}(s\mid 0,x)\}-\frac{g_{21}(s\mid 0,x)}{e_1(0\mid x)}\triangle\tilde{\Alpha}_{12}(s\mid 0,x)\bigg\}.\label{eqn:restriction-check-1}
  \end{align}
  On the other hand, 
  \begin{flalign}
    \MoveEqLeft E\bigg\{h_{0}^{*}(\tilde{T},\tilde{J},X)\int_{0}^{t}\frac{\d \tilde{M}_{01}(s\mid X)}{G_0(s\!-\!\mid X)}\,\bigg\vert\,X=x,D=0\bigg\} &\nonumber\\
    &= \frac{1}{\alpha}E\bigg\{\int_{0}^{t}\frac{\d \tilde{M}_{01}(s\mid X)}{G_0(s\!-\!\mid X)}\int_{0}^{t}\pi(X)\frac{w_0(s\mid X)g_{11}(s\mid 0,X)}{G_0(s\!-\!\mid X)}\d\tilde{M}_{01}(s\mid X)\,\bigg\vert\,X=x,D=0\bigg\}& \nonumber\\
    &= \frac{1}{\alpha}E\bigg\{\int_{0}^{t}\pi(X)\frac{w_0(s\mid X)g_{11}(s\mid 0,X)}{\{G_0(s\!-\!\mid X)\}^{2}}\d\langle\tilde{M}_{01}\rangle(s\mid X)\,\bigg\vert\,X=x,D=0\bigg\}& \nonumber\\
    &= \frac{1}{\alpha}\int_{0}^{t}\frac{\d\tilde{\Alpha}_{11}(s\mid 0,x)}{G_0(s\!-\!\mid X)}\pi(x)w_0(s\mid x)g_{11}(s\mid 0,x)\{1-\triangle\tilde{\Alpha}_{11}(s\mid 0,x)\}.& \label{eqn:restriction-check-2}
  \end{flalign}

    The restriction \eqref{eqn:restriction-expectation} holds if \eqref{eqn:restriction-check-1} and \eqref{eqn:restriction-check-2} are equal.
    Therefore, we only need to show that for \(t\in(0,\tau]\),
  \begin{multline}
    \label{eqn:restriction-check}
    \pi(x)\{1-\triangle\tilde{\Alpha}_{11}(t\mid 0,x)\}g_{11}(t\mid 0,x)\bigg\{\frac{w_1(t\mid x)}{G_1(t\!-\!\mid 0,x)}-\frac{w_0(t\mid x)}{G_0(t\!-\!\mid x)}\bigg\} \\
    = \frac{\triangle\tilde{\Alpha}_{12}(t\mid 0,x)}{G_1(t\!-\!\mid 0,x)e_1(0\mid x)}g_{21}(t\mid 0,x).
  \end{multline}
  The difference in the braces is
  \begin{align*}
    \MoveEqLeft\frac{w_1(t\mid x)}{G_1(t\!-\!\mid 0,x)}-\frac{w_0(t\mid x)}{G_0(t\!-\!\mid x)}\\
    &= \frac{1}{w_{\dott}(t\mid x)}\frac{\triangle\tilde{\Alpha}_{12}(t\mid a,x)}{G_1(t\!-\!\mid 0,x)}g_{21}(t\mid 0,x)\bigg\{\frac{1-\pi(x)}{\pi(x)e_1(0\mid x)G_1(t\!-\!\mid 0,x)}+\frac{1}{G_0(t\!-\!\mid x)}\bigg\} \\
    &= \frac{1}{\pi(x)}\frac{1}{w_{\dott}(t\mid x)}\bigg\{\frac{1-\pi(x)}{G_1(t\!-\!\mid 0,x)}+\frac{\pi(x)e_1(0\mid x)}{G_0(t\!-\!\mid x)}\bigg\}\\
    &\hphantom{=}\qquad\frac{\triangle\tilde{\Alpha}_{12}(t\mid 0,x)}{G_1(t\!-\!\mid 0,x)e_1(0\mid x)}g_{21}(t\mid 0,x),\\
    \intertext{which by inserting the definition of \(w_{\dott}(t\mid x)\) is simply}
    &= \frac{1}{\pi(x)\{1-\triangle\tilde{\Alpha}_{11}(t\mid 0,x)\}g_{11}(t\mid 0,x)}\frac{\triangle\tilde{\Alpha}_{12}(t\mid 0,x)}{G_1(t\!-\!\mid 0,x)e_1(0\mid x)}g_{21}(t\mid 0,x),
  \end{align*}
  and therefore \eqref{eqn:restriction-check} holds.
  
  That the terms \eqref{eqn:restriction-check-1} and \eqref{eqn:restriction-check-2} are equal shows that \(Dh_{1}^{*}(\tilde{T},\tilde{J},A,X)+(1-D)h_{0}^{*}(\tilde{T},\tilde{J},X)\in\dot{\mathcal{P}}_{1}\).
  Therefore, \(\varphi_{1}(0)(O)\) belongs to the tangent space \(\dot{\mathcal{P}}\) of the model \(\mathcal{P}\) at \(P\).
  Hence, it is the efficient influence function of \(\theta_{1}(0)\).

  The proof of \(\varphi_2(0)(O)\) being the efficient influence function of \(\theta_2(0)\) at \(P\in\mathcal{P}\) can be obtained by slightly modifying the derivations above and is thus omitted.
\end{proof}

\begin{remark}
\label{rem:if}
Inspecting the proof of Proposition~\ref{ppn:eif-discontinuity}, we find two more influence functions of the parameter \(\theta_{1}(0)\) in the model \(\mathcal{P}\).
The first influence function is obtained by replacing both \(w_1(t\mid x)\) and \(w_0(t\mid x)\) with \([\{1-\pi(x)\}+\pi(x)e_1(0\mid x)]^{-1}\) in \(\varphi_1(0)\).
The second one is obtained by replacing \(w_1(t\mid x)\) with \(\{\pi(x)e_1(0\mid x)\}^{-1}\) and \(w_0(t\mid x)\) with \(0\) in \(\varphi_1(0)\).
In this case, the resulting influence function is identical to the one proposed in \citetsuppmat{rytgaard2023estimation} but restricted to the RCT population.
\end{remark}

Lemma~\ref{lem:eif} is a simplification of the expressions of the efficient influence functions in Proposition~\ref{ppn:eif-discontinuity} under Assumptions~\ref{asn:censoring} and \ref{asn:discontinuity}.

\begin{proof}[Proof of Lemma~\ref{lem:eif}]
  Under Assumptions~\ref{asn:censoring} and \ref{asn:discontinuity}, we rewrite the following quantities from the main text using only the observed data distribution:
  \begin{align*}
    H_{\dott}(t\mid x) &= \pi(x)e_1(0\mid x)G_1(t\mid 0,x)+\{1-\pi(x)\}G_0(t\mid x),\\
    H_{1}(t\mid a,x) &= e_1(a\mid x)G_1(t\mid a,x),\\
    W_{k1}(t\mid a,x) &= \I(k=1)\tilde{S}_{1}(t\mid a,x)-\frac{\tilde{F}_{11}(\tau\mid a,x)-\tilde{F}_{11}(t\mid a,x)}{1-\triangle\tilde{\Alpha}_{k1}(t\mid 0,X)}.
  \end{align*}
  Let
  \begin{align*}
    \varphi_{1}^{*}(0)(O) &= \frac{(1-A)\pi(X)}{\alpha}\int_{0}^{\tau}\frac{g_{11}(t\mid A,X)}{H_{\dott}(t\!-\!\mid X)}\d\tilde{M}_{11}(t\mid A,X) \nonumber\\
      &\hphantom{=}\quad + \frac{D(1-A)}{\alpha}\int_{0}^{\tau}\frac{g_{21}(t\mid A,X)}{H_1(t\!-\!\mid A,X)}\d\tilde{M}_{12}(t\mid A,X)+ \frac{D}{\alpha}\big\{\tilde{F}_{11}(\tau\mid 0,X)-\theta_{1}(0)\big\}\\
   \varphi_{1}^{\dagger}(0)(O) &= \frac{(1-A)\pi(X)}{\alpha}\int_{0}^{\tau}\frac{W_{11}(t\mid A,X)}{H_{\dott}(t\!-\!\mid X)}\d\tilde{M}_{11}(t\mid A,X) \nonumber\\
    &\hphantom{=}\quad + \frac{D(1-A)}{\alpha}\int_{0}^{\tau}\frac{W_{21}(t\mid A,X)}{H_1(t\!-\!\mid A,X)}\d\tilde{M}_{12}(t\mid A,X)+ \frac{D}{\alpha}\big\{\tilde{F}_{11}(\tau\mid 0,X)-\theta_{1}(0)\big\}.
  \end{align*}
  The function \(\varphi_{1}^{\dagger}(0)(O)\) is the expression appearing on the right-hand side of the statement of the efficient influence function in the lemma.
  To show the lemma, we need to check for \(\varphi_1(0)(O)\) from Proposition~\ref{ppn:eif-discontinuity} that
  \[
    \E\big\{\varphi_1(0)(O)-\varphi_1^*(0)(O)\big\}^2=0,\quad \E\big\{\varphi_1^{*}(0)(O)-\varphi_1^\dagger(0)(O)\big\}^2=0,
  \]
  so that \(\varphi_1(0)(O)=\varphi_1^\dagger(0)(O)\) \(P\)-almost surely.
  Noting that
  \((1-D)\d\tilde{M}_{11}(t\mid 0,X)=(1-D)\d\tilde{M}_{01}(t\mid X)\) and
  \[
    \frac{g_{11}(t\mid 0,X)}{w_{\dott}(t\mid X)}\frac{1-\triangle\tilde{\Alpha}_{11}(t\mid 0,X)}{G_1(t\!-\!\mid 0,X)G_0(t\!-\!\mid X)}=\frac{1}{H_{\dott}(t\!-\!\mid X)},
  \]
  the \(L_2(P)\)-norm of the difference \(\E\big\{\varphi_1(0)(O)-\varphi_1^*(0)(O)\big\}^2\) is the expectation of the sum of two squared martingales
  \begin{align*}
    \MoveEqLeft \E\big\{\varphi_1(0)(O)-\varphi_1^*(0)(O)\big\}^2 \\
    &= \begin{multlined}[t][.9\textwidth]
    \E\bigg[\frac{D(1-A)}{\alpha}\I\{\pi(X)>0\}\int_0^\tau\frac{1}{w_{\dott}(t\mid X)}\frac{\triangle\tilde\Alpha_{12}(t\mid 0,X)}{G_1^2(t\!-\!\mid 0,X)}\frac{1-\pi(X)}{e_1(0\mid X)}\\
    (g_{11}g_{21})(t\mid 0,X)\d\tilde{M}_{11}(t\mid 0,X)\bigg]^2
    \end{multlined} \\
    &\hphantom{=}\quad +
    \begin{multlined}[t][.85\textwidth]
    \E\bigg\{\frac{1-D}{\alpha}\pi(X)\int_0^\tau\frac{1}{w_{\dott}(t\mid X)}\frac{\triangle\tilde\Alpha_{12}(t\mid 0,X)}{G_1(t\!-\!\mid 0,X)G_0(t\!-\!\mid X)}\\
    (g_{11}g_{21})(t\mid 0,X)\d\tilde{M}_{01}(t\mid X)\bigg\}^2
    \end{multlined}\\
    &= \begin{multlined}[t][.9\textwidth]
    \E\bigg[\frac{D(1-A)}{\alpha^2}\I\{\pi(X)>0\}\int_0^\tau\frac{1}{w_{\dott}^2(t\mid X)}\frac{\triangle\tilde\Alpha_{12}^2(t\mid 0,X)}{G_1^4(t\!-\!\mid 0,X)}\frac{\{1-\pi(X)\}^2}{e_1^2(0\mid X)}\\
    (g_{11}g_{21})^2(t\mid 0,X)\d\langle\tilde{M}_{11}\rangle(t\mid 0,X)\bigg]
    \end{multlined}\\
    &\hphantom{=}\quad +\begin{multlined}[t][.85\textwidth]
    \E\bigg\{\frac{1-D}{\alpha^2}\pi^2(X)\int_0^\tau\frac{1}{w_{\dott}^2(t\mid X)}\frac{\triangle\tilde\Alpha_{12}^2(t\mid 0,X)}{G_1^2(t\!-\!\mid 0,X)G_0^2(t\!-\!\mid X)}\\
    (g_{11}g_{21})^2(t\mid 0,X)\d\langle\tilde{M}_{01}\rangle(t\mid X)\bigg\}
    \end{multlined}\\
    &= \begin{multlined}[t][.9\textwidth]
    \E\bigg[\frac{D(1-A)}{\alpha^2}\I\{\pi(X)>0\}\int_0^\tau\frac{1}{w_{\dott}^2(t\mid X)}\frac{\triangle\tilde\Alpha_{12}^2(t\mid 0,X)}{G_1^4(t\!-\!\mid 0,X)}\frac{\{1-\pi(X)\}^2}{e_1^2(0\mid X)}\\
    (g_{11}g_{21})^2(t\mid 0,X)\I(\tilde{T}\geq t)\{1-\triangle\tilde{\Alpha}_{11}(t\mid 0,X)\}\d\tilde{\Alpha}_{11}(t\mid 0,X)\bigg]
    \end{multlined}\\
    &\hphantom{=}\quad +\begin{multlined}[t][.85\textwidth]
    \E\bigg\{\frac{1-D}{\alpha^2}\pi^2(X)\int_0^\tau\frac{1}{w_{\dott}^2(t\mid X)}\frac{\triangle\tilde\Alpha_{12}^2(t\mid 0,X)}{G_1^2(t\!-\!\mid 0,X)G_0^2(t\!-\!\mid X)}\\
    (g_{11}g_{21})^2(t\mid 0,X)\I(\tilde{T}\geq t)\{1-\triangle\tilde\Alpha_{01}(t\mid X)\}\d\tilde\Alpha_{01}(t\mid X)\bigg\}
    \end{multlined}\\
    &=0.
  \end{align*}
  The last equality is a direct consequence of Assumption~\ref{asn:discontinuity}, that is, \(\triangle\tilde{\Alpha}_{11}(t\mid 0,x)\triangle\tilde{\Alpha}_{12}(t\mid 0,x)=0\) and \(\triangle\tilde{\Alpha}_{01}(t\mid x)\triangle\tilde{\Alpha}_{12}(t\mid 0,x)=0\) for any \(x\in\mathcal{X}_0\cap\mathcal{X}_1\).
  Similarly, we have
  \begin{align*}
    \MoveEqLeft \E\big\{\varphi_1^{*}(0)(O)-\varphi_1^\dagger(0)(O)\big\}^2 \\
    &= \E\bigg\{\frac{(1-A)\pi(X)}{\alpha}\int_{0}^{\tau}\frac{(g_{11}-W_{11})(t\mid A,X)}{H_{\dott}(t\!-\!\mid X)}\d\tilde{M}_{11}(t\mid A,X)\bigg\}^2\\
    &\hphantom{=}\quad +\E\bigg\{\frac{D(1-A)}{\alpha}\int_{0}^{\tau}\frac{(g_{21}-W_{21})(t\mid A,X)}{H_1(t\!-\!\mid A,X)}\d\tilde{M}_{12}(t\mid A,X)\bigg\}^2\\
    &= \E\bigg\{\frac{(1-A)\pi^2(X)}{\alpha^2}\int_{0}^{\tau}\frac{(g_{11}-W_{11})^2(t\mid A,X)}{H_{\dott}^2(t\!-\!\mid X)}\d\langle\tilde{M}_{11}\rangle(t\mid A,X)\bigg\}\\
    &\hphantom{=}\quad +\E\bigg\{\frac{D(1-A)}{\alpha^2}\int_{0}^{\tau}\frac{(g_{21}-W_{21})^2(t\mid A,X)}{H_1^2(t\!-\!\mid A,X)}\d\langle\tilde{M}_{12}\rangle(t\mid A,X)\bigg\}\\
    &= \E\bigg[\frac{(1-A)\pi^2(X)}{\alpha^2}\int_{0}^{\tau}\frac{\triangle\tilde{\Alpha}_{12}^2(t\mid 0,X)}{\big\{1-\triangle\tilde{\Alpha}_{11}(t\mid 0,X)\big\}^2\big\{1-\triangle(\tilde{\Alpha}_{11}+\tilde{\Alpha}_{12})(t\mid 0,X)\big\}^2}\\
    &\hphantom{=}\qquad \frac{\big\{\tilde{F}_{11}(\tau\mid 0,X)-\tilde{F}_{11}(t\mid 0,X)\big\}^2}{H_{\dott}^2(t\!-\!\mid X)}I(\tilde{T}\geq t)\big\{1-\triangle\tilde{\Alpha}_{11}(t\mid 0,X)\big\}\d\tilde{\Alpha}_{11}(t\mid 0,X)\bigg]\\
    &\hphantom{=}\quad + \E\bigg[\frac{D(1-A)}{\alpha^2}\int_{0}^{\tau}\frac{\triangle\tilde{\Alpha}_{11}^2(t\mid 0,X)}{\big\{1-\triangle\tilde{\Alpha}_{12}(t\mid 0,X)\big\}^2\big\{1-\triangle(\tilde{\Alpha}_{11}+\tilde{\Alpha}_{12})(t\mid 0,X)\big\}^2}\\
    &\hphantom{=}\qquad \frac{\big\{\tilde{F}_{11}(\tau\mid 0,X)-\tilde{F}_{11}(t\mid 0,X)\big\}^2}{H_{1}^2(t\!-\!\mid 0,X)}I(\tilde{T}\geq t)\big\{1-\triangle\tilde{\Alpha}_{12}(t\mid 0,X)\big\}\d\tilde{\Alpha}_{12}(t\mid 0,X)\bigg]\\
    &=0.
  \end{align*}
  The last equality follows from Assumption~\ref{asn:discontinuity}, \(\triangle\tilde{\Alpha}_{11}(t\mid 0,x)\triangle\tilde{\Alpha}_{12}(t\mid 0,x)=0\) for any \(x\in \mathcal{X}_1\).
  
  The equivalence for the efficient influence function \(\varphi_2(0)(O)\) can be argued analogously.
\end{proof}

\subsection{Proof of Corollary~\ref{cor:variance-reduction}}
  Let \(\tilde{\mathcal{P}}\) be the same as the model \(\mathcal{P}\) but the restriction \(\d\tilde{\Alpha}_{11}(t\mid 0,x)=\d\tilde{\Alpha}_{01}(t\mid x)\) is removed.
  The semiparametric efficiency bound of \(\theta_{1}(0)\) under \(P\in\tilde{\mathcal{P}}\) can be characterized by the variance of the efficient influence function
  \begin{align*}
    \tilde{\varphi}_{1}(0)(O) &= \frac{D}{\alpha}\frac{1-A}{e_1(A\mid X)}\int_{0}^{\tau}\frac{g_{11}(t\mid A,X)}{G_1(t\!-\!\mid A,X)}\d \tilde{M}_{11}(t\mid A,X) \\
                              &\hphantom{=}\quad + \frac{D}{\alpha}\frac{1-A}{e_1(A\mid X)}\int_{0}^{\tau}\frac{g_{21}(t\mid A,X)}{G_1(t\!-\!\mid A,X)}\d \tilde{M}_{12}(t\mid A,X) \\
                              &\hphantom{=}\quad + \frac{D}{\alpha}\big\{\tilde{F}_{11}(\tau\mid 0,X)-\theta_{1}(0)\big\}.
  \end{align*}
  See Remark~\ref{rem:if} for the justification of this claim.
  Then the variance of the difference in the efficient influence functions under these two models with respect to \(P\in\mathcal{P}\) is
  \begin{align*}
    \MoveEqLeft E\{\varphi_{1}(0)(O)-\tilde{\varphi}_{1}(0)(O)\}^{2}\\
    &=\begin{multlined}[t][0.9\textwidth]
    E\bigg\{\frac{\pi(X)}{\alpha}(1-A)\int_{0}^{\tau}\frac{W_{11}(t\mid A,X)}{H_{\dott}(t\!-\!\mid X)}\d \tilde{M}_{11}(t\mid A,X)\\
    -\frac{D}{\alpha}(1-A)\int_{0}^{\tau}\frac{W_{11}(t\mid A,X)}{H_1(t\!-\!\mid A,X)}\d \tilde{M}_{11}(t\mid A,X)\bigg\}^{2}
    \end{multlined}\\
    &= E\bigg[\frac{1-A}{\alpha}\int_{0}^{\tau}\bigg\{\frac{\pi(X)}{H_{\dott}(t\!-\!\mid X)}-\frac{D}{H_1(t\!-\!\mid A,X)}\bigg\}
    W_{11}(t\mid A,X)\d\tilde{M}_{11}(t\mid A,X)\bigg]^{2}\\
    &= E\bigg\langle\frac{1-A}{\alpha}\int_{0}^{\tau}\bigg\{\frac{\pi(X)}{H_{\dott}(t\!-\!\mid X)}-\frac{D}{H_1(t\!-\!\mid A,X)}\bigg\}
    W_{11}(t\mid A,X)\d\tilde{M}_{11}(t\mid A,X)\bigg\rangle\\
    &= E\bigg[\frac{1-A}{\alpha^{2}}\int_{0}^{\tau}\bigg\{\frac{\pi(X)}{H_{\dott}(t\!-\!\mid X)}-\frac{D}{H_1(t\!-\!\mid A,X)}\bigg\}^{2}
    \{W_{11}(t\mid A,X)\}^{2}\d\langle\tilde{M}_{11}\rangle(t\mid A,X)\bigg]\\
    &= \begin{multlined}[t][0.9\textwidth]
    E\bigg[\frac{1-A}{\alpha^{2}}\int_{0}^{\tau}\bigg\{\frac{\pi(X)}{H_{\dott}(t\!-\!\mid X)}-\frac{D}{H_1(t\!-\!\mid A,X)}\bigg\}^{2} \\
    \{W_{11}(t\mid A,X)\}^{2}\I(\tilde{T}\geq t)\{1-\triangle\tilde{\Alpha}_{\dott 1}(t\mid A,X)\}\d\tilde{\Alpha}_{\dott 1}(t\mid A,X)\bigg] 
    \end{multlined}\\
    &= \begin{multlined}[t][0.9\textwidth]
    E\bigg[\frac{1}{\alpha^{2}}\int_{0}^{\tau}\bigg\{\frac{\{\pi(X)\}^{2}}{H_{\dott}(t\!-\!\mid X)}+\frac{\pi(X)}{H_1(t\!-\!\mid 0,X)}-\frac{2\{\pi(X)\}^{2}}{H_{\dott}(t\!-\!\mid X)}\bigg\}\\
    \{W_{11}(t\mid 0,X)\}^{2}\{1-\triangle\tilde{\Alpha}_{\dott 1}(t\mid 0,X)\}\d\tilde{\Alpha}_{\dott 1}(t\mid 0,X)\bigg] 
    \end{multlined}\\
    &= \begin{multlined}[t][0.9\textwidth]
    E\bigg[\frac{\pi(X)}{\alpha^{2}}\int_{0}^{\tau}\bigg\{\frac{1}{H_1(t\!-\!\mid 0,X)}-\frac{\pi(X)}{H_{\dott}(t\!-\!\mid X)}\bigg\}\\
    \{W_{11}(t\mid 0,X)\}^{2}\{1-\triangle\tilde{\Alpha}_{\dott 1}(t\mid 0,X)\}\d\tilde{\Alpha}_{\dott 1}(t\mid 0,X)\bigg]
    \end{multlined}\\
    &= \begin{multlined}[t][0.9\textwidth]
    E\bigg[\frac{\pi(X)\{1-\pi(X)\}}{\alpha^{2}}\int_{0}^{\tau}\frac{(S_{0}S_0^c)(t\!-\!\mid X)}{H_1(t\!-\!\mid 0,X)H_{\dott}(t\!-\!\mid X)} \\
    \{W_{11}(t\mid 0,X)\}^{2}\{1-\triangle\tilde{\Alpha}_{\dott 1}(t\mid 0,X)\}\d\tilde{\Alpha}_{\dott 1}(t\mid 0,X)\bigg].
    \end{multlined}
  \end{align*}
  The expression in the statement of the corollary follows from Assumptions~\ref{asn:censoring} and \ref{asn:discontinuity}.

  \subsection{Proof of Theorem~\ref{thm:asymptotic}}

We first state the deferred assumptions in the statement of Theorem~\ref{thm:asymptotic}.
For any \(\Alpha,\Alpha^{*}\in\mathcal{A}\), let \(\Alpha\perp_{\triangle}\Alpha^*\) denote that \(\int_0^\tau\triangle{\Alpha}(t)\d{\Alpha}^*(t)=0\).

\begin{assumption}[Regularity conditions]
  \hfill
  \label{asn:regularity}
  \begin{enumerate}[nosep,label=(\roman*)]
  \item There exists a universal constant \(C> 1\) such that
  \begin{align*}
  &\hat\alpha\geq C^{-1}, \hat{e}_1(0\mid x)\geq C^{-1}, \bar{e}_1(0\mid x)\geq C^{-1},\\
  &(\Pi\hat{\Alpha}_{\dott 1})(\tau\mid 0,x)\geq C^{-1},(\Pi\bar{\Alpha}_{\dott 1})(\tau\mid 0,x)\geq C^{-1},\\
  &(\Pi\hat{\Alpha}_{12})(\tau\mid 0,x)\geq C^{-1}, 
  (\Pi\bar{\Alpha}_{12})(\tau\mid 0,x)\geq C^{-1},
  \end{align*}
  wherever \(\pi(x)>0\), and
  \begin{align*}
  & (\Pi\hat{\Alpha}_{02})(\tau\mid x)\geq C^{-1},
  (\Pi\bar{\Alpha}_{02})(\tau\mid x)\geq C^{-1}, \\
  &(\Pi\hat{\Alpha}_{0}^c)(\tau\mid x)\geq C^{-1},
  (\Pi\bar{\Alpha}_{0}^c)(\tau\mid x)\geq C^{-1},
  \end{align*}
  wherever \(\pi(x)\{1-\pi(x)\}>0\);
  \item \(\{x:\bar\pi(x)>0\}\subset\mathcal{X}_1\);
  \item For \(x\in\mathcal{X}_1\),
    \begin{align*}
      \big\{\hat\Alpha_{\dott 1}(t\mid 0,x),\Alpha_{\dott 1}(t\mid 0,x)\big\} &\perp_{\triangle} \big\{\hat\Alpha_{12}(t\mid 0,x),\Alpha_{12}(t\mid 0,x),\hat\Alpha_{02}(t\mid x),\Alpha_{02}(t\mid x)\big\},\\
    \big\{\hat\Alpha_{\dott 1}(t\mid 0,x),\bar\Alpha_{\dott 1}(t\mid 0,x)\big\} &\perp_{\triangle} \begin{multlined}[t][.512\textwidth]
    \big\{\hat\Alpha_{12}(t\mid 0,x),\bar\Alpha_{12}(t\mid 0,x),\bar\Alpha_{02}(t\mid x),\\
    \bar\Alpha_{1}^c(t\mid 0,x),\bar\Alpha_{0}^c(t\mid x)\big\},
    \end{multlined}\\
    \big\{\hat\Alpha_{12}(t\mid 0,x),\bar\Alpha_{12}(t\mid 0,x)\big\} & \perp_{\triangle} \big\{\bar\Alpha_{\dott 1}(t\mid 0,x),\bar\Alpha_{1}^c(t\mid 0,x)\big\};
    \end{align*}
  \item \(\hat\ell_1(0)\) and \(\ell_1(0)\) belong to some \(P\)-Donsker class.
  \end{enumerate}
\end{assumption}

\begin{assumption}[Rate conditions]
\label{asn:rate}
The following integrals converge sufficiently fast:
      \begin{multline}
P\bigg[\int_{0}^{\tau}\bigg\{\hat\pi(X)\frac{H^*_{\dott}}{\hat{H}^{*}_{\dott}}(t\!-\!\mid X)-\pi(X)\frac{\Pi\Alpha_{12}}{\Pi\hat{\Alpha}_{12}}(t\!-\!\mid 0,X)\bigg\} \\
    \hat{W}_{\dott 1}(t\mid 0,X)\big\{1-\triangle\hat\Alpha_{\dott 1}(t\mid 0, X)\big\}\d\bigg(\frac{\Pi\Alpha_{\dott 1}}{\Pi\hat{\Alpha}_{\dott 1}}\bigg)(t\mid 0,X)\bigg]=o_{P}(n^{-1/2}),
    \label{eqn:remainder-1}
  \end{multline}
  \begin{multline}
    P\bigg[\pi(X)\int_{0}^{\tau}\bigg\{\frac{e_1(0\mid X)S_1^c(t\!-\!\mid 0,X)}{\hat{e}_1(0\mid X)\hat{S}_1^c(t\!-\!\mid 0,X)}-1\bigg\}\frac{\Pi\Alpha_{\dott 1}}{\Pi\hat{\Alpha}_{\dott 1}}(t\!-\!\mid 0,X)\\
    \hat{W}_{12}(t\mid 0,X)\big\{1-\triangle\hat\Alpha_{12}(t\mid 0,X)\big\}\d\bigg(\frac{\Pi\Alpha_{12}}{\Pi\hat{\Alpha}_{12}}\bigg)(t\mid 0,X)\bigg]=o_{P}(n^{-1/2}),
    \label{eqn:remainder-2}
  \end{multline}
  where
  \[
    H^{*}_{\dott}(t\mid X)=\pi(X)e_{1}(0\mid X)\{(\Pi\Alpha_{12})S_1^c\}(t\mid 0,X)+\{1-\pi(X)\}\{(\Pi\Alpha_{02})S_0^c\}(t\mid X).
  \]
\end{assumption}

\begin{remark}
\label{rem:rate}
Since estimators for cumulative hazards often contain jumps, the convergence is stated in terms of the means of stochastic integrals \citepsuppmat{westling2024inference}.
When the event time distribution is absolutely continuous with respect to the Lebesgue measure and the conditional cumulative hazards are estimated by continuous functions, the remainder terms will admit a more conventional product structure.
This is because the Cauchy-Schwarz inequality can be applied with respect to the product measure of \(P\) marginalized to the support \(\mathcal{X}_1\cup\mathcal{X}_0\) of \(X\) and the Lebesgue measure over the time interval \((0,\tau]\); refer to \citetsuppmat{rytgaard2023estimation} for precise formulations.
\end{remark}

\begin{remark}
\label{rem:rate-rmtl}
To establish asymptotic linearity of the estimator \(\hat\gamma_1(0)\), the same convergence rates of the remainder terms \eqref{eqn:remainder-1}--\eqref{eqn:remainder-2} should hold when swapping \(\hat{W}_{\dott 1}(t\mid 0,X)\) and \(\hat{W}_{12}(t\mid 0,X)\) out for
\begin{align*}
  &\int_t^\tau\bigg\{\hat{S}_{1}(t\!-\!\mid 0,X)-\frac{\hat{F}_{11}(s\mid 0,X)-\hat{F}_{11}(t\mid 0,X)}{1-\triangle\hat{\Alpha}_{\dott 1}(t\mid 0,X)}\bigg\}\d s,\\
  &\int_t^\tau\frac{\hat{F}_{11}(s\mid 0,X)-\hat{F}_{11}(t\mid 0,X)}{1-\triangle\hat{\Alpha}_{12}(t\mid 0,X)}\d s.
\end{align*}
\end{remark}

We will use the following lemmas from the literature.
\begin{lemma}[Integration by parts, {\citealpsuppmat[Theorem A.1.2]{fleming1991counting}}]
  \label{lem:integration}
  Let \(F:[0,\infty)\to\mathbb{R}\) and \(G:[0,\infty)\to\mathbb{R}\) be c\`{a}dl\`{a}g functions of bounded variation on any finite interval.
  Then
  \[
    F(t)G(t)-F(s)G(s) = \int_{(s,t]}F(u-)\d G(u) + \int_{(s,t]}G(u)\d F(u).
  \]
\end{lemma}

\begin{lemma}[Duhamel and backward equations, {\citealpsuppmat{gill1990survey}}]
  \label{lem:product}
  Let \(F:[0,\infty)\to\mathbb{R}\) and \(G:[0,\infty)\to\mathbb{R}\) be c\`{a}dl\`{a}g functions of bounded variation on any finite interval.
  Then
  \begin{multline*}
    \prod_{u\in(s,t]}\{1+\d F(u)\}-\prod_{u\in(s,t]}\{1+\d G(u)\} \\
    =\int_{u\in(s,t]}\prod_{v\in(s,u)}\{1+\d F(v)\}\d (F-G)(u)\prod_{v\in(u,t]}\{1+\d G(v)\}, \tag{Duhamel}
  \end{multline*}
  \begin{equation*}
  \prod_{u\in(s,t]}\{1+\d F(u)\}-1 = \int_{u\in(s,t]}\prod_{v\in(u,t]}\{1+\d F(v)\}\d F(u). \tag{backward}
\end{equation*}
\end{lemma}

To make the notations more compact, we use the symbol \(S\) to represent the product integral \(\Pi\Alpha\) for \(\Alpha\in\mathcal{A}\), and superscripts and subscripts in \(\Alpha\) are carried over to \(S\).
For example, \(S_{\dott 1}(t\mid 0,x)=(\Pi\Alpha_{\dott 1})(t\mid 0,x)\).
We use the nuisance parameters with checkmarks as a placeholder for either the probability limits in Assumption~\ref{asn:plim} (with bars) or the estimated nuisance parameters (with hats).
Define
\begin{align*}
  \check{q}_{1}(t\mid X) &= \check{S}_{1}(t\mid 0,X)-\check{F}_{11}(\tau\mid 0,X)+\check{F}_{11}(t\mid 0,X),\\
  \check{q}_{2}(t\mid X) &= -\check{F}_{11}(\tau\mid 0,X)+\check{F}_{11}(t\mid 0,X),\\
  \check{b}_{1}(t\mid X) &= \check{\pi}(x)\check{e}_{1}(0\mid X)(\check{S}_{12}\check{S}_1^{c})(t\!-\!\mid 0,X) + \{1-\check\pi(X)\}(\check{S}_{02}\check{S}_0^{c})(t\!-\!\mid X),\\
  {b}_{1}(t\mid X) &= {\pi}(x){e}_{1}(0\mid X)({S}_{12}{S}_1^{c})(t\!-\!\mid 0,X) + \{1-\pi(X)\}({S}_{02}{S}_0^{c})(t\!-\!\mid X),\\
  \check{b}_{2}(t\mid X) &= \check{e}_{1}(0\mid X)(\check{S}_{\dott 1}\check{S}_1^{c})(t\!-\!\mid 0,X),\\
  {b}_{2}(t\mid X) &= {e}_{1}(0\mid X)(S_{\dott 1}{S}_1^{c})(t\!-\!\mid 0,X).
\end{align*}
The function obtained by substituting all nuisance parameters in \(\ell_1(0)\) by their version with the checkmark can be written in terms of the quantities above as
\[
  \check{\ell}_{1}(0)(O)=\sum_{m=1}^{3}\check\ell_{1m}(0)(O),
\]
where
\begin{align*}
  \check\ell_{11}(0)(O) &= \frac{1-A}{\check\alpha}\int_{0}^{\tau}\check{\pi}(X)\frac{\check{q}_1}{\check{b}_1}(t\mid X)\frac{\d\check{M}_{\dott 1}}{\check{S}_{\dott 1}}(t\mid 0,X), \\
  \check\ell_{12}(0)(O) &= \frac{D(1-A)}{\check\alpha}\int_{0}^{\tau}\frac{\check{q}_2}{\check{b}_2}(t\mid X)\frac{\d\check{M}_{12}}{\check{S}_{12}}(t\mid 0,X), \\
  \check\ell_{13}(0)(O) &= \frac{D}{\check\alpha}\check{F}_{11}(t\mid 0,X).
\end{align*}

The following lemma will be used in two versions by substituting the nuisance parameters with checkmark with their estimates and the probability limits of their estimators, respectively.
\begin{lemma}
  \label{lem:diff}
  Suppose Assumptions~\ref{asn:plim} and \ref{asn:regularity} hold.
  Then
  \begin{align*}
    \MoveEqLeft P\bigg\{\check{\ell}_1(0)-\frac{\alpha}{\check\alpha}\ell_1(0)\bigg\}\\
    &= P\bigg[\frac{1}{\check\alpha}\int_{0}^{\tau}\bigg\{\check{\pi}(X)\frac{b_{1}}{\check{b}_1}(t\mid X)-\pi(X)\frac{S_{12}}{\check{S}_{12}}(t\!-\!\mid 0,X)\bigg\}\check{q}_1(t\mid X)\d\bigg(1-\frac{S_{\dott 1}}{\check{S}_{\dott 1}}\bigg)(t\mid 0,X)\bigg]\\
    &\hphantom{=}\quad + P\bigg[\frac{\pi(X)}{\check{\alpha}}\int_{0}^{\tau}\bigg\{\frac{b_{2}}{\check{b}_2}(t\mid X)-\frac{S_{\dott 1}}{\check{S}_{\dott 1}}(t\!-\!\mid 0,X)\bigg\}\check{q}_2(t\mid X)\d\bigg(1-\frac{S_{12}}{\check{S}_{12}}\bigg)(t\mid 0,X)\bigg].
  \end{align*}
\end{lemma}

\begin{proof}
  We have
  \begin{align*}
    P\check\ell_{11}(0) &= P\bigg\{\frac{1-A}{\check\alpha}\int_{0}^{\tau}\check{\pi}(X)\frac{\check{q}_1}{\check{b}_1}(t\mid X)\frac{\d\check{M}_{\dott 1}}{\check{S}_{\dott 1}}(t\mid 0,X)\bigg\}\\
    &= P\bigg\{\frac{1-A}{\check\alpha}\int_{0}^{\tau}\check{\pi}(X)\frac{\check{q}_1}{\check{b}_1}(t\mid X)\I(\tilde{T}\geq t)\frac{\d(\Alpha_{\dott 1}-\check{\Alpha}_{\dott 1})}{\check{S}_{\dott 1}}(t\mid 0,X)\bigg\}\\
    &= P\bigg\{\frac{1}{\check\alpha}\int_{0}^{\tau}\check{\pi}(X)\frac{\check{q}_1}{\check{b}_1}(t\mid X)P(\tilde{T}\geq t,A=0\mid X)\frac{\d(\Alpha_{\dott 1}-\check{\Alpha}_{\dott 1})}{\check{S}_{\dott 1}}(t\mid 0,X)\bigg\}\\
    &= P\bigg\{\frac{1}{\check\alpha}\int_{0}^{\tau}\check{\pi}(X)\frac{\check{q}_1}{\check{b}_1}(t\mid X)\frac{1}{\check{S}_{\dott 1}(t\mid 0,X)}b_1(t\mid X)S_{\dott 1}(t\!-\!\mid 0,X)\d(\Alpha_{\dott 1}-\check{\Alpha}_{\dott 1})(t\mid 0,X)\bigg\}\\
    &= P\bigg\{\frac{1}{\check\alpha}\int_{0}^{\tau}\check{\pi}(X)\frac{b_1\check{q}_1}{\check{b}_1}(t\mid X)\frac{S_{\dott 1}}{\check{S}_{\dott 1}}(t\mid 0,X)\d(\Alpha_{\dott 1}-\check{\Alpha}_{\dott 1})(t\mid 0,X)\bigg\}.
  \end{align*}
  Also, we have
  \begin{align*}
    P\check\ell_{12}(0) &= P\bigg\{\frac{D(1-A)}{\check\alpha}\int_{0}^{\tau}\frac{\check{q}_2}{\check{b}_2}(t\mid X)\frac{\d\check{M}_{12}}{\check{S}_{12}}(t\mid 0,X)\bigg\}\\
    &=P\bigg[\frac{\pi(X)}{\check{\alpha}}\int_{0}^{\tau}\frac{b_2\check{q}_2}{\check{b}_2}(t\mid X)\frac{S_{12}(t\!-\!\mid 0,X)}{\check{S}_{12}(t\mid 0,X)}\d(\Alpha_{12}-\check{\Alpha}_{12})(t\mid 0,X)\bigg].
  \end{align*}
  Let \(\Alpha_{1}(t\mid 0,x)=(\Alpha_{\dott 1}+\Alpha_{12})(t\mid 0,x)\) and \(\check\Alpha_{1}(t\mid 0,x)=(\check\Alpha_{\dott 1}+\check\Alpha_{12})(t\mid 0,x)\), the all-cause hazard.
  By the Duhamel equation in Lemma~\ref{lem:product},
  \begin{align*}
    \MoveEqLeft(\check{F}_{11}-F_{11})(\tau\mid 0,x)\\
    &= \int_{0}^{\tau}\check{S}_{1}(t\!-\!\mid 0,x)\d\check{\Alpha}_{\dott 1}(t\mid 0,x)-\int_{0}^{\tau}S_{1}(t\!-\!\mid 0,x)\d\Alpha_{\dott 1}(t\mid 0,x)\\
    &= \int_{0}^{\tau}(\check{S}_{1}-S_1)(t\!-\!\mid 0,x)\d\check{\Alpha}_{\dott 1}(t\mid 0,x)+\int_{0}^{\tau}S_{1}(t\!-\!\mid 0,x)\d(\check\Alpha_{\dott 1}-\Alpha_{\dott 1})(t\mid 0,x)\\
    &= \int_{0}^{\tau}\int_{s\in(0,t)}\check{S}_{1}(t\!-\!\mid 0,x)\frac{S_1(s\!-\!\mid 0,x)}{\check{S}_1(s\mid 0,x)}\d(\Alpha_{1}-\check{\Alpha}_{1})(s\mid 0,x)\d\check{\Alpha}_{\dott 1}(t\mid 0,x)\\
    &\hphantom{=}\quad +\int_{0}^{\tau}S_{1}(t\!-\!\mid 0,x)\d(\check\Alpha_{\dott 1}-\Alpha_{\dott 1})(t\mid 0,x)\\
    &= \int_{s\in(0,\tau)}\int_{s}^{\tau}\check{S}_{1}(t\!-\!\mid 0,x)\d\check{\Alpha}_{\dott 1}(t\mid 0,x)\frac{S_1(s\!-\!\mid 0,x)}{\check{S}_1(s\mid 0,x)}\d(\Alpha_{1}-\check{\Alpha}_{1})(s\mid 0,x)\\
    &\hphantom{=}\quad +\int_{0}^{\tau}S_{1}(t\!-\!\mid 0,x)\d(\check\Alpha_{\dott 1}-\Alpha_{\dott 1})(t\mid 0,x)\\
    &= \int_{s\in(0,\tau)}\frac{\check{F}_{11}(\tau\mid 0,x)-\check{F}_{11}(s\mid 0,x)}{1-\triangle\check{\Alpha}_{1}(s\mid 0,x)}\frac{S_1}{\check{S}_1}(s\!-\!\mid 0,x)\d(\Alpha_{1}-\check{\Alpha}_{1})(s\mid 0,x)\\
    &\hphantom{=}\quad +\int_{0}^{\tau}S_{1}(t\!-\!\mid 0,x)\d(\check\Alpha_{\dott 1}-\Alpha_{\dott 1})(t\mid 0,x)\\
    &= -\int_{0}^{\tau}\bigg\{S_{1}(t\!-\!\mid 0,x)-\frac{\check{F}_{11}(\tau\mid 0,x)-\check{F}_{11}(t\mid 0,x)}{1-\triangle\check{\Alpha}_{\dott 1}(t\mid 0,x)}\frac{S_1}{\check{S}_1}(t\!-\!\mid 0,x)\bigg\}\d(\Alpha_{\dott 1}-\check{\Alpha}_{\dott 1})(t\mid 0,x)\\
    &\hphantom{=}\quad  +\int_{0}^{\tau}\frac{\check{F}_{11}(\tau\mid 0,x)-\check{F}_{11}(t\mid 0,x)}{1-\triangle\check{\Alpha}_{12}(t\mid 0,x)}\frac{S_1}{\check{S}_1}(t\!-\!\mid 0,x)\d(\Alpha_{12}-\check{\Alpha}_{12})(t\mid 0,x)\\
    &= -\int_{0}^{\tau}\bigg\{S_{12}(t\!-\!\mid 0,x)-\frac{\check{F}_{11}(\tau\mid 0,x)-\check{F}_{11}(t\mid 0,x)}{\check{S}_{\dott 1}(t\mid 0,x)}\frac{S_{12}}{\check{S}_{12}}(t\!-\!\mid 0,x)\bigg\}\\
    &\qquad\qquad S_{\dott 1}(t\!-\!\mid 0,x)\d(\Alpha_{\dott 1}-\check{\Alpha}_{\dott 1})(t\mid 0,x)\\
    &\hphantom{=}\quad  +\int_{0}^{\tau}\frac{\check{F}_{11}(\tau\mid 0,x)-\check{F}_{11}(t\mid 0,x)}{\check{S}_{12}(t\mid 0,x)}\frac{S_{\dott 1}S_{12}}{\check{S}_{\dott 1}}(t\!-\!\mid 0,x)\d(\Alpha_{12}-\check{\Alpha}_{12})(t\mid 0,x)\\
    &= -\int_{0}^{\tau}\frac{S_{12}}{\check{S}_{12}}(t\!-\!\mid 0,x)\big\{\check{S}_{\dott 1}(t\mid 0,x)\check{S}_{12}(t\!-\!\mid 0,x)-\check{F}_{11}(\tau\mid 0,x)+\check{F}_{11}(t\mid 0,x)\big\}\\
    &\qquad\qquad \frac{S_{\dott 1}(t\!-\!\mid 0,x)}{\check{S}_{\dott 1}(t\mid 0,x)}\d(\Alpha_{\dott 1}-\check{\Alpha}_{\dott 1})(t\mid 0,x)\\
    &\hphantom{=}\quad  +\int_{0}^{\tau}\frac{S_{\dott 1}}{\check{S}_{\dott 1}}(t\!-\!\mid 0,x)\big\{\check{F}_{11}(\tau\mid 0,x)-\check{F}_{11}(t\mid 0,x)\big\}\frac{S_{12}(t\!-\!\mid 0,x)}{\check{S}_{12}(t\mid 0,x)}\d(\Alpha_{12}-\check{\Alpha}_{12})(t\mid 0,x).
  \end{align*}
  By Assumption~\ref{asn:regularity},
  \[
  \check{S}_{12}(t\!-\!\mid 0,x)\d(\Alpha_{\dott 1}-\check{\Alpha}_{\dott 1})(t\mid 0,x)=\check{S}_{12}(t\mid 0,x)\d(\Alpha_{\dott 1}-\check{\Alpha}_{\dott 1})(t\mid 0,x).
  \]
  Therefore,
  \begin{align*}
    \MoveEqLeft P\bigg\{\check\ell_{13}(0)- \frac{\alpha}{\check\alpha}\ell_1(0)\bigg\} \\
    &=P\check\ell_{13}(0) - \frac{\alpha}{\check\alpha}\theta_{1}(0)\\
    &= P\bigg[\frac{D}{\check{\alpha}}(\check{F}_{11}-F_{11})(\tau\mid 0,X)\bigg] \\
    &= -P\bigg[\frac{\pi(X)}{\check{\alpha}}\int_{0}^{\tau}\frac{S_{12}}{\check{S}_{12}}(t\!-\!\mid 0,X)\check{q}_1(t\mid X)\frac{S_{\dott 1}(t\!-\!\mid 0,X)}{\check{S}_{\dott 1}(t\mid 0,X)}\d(\Alpha_{\dott 1}-\check{\Alpha}_{\dott 1})(t\mid 0,X)\bigg] \\
    &\hphantom{=}\quad -P\bigg[\frac{\pi(X)}{\check{\alpha}}\int_{0}^{\tau}\frac{S_{\dott 1}}{\check{S}_{\dott 1}}(t\!-\!\mid 0,X)\check{q}_2(t\mid X)\frac{S_{12}(t\!-\!\mid 0,X)}{\check{S}_{12}(t\mid 0,X)}\d(\Alpha_{12}-\check{\Alpha}_{12})(t\mid 0,X)\bigg].
  \end{align*}
  Summing up the three terms in the previous displays gives
  \begin{align*}
    \MoveEqLeft P\bigg\{\check{\ell}_1(0)-\frac{\alpha}{\check\alpha}\ell_1(0)\bigg\} \\
    &= \begin{multlined}[t][.9\textwidth]
    P\bigg[\frac{1}{\check\alpha}\int_{0}^{\tau}\bigg\{\check{\pi}(X)\frac{b_{1}}{\check{b}_1}(t\mid X)-\pi(X)\frac{S_{12}}{\check{S}_{12}}(t\!-\!\mid 0,X)\bigg\}\\
    \check{q}_1(t\mid X)\frac{S_{\dott 1}(t\!-\!\mid 0,X)}{\check{S}_{\dott 1}(t\mid 0,X)}\d(\Alpha_{\dott 1}-\check{\Alpha}_{\dott 1})(t\mid 0,X)\bigg]
    \end{multlined}\\
    &\hphantom{=}\quad + \begin{multlined}[t][.85\textwidth]
    P\bigg[\frac{\pi(X)}{\check{\alpha}}\int_{0}^{\tau}\bigg\{\frac{b_{2}}{\check{b}_2}(t\mid X)-\frac{S_{\dott 1}}{\check{S}_{\dott 1}}(t\!-\!\mid 0,X)\bigg\}\\
    \check{q}_2(t\mid X)\frac{S_{12}(t\!-\!\mid 0,X)}{\check{S}_{12}(t\mid 0,X)}\d(\Alpha_{12}-\check{\Alpha}_{12})(t\mid 0,X)\bigg]
    \end{multlined}\\
    &= P\bigg[\frac{1}{\check\alpha}\int_{0}^{\tau}\bigg\{\check{\pi}(X)\frac{b_{1}}{\check{b}_1}(t\mid X)-\pi(X)\frac{S_{12}}{\check{S}_{12}}(t\!-\!\mid 0,X)\bigg\}\check{q}_1(t\mid X)\d\bigg(1-\frac{S_{\dott 1}}{\check{S}_{\dott 1}}\bigg)(t\mid 0,X)\bigg]\\
    &\hphantom{=}\quad + P\bigg[\frac{\pi(X)}{\check{\alpha}}\int_{0}^{\tau}\bigg\{\frac{b_{2}}{\check{b}_2}(t\mid X)-\frac{S_{\dott 1}}{\check{S}_{\dott 1}}(t\!-\!\mid 0,X)\bigg\}\check{q}_2(t\mid X)\d\bigg(1-\frac{S_{12}}{\check{S}_{12}}\bigg)(t\mid 0,X)\bigg],
  \end{align*}
  where the last step is again by the Duhamel equation in Lemma~\ref{lem:product}.
\end{proof}

Let \(\bar\ell_1(0)\) be obtained by substituting the probability limits of the nuisance parameters into the function \(\check\ell_1(0)\).
\begin{lemma}
  \label{lem:l2-convergence}
  Suppose Assumptions~\ref{asn:plim} and \ref{asn:regularity} hold. Then \(P\{\hat\ell_1(0)-\bar\ell_1(0)\}^2\overset{\mathrm{p}}{\to}0\).
\end{lemma}

\begin{proof}
  We use the notation \(A_n\lesssim B_n\) to denote \(A_n\leq CB_n\) for some universal constant \(C\geq 1\).
  Let
  \begin{align*}
    \hat{b}_{1}(t\mid X) &= \hat\pi(X)\hat{e}_{1}(0\mid X)(\hat{S}_{12}\hat{S}_1^{c})(t\!-\!\mid 0,X) + \{1-\hat\pi(X)\}(\hat{S}_{02}\hat{S}_0^{c})(t\!-\!\mid X),\\
    \bar{b}_{1}(t\mid X) &= \bar\pi(X)\bar{e}_{1}(0\mid X)(\bar{S}_{12}\bar{S}_1^{c})(t\!-\!\mid 0,X) + \{1-\bar\pi(X)\}(\bar{S}_{02}\bar{S}_0^{c})(t\!-\!\mid X),\\
    \hat{b}_{2}(t\mid X) &= \hat{e}_{1}(0\mid X)(\hat{S}_{\dott 1}\hat{S}_1^{c})(t\!-\!\mid 0,X),\\
    \bar{b}_{2}(t\mid X) &=  \bar{e}_{1}(0\mid X)(\bar{S}_{\dott 1}\bar{S}_1^{c})(t\!-\!\mid 0,X),\\
    \hat{r}_{1}(t\mid X) &= \frac{1}{\hat{S}_{\dott 1}(\tau\mid 0,X)}\{\hat{S}_{1}(t\mid 0,X)-\hat{F}_{11}(\tau\mid 0,X)+\hat{F}_{11}(t\mid 0,X)\},\\
    \bar{r}_{1}(t\mid X) &= \frac{1}{\bar{S}_{\dott 1}(\tau\mid 0,X)}\{\bar{S}_{1}(t\mid 0,X)-\bar{F}_{11}(\tau\mid 0,X)+\bar{F}_{11}(t\mid 0,X)\},\\
    \hat{r}_{2}(t\mid X) &= -\frac{1}{\hat{S}_{12}(\tau\mid 0,X)}\{\hat{F}_{11}(\tau\mid 0,X)-\hat{F}_{11}(t\mid 0,X)\},\\
    \bar{r}_{2}(t\mid X) &= -\frac{1}{\bar{S}_{12}(\tau\mid 0,X)}\{\bar{F}_{11}(\tau\mid 0,X)-\bar{F}_{11}(t\mid 0,X)\}.
  \end{align*}
  
  Then
  \begin{align*}
    \hat\ell_{1}(0)(O) &= \frac{1-A}{\hat\alpha}\int_{0}^{\tau}\hat\pi(X)\frac{\hat{r}_1}{\hat{b}_1}(t\mid X)\frac{\hat{S}_{\dott 1}(\tau\mid 0,X)}{\hat{S}_{\dott 1}(t\mid 0,X)}\d\hat{M}_{\dott 1}(t\mid 0,X) \\
                       &\hphantom{=}\quad +\frac{D(1-A)}{\hat\alpha}\int_{0}^{\tau}\frac{\hat{r}_2}{\hat{b}_2}(t\mid X)\frac{\hat{S}_{12}(\tau\mid 0,X)}{\hat{S}_{12}(t\mid 0,X)}\d\hat{M}_{12}(t\mid 0,X) + \frac{D}{\hat\alpha}\hat{F}_{11}(\tau\mid 0,X),\\
    \bar\ell_{1}(0)(O) &= \frac{1-A}{\bar\alpha}\int_{0}^{\tau}\bar\pi(X)\frac{\bar{r}_1}{\bar{b}_1}(t\mid X)\frac{\bar{S}_{\dott 1}(\tau\mid 0,X)}{\bar{S}_{\dott 1}(t\mid 0,X)}\d\bar{M}_{\dott 1}(t\mid 0,X) \\
                       &\hphantom{=}\quad +\frac{D(1-A)}{\bar\alpha}\int_{0}^{\tau}\frac{\bar{r}_2}{\bar{b}_2}(t\mid X)\frac{\bar{S}_{12}(\tau\mid 0,X)}{\bar{S}_{12}(t\mid 0,X)}\d\bar{M}_{12}(t\mid 0,X) + \frac{D}{\bar\alpha}\bar{F}_{11}(\tau\mid 0,X).
  \end{align*}
  
  By Assumption~\ref{asn:regularity}, uniformly for \(t\in(0,\tau]\) and \(x\in\mathcal{X}_1\), \(\hat{b}_{1}\), \(\bar{b}_{1}\), \(\hat{b}_{2}\), and \(\bar{b}_{2}\) are bounded away from \(0\) and from above, \(\hat{r}_{1}\) and \(\bar{r}_{1}\) are positive and bounded from above, while \(\hat{r}_{2}\) and \(\bar{r}_{2}\) negative and bounded from below.

  Decompose the difference as
  \[
    \hat\ell_1(0)-\bar\ell_1(0) = \sum_{m=1}^{10}\delta_m,
  \]
  where
  \begin{align*}
    \delta_1 &= \frac{D}{\hat\alpha}\hat{F}_{11}(\tau\mid 0,X) - \frac{D}{\bar\alpha}\bar{F}_{11}(\tau\mid 0,X),\\
        \delta_2 &= \frac{1-A}{\hat\alpha}\int_{0}^{\tau}(\hat\pi-\bar\pi)(X)\frac{\hat{r}_{1}}{\hat{b}_1}(t\mid X)\frac{\hat{S}_{\dott 1}(\tau\mid 0,X)}{\hat{S}_{\dott 1}(t\mid 0,X)}\d\hat{M}_{\dott 1}(t\mid 0,X),\\
    \delta_3 &= \frac{1-A}{\hat\alpha}\int_{0}^{\tau}\bar\pi(X)\frac{\hat{r}_{1}-\bar{r}_1}{\hat{b}_1}(t\mid X)\frac{\hat{S}_{\dott 1}(\tau\mid 0,X)}{\hat{S}_{\dott 1}(t\mid 0,X)}\d\hat{M}_{\dott 1}(t\mid 0,X),\\
    \delta_4 &= -\frac{1-A}{\hat\alpha\bar\alpha}\int_{0}^{\tau}\bar{\pi}(X)\frac{\bar{r}_{1}(\hat\alpha\hat{b}_1-\bar\alpha\bar{b}_1)}{\hat{b}_1\bar{b}_1}(t\mid X)\frac{\hat{S}_{\dott 1}(\tau\mid 0,X)}{\hat{S}_{\dott 1}(t\mid 0,X)}\d\hat{M}_{\dott 1}(t\mid 0,X),\\
    \delta_5 &= \frac{1-A}{\bar\alpha}\int_{0}^{\tau}\bar\pi(X)\frac{\bar{r}_{1}}{\bar{b}_1}(t\mid X) \bigg\{\frac{\hat{S}_{\dott 1}(\tau\mid 0,X)}{\hat{S}_{\dott 1}(t\mid 0,X)}-\frac{\bar{S}_{\dott 1}(\tau\mid 0,X)}{\bar{S}_{\dott 1}(t\mid 0,X)}\bigg\} \d N_1(t),\\
    \delta_6 &= -\frac{1-A}{\bar\alpha}\int_{0}^{\tau}\bar\pi(X)\frac{\bar{r}_{1}}{\bar{b}_1}(t\mid X) Y(t)\bigg\{\frac{\hat{S}_{\dott 1}(\tau\mid 0,X)}{\hat{S}_{\dott 1}(t\mid 0,X)}\d\hat\Alpha_{\dott 1}(t\mid 0,X)-\frac{\bar{S}_{\dott 1}(\tau\mid 0,X)}{\bar{S}_{\dott 1}(t\mid 0,X)}\d\bar\Alpha_{\dott 1}(t\mid 0,X)\bigg\},\\
    \delta_7 &= \frac{D}{\hat\alpha}(1-A)\int_{0}^{\tau}\frac{\hat{r}_{2}-\bar{r}_{2}}{\hat{b}_2}(t\mid X)\frac{\hat{S}_{12}(\tau\mid 0,X)}{\hat{S}_{12}(t\mid 0,X)}\d\hat{M}_{12}(t\mid 0,X),\\
    \delta_8 &= -\frac{D}{\hat\alpha\bar\alpha}(1-A)\int_{0}^{\tau}\frac{\bar{r}_2(\hat\alpha\hat{b}_2-\bar\alpha\bar{b}_2)}{(\hat{b}_2\bar{b}_2)}(t\mid X)\frac{\hat{S}_{12}(\tau\mid 0,X)}{\hat{S}_{12}(t\mid 0,X)}\d\hat{M}_{12}(t\mid 0,X),\\
    \delta_9 &= \frac{D}{\bar\alpha}(1-A)\int_{0}^{\tau}\frac{\bar{r}_{2}}{\bar{b}_2}(t\mid X)\bigg\{\frac{\hat{S}_{12}(\tau\mid 0,X)}{\hat{S}_{12}(t\mid 0,X)}-\frac{\bar{S}_{12}(\tau\mid 0,X)}{\bar{S}_{12}(t\mid 0,X)}\bigg\} \d N_{2}(t),\\
    \delta_{10} &= -\frac{D}{\bar\alpha}(1-A)\int_{0}^{\tau}\frac{\bar{r}_{2}}{\bar{b}_2}(t\mid X)Y(t)\bigg\{\frac{\hat{S}_{12}(\tau\mid 0,X)}{\hat{S}_{12}(t\mid 0,X)}\d\hat\Alpha_{12}(t\mid 0,X)-\frac{\bar{S}_{12}(\tau\mid 0,X)}{\bar{S}_{12}(t\mid 0,X)}\d\bar\Alpha_{12}(t\mid 0,X)\bigg\},
  \end{align*}
  where in the second to last step, we used Lemma~\ref{lem:integration}.

  We first relate the difference of cumulative incidences for cause 1 to the survival functions as follows: for \(t\in(0,\tau]\),
  \begin{align*}
    \MoveEqLeft|\hat{F}_{11}-\bar{F}_{11}|(t\mid 0,X)\\
    &= \bigg|\int_{0}^{t}(\hat{S}_{\dott 1}\hat{S}_{12})(s\!-\!\mid 0,X)\d\hat{\Alpha}_{\dott 1}(s\mid 0,X) - \int_{0}^{\tau}(\bar{S}_{\dott 1}\bar{S}_{12})(s\!-\!\mid 0,X)\d\bar\Alpha_{\dott 1}(s\mid 0,X)\bigg| \\
    &\leq  \bigg|\int_{0}^{t}\{(\hat{S}_{12}-\bar{S}_{12})\hat{S}_{\dott 1}\}(s\!-\!\mid 0,X)\d\hat{\Alpha}_{\dott 1}(s\mid 0,X)\bigg|\\
    &\hphantom{\leq}\quad +\bigg|\int_{0}^{t}\bar{S}_{12}(s\!-\!\mid 0,X)\big\{\hat{S}_{\dott 1}(s\!-\!\mid 0,X)\d\hat\Alpha_{\dott 1}(s\mid 0,X)-\bar{S}_{\dott 1}(s\!-\!\mid 0,X)\d\bar\Alpha_{\dott 1}(s\mid 0,X)\big\}\bigg|\\
    &\leq  \sup_{s\in(0,\tau]}|\hat{S}_{12}-\bar{S}_{12}|(s\mid 0,X)\hat{S}_{\dott 1}(\tau\mid 0,X)\\
    &\hphantom{\leq}\quad +\bigg|\{(\bar{S}_{\dott 1}-\hat{S}_{\dott 1})\bar{S}_{12}\}(t\mid 0,X)-\int_{0}^{t}(\bar{S}_{\dott 1}-\hat{S}_{\dott 1})(s\mid 0,X)\d\bar{S}_{12}(s\mid 0,X)\bigg|\\
    &\leq  \sup_{s\in(0,\tau]}|\hat{S}_{12}-\bar{S}_{12}|(s\mid 0,X) + \sup_{s\in(0,\tau]}|\hat{S}_{\dott 1}-\bar{S}_{\dott 1}|(s\mid 0,X).
  \end{align*}

  By the triangular inequality, \(\|\hat\ell_1(0)-\bar\ell_1(0)\|_{P} \leq \sum_{m=1}^{10}\{P\delta_m^2\}^{1/2}\).
  Below we bound each term \(P\delta_m^2\).
  Let \(P_1\) denote the probability measure \(P(\cdot\mid D=1)\).

  \paragraph{Term \(\delta_1\)}
  \begin{align*}
    P\delta_1^2 &= P\bigg\{\frac{D}{\hat\alpha}\hat{F}_{11}(\tau\mid 0,X) - \frac{D}{\bar\alpha}\bar{F}_{11}(\tau\mid 0,X)\bigg\}^2 \\
                &\lesssim P\bigg[\frac{D}{\hat\alpha}\{\hat{F}_{11}(\tau\mid 0,X) -\bar{F}_{11}(\tau\mid 0,X)\}\bigg]^2 + P\bigg\{D\frac{\hat\alpha-\bar\alpha}{\bar\alpha\hat\alpha}\bar{F}_{11}(\tau\mid 0,X)\bigg\}^2 \\
                &\lesssim P_1(\hat{F}_{11}-\bar{F}_{11})^2(\tau\mid 0,X)+ |\hat\alpha-\bar\alpha|^2\\
                &\leq |\hat\alpha-\bar\alpha|^2 + P_1\bigg\{\sup_{t\in(0,\tau]}|\hat{S}_{12}-\bar{S}_{12}|(t\mid 0,X)\bigg\}^2 + P_1\bigg\{\sup_{t\in(0,\tau]}|\hat{S}_{\dott 1}-\bar{S}_{\dott 1}|(t\mid 0,X)\bigg\}^2.
  \end{align*}

  \paragraph{Term \(\delta_2\)}
  \begin{align*}
    P\delta_2^2 &= P\bigg[\frac{1-A}{\hat\alpha}\int_{0}^{\tau}(\hat\pi-\bar\pi)(X)\frac{\hat{r}_{1}}{\hat{b}_1}(t\mid X)\frac{\hat{S}_{\dott 1}(\tau\mid 0,X)}{\hat{S}_{\dott 1}(t\mid 0,X)}\d\hat{M}_{\dott 1}(t\mid 0,X)\bigg]^2 \\
                &\lesssim P\bigg[(\hat\pi-\bar\pi)(X)\int_{0}^{\tau}\frac{\hat{S}_{\dott 1}(\tau\mid 0,X)}{\hat{S}_{\dott 1}(t\mid 0,X)}\d\hat{M}_{\dott 1}(t\mid 0,X)\bigg]^2\\
                &\leq P\bigg[|\hat\pi-\bar\pi|(X)\int_{0}^{\tau}\frac{\hat{S}_{\dott 1}(\tau\mid 0,X)}{\hat{S}_{\dott 1}(t\mid 0,X)}\{\d N_1(t) + Y(t)\d \hat{\Alpha}_{\dott 1}(t\mid 0,X)\}\bigg]^2\\
                &\lesssim P|\hat\pi-\bar\pi|^2(X),
  \end{align*}
  where in the last step we used 
  \[
    \int_{0}^{\tau}\frac{\hat{S}_{\dott 1}(\tau\mid 0,X)}{\hat{S}_{\dott 1}(t\mid 0,X)}\d N_1(t) = \I(\tilde{T}\leq \tau,\tilde{J}=1)\frac{\hat{S}_{\dott 1}(\tau\mid 0,X)}{\hat{S}_{\dott 1}(\tilde{T}\mid 0,X)}\leq 1,
  \]
  and, by the backward equation in Lemma~\ref{lem:product},
  \[
    \int_{0}^{\tau}\frac{\hat{S}_{\dott 1}(\tau\mid 0,X)}{\hat{S}_{\dott 1}(t\mid 0,X)}Y(t)\d \hat{\Alpha}_{\dott 1}(t) = \frac{\hat{S}_{\dott 1}(\tau\mid 0,X)}{\hat{S}_{\dott 1}(\tilde{T}\wedge \tau\mid 0,X)}-\hat{S}_{\dott 1}(\tau\mid 0,X)\leq 2.
  \]

  \paragraph{Term \(\delta_3\)}
  \begin{align*}
    P\delta_3^2 &= P\bigg[\frac{1-A}{\hat\alpha}\int_{0}^{\tau}\bar\pi(X)\frac{\hat{r}_{1}-\bar{r}_1}{\hat{b}_1}(t\mid X)\frac{\hat{S}_{\dott 1}(\tau\mid 0,X)}{\hat{S}_{\dott 1}(t\mid 0,X)}\d\hat{M}_{\dott 1}(t\mid 0,X)\bigg]^2 \\
                &\lesssim P\bigg[\int_{0}^{\tau}(\hat{r}_{1}-\bar{r}_1)(t\mid X)\frac{\hat{S}_{\dott 1}(\tau\mid 0,X)}{\hat{S}_{\dott 1}(t\mid 0,X)}\d\hat{M}_{\dott 1}(t\mid 0,X)\bigg]^2\\
                &= P\bigg\{\int_{0}^{\tau}\bigg(\frac{\bar{S}_{\dott 1}-\hat{S}_{\dott 1}}{\bar{S}_{\dott 1}}(\tau\mid 0,X)\hat{r}_1(t\mid X)+ \frac{1}{\bar{S}_{\dott 1}(\tau\mid 0,X)}\big[\{(\hat{S}_{\dott 1}-\bar{S}_{\dott 1})\}\hat{S}_{12}(t\mid 0,X) \\
                &\hphantom{= P\bigg\{\int_{0}^{\tau}\bigg\{\quad}+\{\bar{S}_{\dott 1}(\hat{S}_{12}-\bar{S}_{12})\}(t\mid 0,X)-(\hat{F}_{11}-\bar{F}_{11})(\tau\mid 0,X)+(\hat{F}_{11}-\bar{F}_{11})(t\mid 0,X)\big]\bigg)\\
                &\lesssim P\bigg[\int_{0}^{\tau}\big\{|\bar{S}_{\dott 1}-\hat{S}_{\dott 1}|(\tau\mid 0,X)+ |\hat{S}_{\dott 1}-\bar{S}_{\dott 1}|(t\mid 0,X) \\
                &\hphantom{= P\bigg\{\int_{0}^{\tau}\bigg\{\quad}+|\hat{S}_{12}-\bar{S}_{12}|(t\mid 0,X)|+|\hat{F}_{11}-\bar{F}_{11}|(\tau\mid 0,X)+|\hat{F}_{11}-\bar{F}_{11}|(t\mid 0,X)\big\}\\
                &\hphantom{= P\bigg[\int_{0}^{\tau}\bigg\{\quad}\frac{\hat{S}_{\dott 1}(\tau\mid 0,X)}{\hat{S}_{\dott 1}(t\mid 0,X)}\{\d N_1(t) + Y(t)\d \hat{\Alpha}_{\dott 1}(t\mid 0,X)\}\bigg]^2 \\
                &\lesssim P_1\bigg[\bigg\{\sup_{t\in(0,\tau]}|\hat{S}_{\dott 1}-\bar{S}_{\dott 1}|(t\mid 0,X) + \sup_{t\in(0,\tau]}|\hat{S}_{12}-\bar{S}_{12}|(t\mid 0,X)\\
                &\hphantom{\lesssim P_1\bigg[\bigg\{} \quad + \sup_{t\in(0,\tau]}|\hat{F}_{11}-\bar{F}_{11}|(t\mid 0,X)\bigg\}\int_{0}^{\tau}\frac{\hat{S}_{\dott 1}(\tau\mid 0,X)}{\hat{S}_{\dott 1}(t\mid 0,X)}\{\d N_1(t) + Y(t)\d \hat{\Alpha}_{\dott 1}(t\mid 0,X)\}\bigg]\\
                &\lesssim P_1\bigg\{\sup_{t\in(0,\tau]}|\hat{S}_{\dott 1}-\bar{S}_{\dott 1}|(t\mid 0,X)\bigg\}^2 + P_1\bigg\{\sup_{t\in(0,\tau]}|\hat{S}_{12}-\bar{S}_{12}|(t\mid 0,X)\bigg\}^2.
  \end{align*}

  \paragraph{Term \(\delta_4\)}
  \begin{align*}
    P\delta_4^2 &= P\bigg\{\frac{1-A}{\hat\alpha\bar\alpha}\int_{0}^{\tau}\bar\pi(X)\frac{\bar{r}_{1}(\hat\alpha\hat{b}_1-\bar\alpha\bar{b}_1)}{\hat{b}_1\bar{b}_1}(t\mid X)\frac{\hat{S}_{\dott 1}(\tau\mid 0,X)}{\hat{S}_{\dott 1}(t\mid 0,X)}\d\hat{M}_{\dott 1}(t\mid 0,X)\bigg\}^2 \\
                &\lesssim P_1\bigg\{\int_{0}^{\tau}(\hat\alpha\hat{b}_1-\bar\alpha\bar{b}_1)(t\mid X)\frac{\hat{S}_{\dott 1}(\tau\mid 0,X)}{\hat{S}_{\dott 1}(t\mid 0,X)}\d\hat{M}_{\dott 1}(t\mid 0,X)\bigg\}^2\\
                &= P_1\bigg\{\int_{0}^{\tau}\big[(\hat\alpha-\bar\alpha)\hat{b}_1(t\mid X) + \bar\alpha(\hat\pi-\bar\pi)(X)\hat{e}_1(0\mid X)(\hat{S}_{12}\hat{S}_1^c)(t\!-\!\mid 0,X)\\
                &\hphantom{\leq} +\bar\alpha\bar\pi(X)(\hat{e}_1-\bar{e}_1)(0\mid X)(\hat{S}_{12}\hat{S}_{1}^c)(t\!-\!\mid 0,X) \\
                &\hphantom{\leq }+ \bar\alpha\bar\pi(X)\bar{e}_1(0\mid X)\{(\hat{S}_{12}-\bar{S}_{12})\hat{S}_{1}^c+\bar{S}_{12}(\hat{S}_{1}^c-\bar{S}_{1}^c)\}(t\!-\!\mid 0,X)\\
                &\hphantom{\leq }+ \bar\alpha(\bar\pi-\hat\pi)(X)(\hat{S}_{02}\hat{S}_0^c)(t\!-\!\mid X) \\
                &\hphantom{\leq}+ \bar\alpha\{1-\bar\pi(X)\}\{(\hat{S}_{02}-\bar{S}_{02})\hat{S}_0^c+\bar{S}_{02}(\hat{S}_{0}^c-\bar{S}_{0}^c)\}(t\!-\!\mid X)\big]\\
                &\hphantom{\leq P_1\bigg\{}\qquad \frac{\hat{S}_{\dott 1}(\tau\mid 0,X)}{\hat{S}_{\dott 1}(t\mid 0,X)}\d \hat{M}_{\dott 1}(t)\bigg\}^2 \\
                &\lesssim P\bigg[\big\{|\hat\alpha-\bar\alpha|+|\hat\pi-\bar\pi|(X)\big\}\int_0^\tau\frac{\hat{S}_{\dott 1}(\tau\mid 0,X)}{\hat{S}_{\dott 1}(t\mid 0,X)}\{\d N_1(t) + Y(t)\d \hat{\Alpha}_{\dott 1}(t\mid 0,X)\}\bigg]^2\\
                &\hphantom{\lesssim}\quad + P\I\{\pi(X)>0\}\bigg[\int_0^\tau\big\{|\hat{e}_1-\bar{e}_1|(0\mid X)+|\hat{S}_{12}-\bar{S}_{12}|(t\!-\!\mid 0,X)+|\hat{S}_1^c-\bar{S}_1^c|(t\!-\!\mid 0,X)\big\}\\
                &\hphantom{\lesssim}\quad \frac{\hat{S}_{\dott 1}(\tau\mid 0,X)}{\hat{S}_{\dott 1}(t\mid 0,X)}\{\d N_1(t) + Y(t)\d \hat{\Alpha}_{\dott 1}(t\mid 0,X)\}\bigg]^2 \\
                &\hphantom{\lesssim}\quad + P\I[\pi(X)\{1-\pi(X)\}>0]\bigg[\int_0^\tau\big\{|\hat{S}_{02}-\bar{S}_{02}|(t\!-\!\mid 0,X)+|\hat{S}_0^c-\bar{S}_0^c|(t\!-\!\mid 0,X)\big\}\\
                &\hphantom{\lesssim}\quad \frac{\hat{S}_{\dott 1}(\tau\mid 0,X)}{\hat{S}_{\dott 1}(t\mid 0,X)}\{\d N_1(t) + Y(t)\d \hat{\Alpha}_{\dott 1}(t\mid 0,X)\}\bigg]^2 \\
                &\lesssim |\hat\alpha-\bar\alpha|^2 + P|\hat\pi-\bar\pi|^2(X)+ P\I\{\pi(X)>0\}|\hat{e}_1-\bar{e}_1|^2(0\mid X) \\
                &\hphantom{\lesssim}\quad + P\I\{\pi(X)>0\}\bigg\{\sup_{t\in(0,\tau]}|\hat{S}_{\dott 1}-\bar{S}_{\dott 1}|(t\mid 0,X)\bigg\}^2\\
                &\hphantom{\lesssim}\quad + P\I\{\pi(X)>0\}\bigg\{\sup_{t\in(0,\tau]}|\hat{S}_{12}-\bar{S}_{12}|(t\mid 0,X)\bigg\}^2 \\
                &\hphantom{\lesssim}\quad + P\I[\pi(X)\{1-\pi(X)\}>0]\bigg\{\sup_{t\in(0,\tau]}|\hat{S}_{02}-\bar{S}_{02}|(t\mid X)\bigg\}^2 \\
                &\hphantom{\lesssim}\quad  + P\I[\pi(X)\{1-\pi(X)\}>0]\bigg\{\sup_{t\in(0,\tau]}|\hat{S}_{1}^c-\bar{S}_{1}^c|(t\mid 0,X)\bigg\}^2 \\
                &\hphantom{\lesssim}\quad + P\I[\pi(X)\{1-\pi(X)\}>0]\bigg\{\sup_{t\in(0,\tau]}|\hat{S}_{0}^c-\bar{S}_{0}^c|(t\mid X)\bigg\}^2.
  \end{align*}

  \paragraph{Term \(\delta_5\)}  
  \begin{align*}
    P\delta_5^2 &= P\bigg\{\frac{1-A}{\bar\alpha}\int_{0}^{\tau}\bar\pi(X)\frac{\bar{r}_{1}}{\bar{b}_1}(t\mid X) \bigg\{\frac{\hat{S}_{\dott 1}(\tau\mid 0,X)}{\hat{S}_{\dott 1}(t\mid 0,X)}-\frac{\bar{S}_{\dott 1}(\tau\mid 0,X)}{\bar{S}_{\dott 1}(t\mid 0,X)}\bigg\} \d N_1(t)\bigg\}^2\\
                &\lesssim P\I\{\pi(X)>0\}\I\{\tilde{T}\leq \tau,\tilde{J}=1\}\bigg|\frac{\hat{S}_{\dott 1}(\tau\mid 0,X)}{\hat{S}_{\dott 1}(\tilde{T}\mid 0,X)}-\frac{\bar{S}_{\dott 1}(\tau\mid 0,X)}{\bar{S}_{\dott 1}(\tilde{T}\mid 0,X)}\bigg|^2 \\
                &\lesssim P\I\{\pi(X)>0\}\bigg\{\sup_{t\in(0,\tau]}|\hat{S}_{\dott 1}-\bar{S}_{\dott 1}|(t\mid 0,X)\bigg\}^2.
  \end{align*}

  \paragraph{Term \(\delta_6\)}
  \begin{align*}
    P\delta_6^2 &= \begin{multlined}[t][.9\textwidth]
    P\bigg[\frac{1-A}{\bar\alpha}\int_{0}^{\tau}\bar{\pi}(X)\frac{\bar{r}_{1}}{\bar{b}_1}(t\mid X) Y(t)\\
    \bigg\{\frac{\hat{S}_{\dott 1}(\tau\mid 0,X)}{\hat{S}_{\dott 1}(t\mid 0,X)}\d\hat\Alpha_{\dott 1}(t\mid 0,X)-\frac{\bar{S}_{\dott 1}(\tau\mid 0,X)}{\bar{S}_{\dott 1}(t\mid 0,X)}\d\bar\Alpha_{\dott 1}(t\mid 0,X)\bigg\}\bigg]^2
    \end{multlined}\\
                &\lesssim P\I\{\pi(X)>0\}\bigg[\int_{0}^{\tau}\frac{\bar{r}_{1}}{\bar{b}_1}(t\mid X) Y(t)\d\bigg\{\frac{\hat{S}_{\dott 1}(\tau\mid 0,X)}{\hat{S}_{\dott 1}(t\mid 0,X)}-\frac{\bar{S}_{\dott 1}(\tau\mid 0,X)}{\bar{S}_{\dott 1}(t\mid 0,X)}\bigg\}\bigg]^2.
  \end{align*}
  By Assumption~\ref{asn:regularity}, \(\{\hat\Alpha_{\dott 1},\bar\Alpha_{\dott 1}\}\) do not share any discontinuity with \(\{\bar\Alpha_{12},\bar\Alpha_{02},\bar\Alpha_{1}^c,\bar\Alpha_{0}^c\}\), then we may replace \(\bar{b}(t\mid X)\) in the display before with the right-continuous version \(\bar{b}(t\!+\!\mid X)\), which equals 
  \[
    \bar\pi(X)\bar{e}_{1}(0\mid X)(\bar{S}_{12}\bar{S}_1^{c})(t\mid 0,X) + \{1-\bar\pi(X)\}(\bar{S}_{02}\bar{S}_0^{c})(t\mid X).
  \]
  Therefore, we can apply integration by parts from Lemma~\ref{lem:integration}. This leads to
  \begin{align*}
    P\delta_6^2 &= P\I\{\pi(X)>0\}\bigg[\frac{\bar{r}_{1}}{\bar{b}_1}(t\mid X)\bigg\{\frac{\hat{S}_{\dott 1}(\tau\mid 0,X)}{\hat{S}_{\dott 1}(t\mid 0,X)}-\frac{\bar{S}_{\dott 1}(\tau\mid 0,X)}{\bar{S}_{\dott 1}(t\mid 0,X)}\bigg\}\bigg\vert_{0}^{\tau\wedge\tilde{T}} \\
                &\hphantom{=P_1\bigg[}\quad - \int_{0}^{\tau\wedge \tilde{T}}\bigg\{\frac{\hat{S}_{\dott 1}(\tau\mid 0,X)}{\hat{S}_{\dott 1}(t\!-\!\mid 0,X)}-\frac{\bar{S}_{\dott 1}(\tau\mid 0,X)}{\bar{S}_{\dott 1}(t\!-\!\mid 0,X)}\bigg\}\d\bigg(\frac{\bar{r}_{1}}{\bar{b}_1}\bigg)(t\mid X)\bigg]^2\\
                &\lesssim P\I\{\pi(X)>0\}\bigg\{\sup_{t\in(0,\tau]}|\hat{S}_{\dott 1}-\bar{S}_{\dott 1}|(t\mid 0,X)\bigg\}^2 \\
                &\hphantom{\lesssim}\quad + P\I\{\pi(X)>0\}\bigg[\bigg\{\sup_{t\in(0,\tau\wedge \tilde{T}]}\bigg|\frac{\hat{S}_{\dott 1}(\tau\mid 0,X)}{\hat{S}_{\dott 1}(t\!-\!\mid 0,X)}-\frac{\bar{S}_{\dott 1}(\tau\mid 0,X)}{\bar{S}_{\dott 1}(t\!-\!\mid 0,X)}\bigg|\bigg\}\bigg\{\frac{-\bar{S}_{1}(t\mid 0,X)}{\bar{b}_1(t\mid X)}\bigg\}\bigg|_{0}^{\tau\wedge\tilde{T}}\bigg]^2\\
                &\hphantom{\lesssim}\quad + P\I\{\pi(X)>0\}\bigg[\bigg\{\sup_{t\in(0,\tau\wedge \tilde{T}]}\bigg|\frac{\hat{S}_{\dott 1}(\tau\mid 0,X)}{\hat{S}_{\dott 1}(t\!-\!\mid 0,X)}-\frac{\bar{S}_{\dott 1}(\tau\mid 0,X)}{\bar{S}_{\dott 1}(t\!-\!\mid 0,X)}\bigg|\bigg\}\\
                &\qquad\qquad\qquad\qquad\qquad\qquad\qquad \bigg\{\frac{\bar{F}_{11}(t\mid 0,X)-\bar{F}_{11}(\tau\mid 0,X)}{\bar{b}_1(t\mid X)}\bigg\}\bigg|_{0}^{\tau\wedge\tilde{T}}\bigg]^2\\
                &\lesssim P\I\{\pi(X)>0\}\bigg\{\sup_{t\in(0,\tau]}|\hat{S}_{\dott 1}-\bar{S}_{\dott 1}|^2(t\mid 0,X)\bigg\}^2.
  \end{align*}

  The remaining terms \(P\delta_7^2\), \(P\delta_8^2\), \(P\delta_9^2\), and \(P\delta_{10}^2\) can be analogously bounded as \(P\delta_3^2\), \(P\delta_4^2\), \(P\delta_5^2\), and \(P\delta_6^2\), respectively.
  Now, for any \(\Alpha_1,\Alpha_2\in\mathcal{A}\) with product integrals \(S_1\) and \(S_2\),
  \begin{align*}
  \sup_{t\in(0,\tau]}|S_1-S_2|(t) &= \sup_{t\in(0,\tau]}\bigg|\int_0^t\frac{S_2(t)}{S_2(s)}S_1(s-)\d(A_2-A_1)(s)\bigg| \\
  &\leq \sup_{t\in(0,\tau]}\int_0^t\d|A_2-A_1|(s)\\
  &\leq \sup_{t\in(0,\tau]}|A_2-A_1|(t).
  \end{align*}
  Therefore, by Assumption~\ref{asn:plim}, the terms \(P\delta_m^2\) are all \(o_{P}(1)\), and the lemma follows.
\end{proof}

\begin{proof}[Proof of Theorem~\ref{thm:asymptotic}]
  We first show consistency of the estimator \(\hat\theta_1(0)\).
  Decompose the bias as
  \[
    \hat\theta_1(0)-\theta_1(0) = (\mathbb{P}_n-P)\hat\ell_1(0) + P\{\hat\ell_1(0)-\bar\ell_1(0)\} +P\bigg\{\bar\ell_1(0)-\frac{\alpha}{\bar\alpha}\ell_1(0)\bigg\}-\frac{\hat\alpha-\alpha}{\hat\alpha}\theta_1(0).
  \]
  The first term is \(o_{P}(1)\) by the uniform law of large numbers because \(\hat{\ell}_1(0)\) belongs to a \(P\)-Glivenko-Cantelli class.
  The second term is bounded by Jensen's inequality as \(P\{\hat\ell_1(0)-\bar\ell_1(0)\}\leq [P\{\hat\ell_1(0)-\bar\ell_1(0)\}^2]^{1/2}\), which converges in probability to zero by Lemma~\ref{lem:l2-convergence}.
  The third term is exactly \(0\) by \(\bar\alpha = \alpha\), the assumption on the correct specifications of the nuisance estimators, and Lemma~\ref{lem:diff}.
  The fourth term is trivially \(o_{P}(1)\) from \(\hat\alpha\overset{\mathrm{p}}{\to}\alpha\) and Slutsky's theorem.
  Therefore, \(\hat\theta_1(0)\overset{\mathrm{p}}{\to}\theta_1(0)\).

  Then we show asymptotic linearity of \(\hat\theta_1(0)\).
  We decompose the bias again as
  \[
    \hat\theta_1(0)-\theta_1(0) = (\mathbb{P}_n-P)\varphi_1(0) + (\mathbb{P}_n-P)\{\hat\ell_1(0)-\ell_1(0)\} + P\bigg\{\hat\ell_1(0)-\frac{\alpha}{\hat\alpha}\ell_1(0)\bigg\}+\frac{(\hat\alpha-\alpha)^{2}}{\hat\alpha\alpha}\theta_1(0).
  \]
  The second term is \(o_{P}(n^{-1/2})\) by Lemma 19.24 in \citetsuppmat{vandervaart1998asymptotic} because \(P\{\hat\ell_1(0)-\ell_1(0)\}^2\overset{\mathrm{p}}{\to}0\) by Lemma~\ref{lem:l2-convergence}.
  The third term is \(o_{P}(n^{-1/2})\) by assumption and Lemma~\ref{lem:diff}.
  The fourth term is trivially \(O_{P}(n^{-1})=o_{P}(n^{-1/2})\) from \(n^{1/2}(\hat\alpha-\alpha)\overset{\mathrm{d}}{\to}\mathrm{Normal}\{0,\alpha(1-\alpha)\}\), \(\hat\alpha\overset{\mathrm{p}}{\to}\alpha\), and Slutsky's theorem.
  Therefore, \(\hat\theta_1(0)-\theta_1(0)=\mathbb{P}_n\varphi_1(0)+o_{P}(n^{-1/2})\), since \(E_P\varphi_1(0)=0\).

\end{proof}

\section{Details on the simulation study}
\label{sec:sim-app}

Define
\begin{align*}
  \hat{W}_{kj}(t,s\mid a,x) &= \I(j=k)\hat{S}_{1}^c(s\!-\!\mid a,x)-\frac{\hat{F}_{1j}(t\mid a,x)-\hat{F}_{1j}(s\mid a,x)}{1-\triangle\hat{\Alpha}_{1k}(s\mid a,x)},\quad (k\neq 1, a\neq 0)\\
  \hat{W}_{\dott j}(t,s\mid 0,x) &= \I(j=1)\hat{S}_{1}^c(s\!-\!\mid 0,x)-\frac{\hat{F}_{1j}(t\mid 0,x)-\hat{F}_{1j}(s\mid 0,x)}{1-\triangle\hat{\Alpha}_{\dott 1}(s\mid 0,x)}.
\end{align*}
Consider the following functions with plug-in nuisance estimators:
\begin{align*}
  \hat\ell_{j}(0,t)(O) &= \frac{1-A}{\hat\alpha}\hat\pi(X)\int_{0}^{t}\frac{\hat{W}_{\dott j}(t,s\mid 0,X)}{\hat{H}_{\dott}(s\!-\!\mid X)}\d\hat{M}_{\dott 1}(s\mid 0,X) \\
                     &\hphantom{=}\quad +\frac{D(1-A)}{\hat\alpha}\int_{0}^{t}\frac{\hat{W}_{2j}(t,s\mid 0,X)}{\hat{H}_{1}(s\!-\!\mid 0,X)}\d\hat{M}_{12}(s\mid 0,X) + \frac{D}{\hat\alpha}\hat{F}_{1j}(t\mid 0,X),\\
  \hat\ell_{j}(1,t)(O) &= \sum_{k\in\{1,2\}}\frac{DA}{\hat\alpha}\int_{0}^{t}\frac{\hat{W}_{kj}(t,s\mid 1,X)}{\hat{H}_1(s\!-\!\mid 1,X)}\d\hat{M}_{1k}(s\mid 1,X) + \frac{D}{\hat\alpha}\hat{F}_{1j}(t\mid 1,X),\\
  \hat\ell_{j}^\dagger(0,t)(O) &= \sum_{k\in\{1,2\}}\frac{D(1-A)}{\hat\alpha}\int_{0}^{t}\frac{\hat{W}_{kj}(t,s\mid 0,X)}{\hat{H}_1(s\!-\!\mid 0,X)}\d\hat{M}_{1k}(s\mid 0,X) + \frac{D}{\hat\alpha}\hat{F}_{1j}(t\mid 0,X).
\end{align*}
Let \(\mathbb{P}_nf=n^{-1}\sum_{i=1}^{n}f(O_i)\).
Then we have the corresponding estimators
\begin{align*}
  \hat\theta_j(a,t) &= \mathbb{P}_n\hat\ell_j(a,t)(O),\quad &\hat\theta_j^\dagger(0,t) &= \mathbb{P}_n\hat\ell_j^\dagger(0,t)(O),\\
  \hat\gamma_j(a,t) &= \mathbb{P}_n\int_0^t\hat\ell_j(a,s)(O)\d s,\quad &\hat\gamma_j^\dagger(a,t) &= \mathbb{P}_n\int_0^t\hat\ell_j^\dagger(a,s)(O)\d s,
\end{align*}
where \(\hat\theta_j^\dagger(0,t)\) and \(\hat\gamma_j^\dagger(0,t)\) are the RCT-only estimators for the parameters \(\theta_j(0,t)\) and \(\gamma_j(0,t)\).

Tables~\ref{tab:sim-cif-app-1}--\ref{tab:sim-rmtl-app-2} display the summary statistics for the estimand \(\theta_{1}(1,t)\), the set of estimands for cause 2 \(\{\theta_2(1,t),\theta_2(0,t),\theta_2\{t\}\}\), the estimand \(\gamma_{1}(1,t)\), and the set of estimands for cause 2 \(\{\gamma_2(1,t),\gamma_2(0,t),\gamma_2\{t\}\}\), respectively.

\begin{table}
  \caption{Simulation results for cumulative incidences \(\theta_1(1,t)\).}
  \label{tab:sim-cif-app-1}
  \footnotesize
  \centering
  
\begin{tabular}{rlrrrrr}
{\(n\)} & {\(t\)} & {Mean} & {Bias} & {RMSE} & {SE} & {Coverage}\\
\(750\) & \(0.25\) & \(0.12\) & \(2.70\) & \(2.54\) & \(2.52\) & \(93.2\)\\
 & \(1\) & \(0.22\) & \(12.19\) & \(3.41\) & \(3.43\) & \(95.2\)\\
 & \(2\) & \(0.28\) & \(16.69\) & \(3.77\) & \(3.83\) & \(95.4\)\\
\(1500\) & \(0.25\) & \(0.12\) & \(-1.26\) & \(1.82\) & \(1.78\) & \(93.0\)\\
 & \(1\) & \(0.22\) & \(-3.30\) & \(2.56\) & \(2.43\) & \(93.4\)\\
 & \(2\) & \(0.27\) & \(-10.69\) & \(2.81\) & \(2.71\) & \(93.0\)\\
\end{tabular}

  \medskip
  \begin{flushleft}
    Mean: average of estimates; Bias: Monte-Carlo bias, \(10^{-4}\); RMSE: root mean squared error, \(10^{-2}\); SE: average of standard error estimates, \(10^{-2}\); Coverage: \(95\%\) confidence interval coverage, \(\%\).
  \end{flushleft}
\end{table}

\begin{table}
  \caption{Simulation results for cumulative incidences \(\theta_2(1,t)\), \(\theta_2(0,t)\), and \(\theta_2(t)\).}
  \label{tab:sim-cif-app-2}
  \footnotesize
  \centering
  
\begin{tabular}{rlrlrrrrrr}
{\(n\)} & {Estimand} & {\(t\)} & {Type} & {Mean} & {Bias} & {RMSE} & {SE} & {Coverage} & {Reduction}\\
\(750\) & \(\theta_2(0,t)\) & \(0.25\) & \(+\) & \(0.23\) & \(13.18\) & \(3.37\) & \(3.27\) & \(93.8\) & \(0.94\)\\
 &  &  & \(-\) & \(0.23\) & \(11.83\) & \(3.41\) & \(3.28\) & \(93.4\) & \(\cdot\)\\
 &  & \(1\) & \(+\) & \(0.44\) & \(-2.74\) & \(3.87\) & \(3.94\) & \(95.1\) & \(5.21\)\\
 &  &  & \(-\) & \(0.43\) & \(-6.66\) & \(3.97\) & \(4.05\) & \(95.3\) & \(\cdot\)\\
 &  & \(2\) & \(+\) & \(0.55\) & \(3.00\) & \(4.10\) & \(4.01\) & \(93.9\) & \(10.84\)\\
 &  &  & \(-\) & \(0.54\) & \(1.31\) & \(4.33\) & \(4.25\) & \(94.3\) & \(\cdot\)\\
 & \(\theta_2(1,t)\) & \(0.25\) & \(-\) & \(0.23\) & \(-3.32\) & \(3.28\) & \(3.28\) & \(94.7\) & \(\cdot\)\\
 &  & \(1\) & \(-\) & \(0.42\) & \(-14.16\) & \(4.06\) & \(4.02\) & \(94.1\) & \(\cdot\)\\
 &  & \(2\) & \(-\) & \(0.51\) & \(-10.60\) & \(4.36\) & \(4.23\) & \(93.9\) & \(\cdot\)\\
 & \(\theta_2(t)\) & \(0.25\) & \(+\) & \(0.00\) & \(-16.49\) & \(4.55\) & \(4.54\) & \(95.1\) & \(0.50\)\\
 &  &  & \(-\) & \(0.00\) & \(-15.15\) & \(4.58\) & \(4.55\) & \(95.1\) & \(\cdot\)\\
 &  & \(1\) & \(+\) & \(-0.02\) & \(-11.42\) & \(5.45\) & \(5.48\) & \(95.0\) & \(2.78\)\\
 &  &  & \(-\) & \(-0.01\) & \(-7.50\) & \(5.55\) & \(5.56\) & \(94.9\) & \(\cdot\)\\
 &  & \(2\) & \(+\) & \(-0.03\) & \(-13.60\) & \(5.71\) & \(5.69\) & \(94.8\) & \(5.72\)\\
 &  &  & \(-\) & \(-0.03\) & \(-11.91\) & \(5.94\) & \(5.86\) & \(94.5\) & \(\cdot\)\\
\(1500\) & \(\theta_2(0,t)\) & \(0.25\) & \(+\) & \(0.23\) & \(-1.05\) & \(2.26\) & \(2.31\) & \(95.4\) & \(0.91\)\\
 &  &  & \(-\) & \(0.23\) & \(-1.61\) & \(2.27\) & \(2.32\) & \(95.0\) & \(\cdot\)\\
 &  & \(1\) & \(+\) & \(0.44\) & \(6.17\) & \(2.81\) & \(2.80\) & \(94.1\) & \(5.13\)\\
 &  &  & \(-\) & \(0.44\) & \(5.88\) & \(2.86\) & \(2.87\) & \(94.5\) & \(\cdot\)\\
 &  & \(2\) & \(+\) & \(0.55\) & \(1.46\) & \(2.86\) & \(2.86\) & \(94.0\) & \(10.70\)\\
 &  &  & \(-\) & \(0.55\) & \(2.04\) & \(3.02\) & \(3.03\) & \(94.5\) & \(\cdot\)\\
 & \(\theta_2(1,t)\) & \(0.25\) & \(-\) & \(0.23\) & \(-7.69\) & \(2.30\) & \(2.32\) & \(94.7\) & \(\cdot\)\\
 &  & \(1\) & \(-\) & \(0.42\) & \(-12.76\) & \(2.84\) & \(2.84\) & \(94.4\) & \(\cdot\)\\
 &  & \(2\) & \(-\) & \(0.51\) & \(-12.98\) & \(2.97\) & \(3.01\) & \(95.1\) & \(\cdot\)\\
 & \(\theta_2(t)\) & \(0.25\) & \(+\) & \(0.00\) & \(-6.64\) & \(3.17\) & \(3.21\) & \(95.4\) & \(0.48\)\\
 &  &  & \(-\) & \(0.00\) & \(-6.08\) & \(3.17\) & \(3.21\) & \(95.6\) & \(\cdot\)\\
 &  & \(1\) & \(+\) & \(-0.02\) & \(-18.93\) & \(3.84\) & \(3.88\) & \(95.5\) & \(2.74\)\\
 &  &  & \(-\) & \(-0.02\) & \(-18.64\) & \(3.87\) & \(3.93\) & \(96.1\) & \(\cdot\)\\
 &  & \(2\) & \(+\) & \(-0.03\) & \(-14.44\) & \(3.92\) & \(4.05\) & \(95.7\) & \(5.65\)\\
 &  &  & \(-\) & \(-0.03\) & \(-15.02\) & \(4.04\) & \(4.17\) & \(96.0\) & \(\cdot\)\\
\end{tabular}

  \medskip
  \begin{flushleft}
    Type: fusion estimator (\(+\)) or RCT-only estimator (\(-\)); Mean: average of estimates; Bias: Monte-Carlo bias, \(10^{-4}\); RMSE: root mean squared error, \(10^{-2}\); SE: average of standard error estimates, \(10^{-2}\); Coverage: \(95\%\) confidence interval coverage, \(\%\); Reduction: average of percentage reduction in squared standard error estimates, \(\%\).
  \end{flushleft}
\end{table}

\begin{table}
  \caption{Simulation results for restricted mean times lost \(\gamma_1(1,t)\).}
  \label{tab:sim-rmtl-app-1}
  \footnotesize
  \centering
  
\begin{tabular}{rlrrrrr}
{\(n\)} & {\(t\)} & {Mean} & {Bias} & {RMSE} & {SE} & {Coverage}\\
\(750\) & \(0.25\) & \(0.02\) & \(-0.55\) & \(0.44\) & \(0.43\) & \(94.3\)\\
 & \(1\) & \(0.15\) & \(-4.06\) & \(2.45\) & \(2.48\) & \(94.3\)\\
 & \(2\) & \(0.40\) & \(-15.55\) & \(5.65\) & \(5.75\) & \(94.5\)\\
\(1500\) & \(0.25\) & \(0.02\) & \(-0.70\) & \(0.31\) & \(0.31\) & \(93.5\)\\
 & \(1\) & \(0.15\) & \(-7.23\) & \(1.84\) & \(1.76\) & \(94.1\)\\
 & \(2\) & \(0.40\) & \(-27.24\) & \(4.32\) & \(4.08\) & \(92.4\)\\
\end{tabular}

  \medskip
  \begin{flushleft}
    Mean: average of estimates; Bias: Monte-Carlo bias, \(10^{-4}\); RMSE: root mean squared error, \(10^{-2}\); SE: average of standard error estimates, \(10^{-2}\); Coverage: \(95\%\) confidence interval coverage, \(\%\).
  \end{flushleft}
\end{table}

\begin{table}
  \caption{Simulation results for restricted mean times lost \(\gamma_2(1,t)\), \(\gamma_2(0,t)\), and \(\gamma_2(t)\).}
  \label{tab:sim-rmtl-app-2}
  \footnotesize
  \centering
  
\begin{tabular}{rlrlrrrrrr}
{\(n\)} & {Estimand} & {\(t\)} & {Type} & {Mean} & {Bias} & {RMSE} & {SE} & {Coverage} & {Reduction}\\
\(750\) & \(\gamma_2(0,t)\) & \(0.25\) & \(+\) & \(0.04\) & \(-1.11\) & \(0.58\) & \(0.57\) & \(94.6\) & \(0.50\)\\
 &  &  & \(-\) & \(0.04\) & \(-1.32\) & \(0.59\) & \(0.57\) & \(94.2\) & \(\cdot\)\\
 &  & \(1\) & \(+\) & \(0.30\) & \(-21.05\) & \(2.98\) & \(2.96\) & \(94.9\) & \(2.85\)\\
 &  &  & \(-\) & \(0.30\) & \(-23.22\) & \(3.04\) & \(3.00\) & \(94.7\) & \(\cdot\)\\
 &  & \(2\) & \(+\) & \(0.79\) & \(-75.58\) & \(6.32\) & \(6.32\) & \(94.7\) & \(6.21\)\\
 &  &  & \(-\) & \(0.79\) & \(-80.45\) & \(6.54\) & \(6.53\) & \(94.2\) & \(\cdot\)\\
 & \(\gamma_2(1,t)\) & \(0.25\) & \(-\) & \(0.04\) & \(-2.76\) & \(0.59\) & \(0.57\) & \(93.5\) & \(\cdot\)\\
 &  & \(1\) & \(-\) & \(0.29\) & \(-33.89\) & \(3.13\) & \(3.01\) & \(93.8\) & \(\cdot\)\\
 &  & \(2\) & \(-\) & \(0.76\) & \(-94.02\) & \(6.93\) & \(6.58\) & \(92.9\) & \(\cdot\)\\
 & \(\gamma_2(t)\) & \(0.25\) & \(+\) & \(0.00\) & \(-1.65\) & \(0.81\) & \(0.80\) & \(94.1\) & \(0.26\)\\
 &  &  & \(-\) & \(0.00\) & \(-1.44\) & \(0.81\) & \(0.80\) & \(94.3\) & \(\cdot\)\\
 &  & \(1\) & \(+\) & \(-0.01\) & \(-12.84\) & \(4.16\) & \(4.09\) & \(93.7\) & \(1.53\)\\
 &  &  & \(-\) & \(-0.00\) & \(-10.66\) & \(4.21\) & \(4.12\) & \(94.0\) & \(\cdot\)\\
 &  & \(2\) & \(+\) & \(-0.03\) & \(-18.44\) & \(8.87\) & \(8.82\) & \(94.7\) & \(3.31\)\\
 &  &  & \(-\) & \(-0.03\) & \(-13.57\) & \(9.09\) & \(8.97\) & \(94.4\) & \(\cdot\)\\
\(1500\) & \(\gamma_2(0,t)\) & \(0.25\) & \(+\) & \(0.04\) & \(-0.98\) & \(0.40\) & \(0.40\) & \(94.1\) & \(0.48\)\\
 &  &  & \(-\) & \(0.04\) & \(-1.06\) & \(0.41\) & \(0.41\) & \(94.2\) & \(\cdot\)\\
 &  & \(1\) & \(+\) & \(0.30\) & \(-12.13\) & \(2.09\) & \(2.10\) & \(94.8\) & \(2.80\)\\
 &  &  & \(-\) & \(0.30\) & \(-12.59\) & \(2.12\) & \(2.13\) & \(94.3\) & \(\cdot\)\\
 &  & \(2\) & \(+\) & \(0.79\) & \(-36.12\) & \(4.53\) & \(4.50\) & \(94.0\) & \(6.18\)\\
 &  &  & \(-\) & \(0.79\) & \(-36.57\) & \(4.65\) & \(4.65\) & \(94.3\) & \(\cdot\)\\
 & \(\gamma_2(1,t)\) & \(0.25\) & \(-\) & \(0.04\) & \(-1.95\) & \(0.41\) & \(0.41\) & \(93.7\) & \(\cdot\)\\
 &  & \(1\) & \(-\) & \(0.29\) & \(-19.07\) & \(2.14\) & \(2.14\) & \(94.7\) & \(\cdot\)\\
 &  & \(2\) & \(-\) & \(0.76\) & \(-52.09\) & \(4.72\) & \(4.68\) & \(94.5\) & \(\cdot\)\\
 & \(\gamma_2(t)\) & \(0.25\) & \(+\) & \(0.00\) & \(-0.96\) & \(0.56\) & \(0.56\) & \(94.7\) & \(0.25\)\\
 &  &  & \(-\) & \(0.00\) & \(-0.89\) & \(0.56\) & \(0.57\) & \(94.7\) & \(\cdot\)\\
 &  & \(1\) & \(+\) & \(-0.00\) & \(-6.94\) & \(2.88\) & \(2.91\) & \(95.9\) & \(1.49\)\\
 &  &  & \(-\) & \(-0.00\) & \(-6.47\) & \(2.89\) & \(2.93\) & \(95.9\) & \(\cdot\)\\
 &  & \(2\) & \(+\) & \(-0.03\) & \(-15.97\) & \(6.17\) & \(6.28\) & \(95.8\) & \(3.28\)\\
 &  &  & \(-\) & \(-0.03\) & \(-15.51\) & \(6.25\) & \(6.38\) & \(96.0\) & \(\cdot\)\\
\end{tabular}

  \medskip
  \begin{flushleft}
    Type: fusion estimator (\(+\)) or RCT-only estimator (\(-\)); Mean: average of estimates; Bias: Monte-Carlo bias, \(10^{-4}\); RMSE: root mean squared error, \(10^{-2}\); SE: average of standard error estimates, \(10^{-2}\); Coverage: \(95\%\) confidence interval coverage, \(\%\); Reduction: average of percentage reduction in squared standard error estimates, \(\%\).
  \end{flushleft}
\end{table}

\section{Details on the real data example}
\label{sec:data-example-app}

In both RCTs used in the data example, the rate of severe adverse events was relatively low, and the vast majority of participants were censored by the end of the study.
Considering the three-point major adverse cardiovascular event (MACE, a composite event of cardiovascular death, non-fatal myocardial infarction, and non-fatal stroke) as the primary event, only \(12\) out of \(1649\) subjects randomized to placebo in SUSTAIN-6 experienced the competing non-cardiovascular death event.
The crude hazard estimate of MACE is \(0.37\) events per \(100\) person-years, where \(1\) year counts as \(365.25\) days.
The number of non-cardiovascular deaths was \(133\) out of \(4672\) for the placebo group in LEADER, with a corresponding hazard of \(0.79\) events per \(100\) person-year.
While MACE was the primary outcome in the original analyses of both studies, we turned to the composite event of non-fatal myocardial infarction and non-fatal stroke as the event of interest.
The main reason is precisely that a greater number of competing events would allow us to better illustrate our method.

The SUSTAIN-6 trial followed participants for a maximum of \(104\) weeks since randomization with an end-of-trial visit at week \(109\).
The LEADER trial, on the other hand, planned a much longer follow-up period of up to \(54\) months.
Therefore, without assuming transportability of the conditional cause-specific hazards of both non-fatal cardiovascular outcome and all-cause death, none of the parameters considered would be identifiable beyond week \(104\).
In the data example, we chose to estimate parameters at \(4\) evenly spaced time points up to week \(104\).

For transportability of the cause-specific hazard of the composite event under placebo, we needed to control for the baseline covariates that are shifted prognostic variables between the RCT population and the external control population.
In the data example, we employed the list of baseline characteristics in Table~1 of \citetsuppmat{marso2016semaglutide}.
All cause-specific hazards were fitted by the Cox proportional hazards model with a linear combination of the baseline covariates as the logarithm of the multiplicative risk.
The concentration of low-density lipoprotein cholesterol was measured in \(\text{mmol}\cdot \text{l}^{-1}\) and subsequently log-transformed to reduce skewness.
History of hemorrhagic stroke was removed from the list, as its presence caused extreme numeric instability during the fitting of the Cox model.
Patients with missing baseline covariates were removed from the data.
Exact numbers of missing subjects per treatment groups are displayed in Table~\ref{tab:missing-app}.

The data included \(11\) tied event times between times to non-fatal cardiovascular event in the joint placebo arm and times to all-cause death in the placebo arm of SUSTAIN-6.
Since our estimator is presented under Assumption~\ref{asn:discontinuity}, we broke the ties by jittering the observed event times.
Specifically, a random sample of noise was drawn from the uniform distribution between \(0\) and \(10^{-5}\) and then added to the observed event times.
The event times used in the analysis were recorded in days as integers.
Therefore, such small perturbations should not have meaningful consequences for the results.

To stabilize the estimators, we set a threshold for inverse weights inside the integrals \(\hat\ell_j(a)(O)\).
For sample size \(n\), the inverse weights above the value \(n^{1/2}\log(n)/5\) were set to that value.

\begin{table}
  \caption{Subjects with missing baseline covariates in SUSTAIN-6 and LEADER.}
  \label{tab:missing-app}
  \footnotesize
  \centering
  \begin{tabular}{lrrr}
    & \multicolumn{2}{c}{SUSTAIN-6} & LEADER \\
    & Semaglutide \(1.0\) mg & Placebo & Placebo \\
    \(N\) & \(822\) & \(1649\) & \(4672\) \\
    Missing (\(N\)) & \(9\) & \(24\) & \(96\) \\
    Missing (\%) & \(1.09\) & \(1.46\) & \(2.05\)
  \end{tabular}
\end{table}

\begin{table}
  \caption{Treatment-specific cumulative incidences in the real data example.}
  \label{tab:analysis-cif-app}
  \footnotesize
  \centering
  \begin{tabular}{lrlrrr}
    Estimand & \(t\) (weeks) & Type & Estimate (\%) & 95\%-CI (\%) & Reduction\\
    \(\theta_1(0,t)\) & 26 & \(+\) & \(1.84\) & \((1.31,2.36)\) & \(23.17\)\\
             &  & \(-\) & \(1.97\) & \((1.29,2.65)\) & .\\
             & 52 & \(+\) & \(3.10\) & \((2.43,3.77)\) & \(23.20\)\\
             &  & \(-\) & \(3.21\) & \((2.34,4.09)\) & .\\
             & 78 & \(+\) & \(4.87\) & \((4.04,5.70)\) & \(21.30\)\\
             &  & \(-\) & \(4.66\) & \((3.60,5.72)\) & .\\
             & 104 & \(+\) & \(6.42\) & \((5.46,7.38)\) & \(21.78\)\\
             &  & \(-\) & \(6.25\) & \((5.02,7.48)\) & .\\
    \(\theta_2(0,t)\) & 26 & \(+\) & \(0.74\) & \((0.32,1.17)\) & \(-0.00\)\\
             &  & \(-\) & \(0.74\) & \((0.32,1.17)\) & .\\
             & 52 & \(+\) & \(1.23\) & \((0.69,1.78)\) & \(-0.00\)\\
             &  & \(-\) & \(1.23\) & \((0.69,1.78)\) & .\\
             & 78 & \(+\) & \(1.85\) & \((1.19,2.52)\) & \(-0.00\)\\
             &  & \(-\) & \(1.85\) & \((1.19,2.52)\) & .\\
             & 104 & \(+\) & \(2.92\) & \((2.08,3.76)\) & \(-0.00\)\\
             &  & \(-\) & \(2.92\) & \((2.08,3.76)\) & .\\
    \(\theta_1(1,t)\) & 26 & . & \(1.58\) & \((0.72,2.44)\) & .\\
             & 52 & . & \(2.49\) & \((1.42,3.55)\) & .\\
             & 78 & . & \(2.88\) & \((1.73,4.03)\) & .\\
             & 104 & . & \(3.69\) & \((2.40,4.98)\) & .\\
    \(\theta_2(1,t)\) & 26 & . & \(0.25\) & \((-0.12,0.62)\) & .\\
             & 52 & . & \(0.85\) & \((0.17,1.54)\) & .\\
             & 78 & . & \(1.90\) & \((0.88,2.93)\) & .\\
             & 104 & . & \(2.71\) & \((1.49,3.93)\) & .\\
  \end{tabular}
  \medskip
  \begin{flushleft}
    {Type: fusion estimator (\(+\)) or RCT-only estimator (\(-\)); CI: confidence interval; Reduction: percentage reduction CI length, \(\%\).}
  \end{flushleft}
\end{table}

\begin{table}
  \caption{Treatment-specific restricted mean times lost in the real data example.}
  \label{tab:analysis-rmtl-app}
  \footnotesize
  \centering
  \begin{tabular}{lrlrrr}
    Estimand & \(t\) (weeks) & Type & Estimate (weeks) & 95\%-CI (weeks) & Reduction\\
\(\gamma_1(0,t)\) & 26 & \(+\) & \(0.27\) & \((0.18,0.35)\) & \(20.83\)\\
 &  & \(-\) & \(0.27\) & \((0.17,0.38)\) & .\\
 & 52 & \(+\) & \(0.94\) & \((0.71,1.17)\) & \(22.17\)\\
 &  & \(-\) & \(0.98\) & \((0.68,1.27)\) & .\\
 & 78 & \(+\) & \(2.03\) & \((1.63,2.43)\) & \(22.43\)\\
 &  & \(-\) & \(2.05\) & \((1.54,2.57)\) & .\\
 & 104 & \(+\) & \(3.50\) & \((2.90,4.11)\) & \(22.35\)\\
 &  & \(-\) & \(3.49\) & \((2.71,4.27)\) & .\\
\(\gamma_2(0,t)\) & 26 & \(+\) & \(0.08\) & \((0.03,0.14)\) & \(-0.00\)\\
 &  & \(-\) & \(0.08\) & \((0.03,0.14)\) & .\\
 & 52 & \(+\) & \(0.33\) & \((0.16,0.49)\) & \(-0.00\)\\
 &  & \(-\) & \(0.33\) & \((0.16,0.49)\) & .\\
 & 78 & \(+\) & \(0.73\) & \((0.42,1.03)\) & \(-0.00\)\\
 &  & \(-\) & \(0.73\) & \((0.42,1.03)\) & .\\
 & 104 & \(+\) & \(1.36\) & \((0.90,1.83)\) & \(-0.00\)\\
 &  & \(-\) & \(1.36\) & \((0.90,1.83)\) & .\\
\(\gamma_1(1,t)\) & 26 & . & \(0.16\) & \((0.06,0.26)\) & .\\
 & 52 & . & \(0.72\) & \((0.39,1.05)\) & .\\
 & 78 & . & \(1.44\) & \((0.83,2.04)\) & .\\
 & 104 & . & \(2.35\) & \((1.45,3.25)\) & .\\
\(\gamma_2(1,t)\) & 26 & . & \(0.03\) & \((-0.03,0.10)\) & .\\
 & 52 & . & \(0.21\) & \((0.02,0.39)\) & .\\
 & 78 & . & \(0.52\) & \((0.15,0.89)\) & .\\
 & 104 & . & \(1.12\) & \((0.51,1.73)\) & .\\
  \end{tabular}

  \medskip
  \begin{flushleft}
    {Type: fusion estimator (\(+\)) or RCT-only estimator (\(-\)); CI: confidence interval; Reduction: percentage reduction CI length, \(\%\).}
  \end{flushleft}
\end{table}



\begin{table}
  \caption{Cumulative incidence differences after removing history of cardiovascular diseases from the baseline variables.}
  \label{tab:analysis-cif-sensitivity-no-history-app}
  \footnotesize
  \centering
  \begin{tabular}{lrlrrr}
    Estimand & \(t\) (weeks) & Type & Estimate (\%) & 95\%-CI (\%) & Reduction\\
    \(\theta_1(t)\) & 26 & \(+\) & \(-0.00\) & \((-0.88,0.87)\) & \(18.11\)\\
 &  & \(-\) & \(-0.56\) & \((-1.63,0.51)\) & .\\
 & 52 & \(+\) & \(-0.66\) & \((-1.79,0.47)\) & \(16.79\)\\
 &  & \(-\) & \(-0.96\) & \((-2.32,0.39)\) & .\\
 & 78 & \(+\) & \(-1.85\) & \((-3.10,-0.60)\) & \(18.39\)\\
 &  & \(-\) & \(-2.03\) & \((-3.56,-0.49)\) & .\\
 & 104 & \(+\) & \(-2.44\) & \((-3.86,-1.03)\) & \(19.02\)\\
 &  & \(-\) & \(-2.88\) & \((-4.62,-1.14)\) & .\\
\(\theta_2(t)\) & 26 & \(+\) & \(-0.53\) & \((-1.06,0.00)\) & \(0.00\)\\
 &  & \(-\) & \(-0.52\) & \((-1.06,0.01)\) & .\\
 & 52 & \(+\) & \(-0.40\) & \((-1.25,0.45)\) & \(-0.00\)\\
 &  & \(-\) & \(-0.40\) & \((-1.25,0.45)\) & .\\
 & 78 & \(+\) & \(0.07\) & \((-1.13,1.26)\) & \(-0.00\)\\
 &  & \(-\) & \(0.07\) & \((-1.13,1.27)\) & .\\
 & 104 & \(+\) & \(-0.18\) & \((-1.64,1.28)\) & \(-0.00\)\\
 &  & \(-\) & \(-0.17\) & \((-1.63,1.29)\) & .\\
\end{tabular}
  \medskip
  \begin{flushleft}
    {Type: fusion estimator (\(+\)) or RCT-only estimator (\(-\)); CI: confidence interval; Reduction: percentage reduction CI length, \(\%\).}
  \end{flushleft}
\end{table}

\begin{table}
  \caption{Restricted mean time lost differences after removing history of cardiovascular diseases from the baseline variables.}
  \label{tab:analysis-rmtl-sensitivity-no-history-app}
  \footnotesize
  \centering
  \begin{tabular}{lrlrrr}
    Estimand & \(t\) (weeks) & Type & Estimate (weeks) & 95\%-CI (weeks) & Reduction\\
    \(\gamma_1(t)\) & 26 & \(+\) & \(-0.05\) & \((-0.16,0.06)\) & \(24.42\)\\
 &  & \(-\) & \(-0.12\) & \((-0.27,0.02)\) & .\\
 & 52 & \(+\) & \(-0.13\) & \((-0.48,0.22)\) & \(19.66\)\\
 &  & \(-\) & \(-0.33\) & \((-0.76,0.10)\) & .\\
 & 78 & \(+\) & \(-0.48\) & \((-1.11,0.16)\) & \(18.37\)\\
 &  & \(-\) & \(-0.75\) & \((-1.53,0.03)\) & .\\
 & 104 & \(+\) & \(-0.97\) & \((-1.93,-0.02)\) & \(18.10\)\\
 &  & \(-\) & \(-1.36\) & \((-2.53,-0.19)\) & .\\
\(\gamma_2(t)\) & 26 & \(+\) & \(-0.05\) & \((-0.13,0.02)\) & \(0.00\)\\
 &  & \(-\) & \(-0.05\) & \((-0.13,0.02)\) & .\\
 & 52 & \(+\) & \(-0.13\) & \((-0.37,0.11)\) & \(-0.00\)\\
 &  & \(-\) & \(-0.13\) & \((-0.37,0.11)\) & .\\
 & 78 & \(+\) & \(-0.22\) & \((-0.67,0.24)\) & \(-0.00\)\\
 &  & \(-\) & \(-0.21\) & \((-0.67,0.25)\) & .\\
 & 104 & \(+\) & \(-0.24\) & \((-0.99,0.51)\) & \(-0.00\)\\
 &  & \(-\) & \(-0.24\) & \((-0.99,0.51)\) & .\\
\end{tabular}
  \medskip
  \begin{flushleft}
    {Type: fusion estimator (\(+\)) or RCT-only estimator (\(-\)); CI: confidence interval; Reduction: percentage reduction CI length, \(\%\).}
  \end{flushleft}
\end{table}

\begin{table}
  \caption{Cumulative incidence differences after removing controls from SUSTAIN-6.}
  \label{tab:analysis-cif-fewer-controls-app}
  \footnotesize
  \centering
  \begin{tabular}{lrlrrr}
    Estimand & \(t\) (weeks) & Type & Estimate (\%) & 95\%-CI (\%) & Reduction\\
\(\theta_1(t)\) & 26 & \(+\) & \(-0.21\) & \((-1.16,0.73)\) & \(48.70\)\\
 &  & \(-\) & \(-1.22\) & \((-3.07,0.62)\) & .\\
 & 52 & \(+\) & \(-0.97\) & \((-2.20,0.25)\) & \(45.48\)\\
 &  & \(-\) & \(-1.74\) & \((-3.99,0.50)\) & .\\
 & 78 & \(+\) & \(-2.35\) & \((-3.73,-0.96)\) & \(46.55\)\\
 &  & \(-\) & \(-2.74\) & \((-5.32,-0.15)\) & .\\
 & 104 & \(+\) & \(-3.10\) & \((-4.67,-1.52)\) & \(48.35\)\\
 &  & \(-\) & \(-4.16\) & \((-7.21,-1.12)\) & .\\
\(\theta_2(t)\) & 26 & \(+\) & \(-1.03\) & \((-2.22,0.15)\) & \(0.01\)\\
 &  & \(-\) & \(-1.03\) & \((-2.21,0.16)\) & .\\
 & 52 & \(+\) & \(-1.17\) & \((-2.75,0.41)\) & \(0.01\)\\
 &  & \(-\) & \(-1.16\) & \((-2.74,0.42)\) & .\\
 & 78 & \(+\) & \(-1.41\) & \((-3.52,0.70)\) & \(0.01\)\\
 &  & \(-\) & \(-1.39\) & \((-3.50,0.72)\) & .\\
 & 104 & \(+\) & \(-1.54\) & \((-3.99,0.92)\) & \(0.01\)\\
 &  & \(-\) & \(-1.51\) & \((-3.97,0.95)\) & .\\
\end{tabular}
  \medskip
  \begin{flushleft}
    {Type: fusion estimator (\(+\)) or RCT-only estimator (\(-\)); CI: confidence interval; Reduction: percentage reduction CI length, \(\%\).}
  \end{flushleft}
\end{table}

\begin{table}
  \caption{Restricted mean time lost differences after removing controls from SUSTAIN-6.}
  \label{tab:analysis-rmtl-fewer-controls-app}
  \footnotesize
  \centering
  \begin{tabular}{lrlrrr}
    Estimand & \(t\) (weeks) & Type & Estimate (weeks) & 95\%-CI (weeks) & Reduction\\
\(\gamma_1(t)\) & 26 & \(+\) & \(-0.10\) & \((-0.23,0.03)\) & \(59.69\)\\
 &  & \(-\) & \(-0.29\) & \((-0.60,0.03)\) & .\\
 & 52 & \(+\) & \(-0.27\) & \((-0.65,0.12)\) & \(51.75\)\\
 &  & \(-\) & \(-0.64\) & \((-1.44,0.15)\) & .\\
 & 78 & \(+\) & \(-0.71\) & \((-1.41,-0.01)\) & \(49.08\)\\
 &  & \(-\) & \(-1.28\) & \((-2.65,0.09)\) & .\\
 & 104 & \(+\) & \(-1.36\) & \((-2.41,-0.31)\) & \(47.99\)\\
 &  & \(-\) & \(-2.10\) & \((-4.11,-0.08)\) & .\\
\(\gamma_2(t)\) & 26 & \(+\) & \(-0.16\) & \((-0.34,0.02)\) & \(0.01\)\\
 &  & \(-\) & \(-0.16\) & \((-0.34,0.02)\) & .\\
 & 52 & \(+\) & \(-0.41\) & \((-0.93,0.10)\) & \(0.01\)\\
 &  & \(-\) & \(-0.41\) & \((-0.92,0.10)\) & .\\
 & 78 & \(+\) & \(-0.78\) & \((-1.70,0.14)\) & \(0.01\)\\
 &  & \(-\) & \(-0.77\) & \((-1.69,0.15)\) & .\\
 & 104 & \(+\) & \(-1.24\) & \((-2.68,0.19)\) & \(0.01\)\\
 &  & \(-\) & \(-1.23\) & \((-2.66,0.20)\) & .\\
\end{tabular}
  \medskip
  \begin{flushleft}
    {Type: fusion estimator (\(+\)) or RCT-only estimator (\(-\)); CI: confidence interval; Reduction: percentage reduction CI length, \(\%\).}
  \end{flushleft}
\end{table}

\section{Implications of weaker transportability assumptions}
\label{sec:weaker-app}

In the main text, we have showcased how the transportability of a conditional cause-specific hazard improves the precision of estimators for cumulative incidence functions and restricted mean times lost.

One consideration is whether this assumption can be reasonably weakened according to the parameter of interest.
To ground ideas, consider the target parameter \(\theta_1(0)=\E\{F_{11}(\tau\mid 0,X)\mid D=1\}\).
If we view the parameter as the mean of a binary outcome \(\E\{\I\{T(0)\leq \tau,J(0)=1\}\mid D=1\}\), a straightforward transportability assumption would be
\[
  \Pr\{T(0)\leq \tau,J(0)=1\mid X=x,D=1\} = \Pr\{T(0)\leq \tau,J(0)=1\mid X=x,D=0\},
\]
which is
\[
  \int_{0}^{\tau}S_1(0)(t\!-\!\mid x)\d\Alpha_{11}(0)(t\mid x)\d t=\int_{0}^{\tau}S_0(0)(t\!-\!\mid x)\d\Alpha_{01}(0)(t\mid x)\d t
\]
for \(x\in\mathcal{X}_1\cap\mathcal{X}_0\), where \(S_{d}(0)(t\mid x)=[\Pi\{\Alpha_{d1}(0)+\Alpha_{d2}(0)\}](t\mid x)\).

There are two peculiarities to point out.
The first is whether it makes sense at all to only restrict the value of conditional cumulative incidence function of cause 1 at the time point \(\tau\).
It is very unnatural to only assume transportability for a single time point.
If this assumption holds, we should also expect the cumulative incidence functions in the time interval around that time point to be quite comparable across populations, especially when the event time distribution is continuous.
Moreover, we would not generally expect a substantial decrease in the semiparametric efficiency bound of the parameter, if the compatibility of the two population exists for a mere single time point on a specific scale defined by the parameter.

The second is a result of the cumulative incidence function of cause 1 being a functional of the cause-specific hazards of both event types.
Therefore, by making this assumption, we are also putting restrictions on the cause 2 hazards between the two populations.
However, reasoning for comparability of the cumulative incidence functions is arguably more difficult than doing so separately for the two event rates.
Note that this observation also applies to transportability assumptions on subdistribution hazards \citepsuppmat{fine1999proportional}, for example, for \(t\in(0,\tau]\),
\[
  \Pr\{T(0)\leq t,J(0)=1\mid X=x,D=1\} = \Pr\{T(0)\leq t,J(0)=1\mid X=x,D=0\}.
\]

Apart from the transportability of the cumulative incidence function, we may also consider the transportability of the all-cause survival function that \(S_1(0)(t\mid x)=S_0(0)(t\mid x)\).
However, that the sum of two cause-specific hazards is equal across the populations can result from many combinations of event rates whose interpretations are drastically different.
For instance, this assumption holds if the cause 1 hazard under placebo in the RCT population equals the cause 2 hazard in the external control population, while their competing risks are completely eliminated.
Since the estimands used in competing risks analysis often seek to separate the treatment effects on different causes, a transportability assumption that does not acknowledge the nature of competing risks may be hard to justify.


\bibliographysuppmat{./competing-risk-external-control.bib}


\end{document}